\documentclass{amsart}

\usepackage{amssymb,gensymb}
\usepackage[a4paper,left=2.25cm,right=2.25cm,top=3cm,bottom=3cm]{geometry}
\usepackage{tikz}\usetikzlibrary{decorations.pathmorphing,decorations.markings}
\usepackage{hyperref}
\usepackage{enumerate,enumitem}
\usepackage{mathdots}
\usepackage{graphicx}

\tikzset{bgplaq/.style={fill=lightgray!30!white}}
\tikzset{bggplaq/.style={fill=lightgray!70!white}}
\tikzset{bline/.style={line width=0.8mm}}
\tikzset{fline/.style={}}
\tikzset{gbline/.style={line width=0.8mm,draw=black!35!green}}
\tikzset{gline/.style={line width=0.3mm,draw=black!35!green}}
\tikzset{bbline/.style={line width=0.8mm,draw=black!35!blue}}
\tikzset{bline/.style={line width=0.3mm,draw=black!35!blue}}
\tikzset{wbline/.style={line width=0.6mm,draw=lightgray!10!white}}
\tikzset{wline/.style={draw=lightgray!10!white}}
\tikzset{part/.style={thick,solid,circle,fill=black,draw=black,inner sep=1.8pt,outer sep=3pt}}
\tikzset{hole/.style={thick,solid,circle,draw=black,fill=white,inner sep=1.8pt,outer sep=3pt}}
\tikzset{arrow/.style={postaction={decorate,thick,decoration={markings,mark = at position #1 with {\arrow{>}}}}},arrow/.default=0.5}
\def\cellbb{\useasboundingbox (-0.5,-0.5) rectangle ++(1,1); \draw[bgplaq] (-0.5,-0.5) rectangle ++(1,1);}
\tikzset{cross/.style={postaction={decorate,thick,decoration={markings,mark = at position #1 with{\draw (-2pt,-2pt) -- (2pt,2pt);\draw (2pt,-2pt) -- (-2pt,2pt);}}}}}
\tikzset{distort/.style={cm={1,0,-\slt,\sltb,(0,0)}}}
\def\slt{0.2}
\pgfmathsetmacro{\sltb}{sqrt(1-\slt*\slt)}

\newcommand\light[2]{\draw[fline,bgplaq] (#1,#2) -- (#1+1,#2) -- (#1+1,#2+1) -- (#1,#2+1) -- (#1,#2);}
\newcommand\dark[2]{\draw[fline,bggplaq] (#1,#2) -- (#1+1,#2) -- (#1+1,#2+1) -- (#1,#2+1) -- (#1,#2);}
\newcommand\bdark[2]{\draw[fline,bggplaq] (#1,#2) -- (#1+0.5,#2) -- (#1+0.5,#2+1) -- (#1,#2+1) -- (#1,#2);}
\newcommand\blight[2]{\draw[fline,bgplaq] (#1,#2) -- (#1+0.5,#2) -- (#1+0.5,#2+1) -- (#1,#2+1) -- (#1,#2);}

\newcommand\bbull[3]{\filldraw[fill=black!35!blue, draw=black] (#1,#2) circle (#3cm);}
\newcommand\gbull[3]{\filldraw[fill=black!35!green, draw=black] (#1,#2) circle (#3cm);}
\newcommand\ebull[3]{\draw[fill=white,draw=black] (#1,#2) circle (#3cm);}

\def\part{\begin{tikzpicture}[scale=0.5,baseline=-1.7] \gbull{0}{0}{0.12}; \end{tikzpicture}}
\def\hole{\begin{tikzpicture}[scale=0.5,baseline=-1.7] \ebull{0}{0}{0.12}; \end{tikzpicture}}
\def\bpart{\begin{tikzpicture}[scale=0.5,baseline=-1.7] \bbull{0}{0}{0.12}; \end{tikzpicture}}

\newcommand{\ddiagdot}{\mathbin{\rotatebox[origin=c]{45}{$\cdots$}}}
\newcommand{\udiagdot}{\mathbin{\rotatebox[origin=c]{-45}{$\cdots$}}}

\allowdisplaybreaks[1]

\newcommand{\bra}[1]{\left\langle #1\right|}
\newcommand{\ket}[1]{\left|#1\right\rangle}

\newcommand{\pf}{\mathop{\rm Pf}}
\renewcommand{\b}[1]{\bar{#1}}
\renewcommand{\leq}{\leqslant}
\renewcommand{\geq}{\geqslant}
\renewcommand\ss{\scriptstyle}
\newcommand\sss{\scriptscriptstyle}

\renewcommand\succeq{\succcurlyeq}

\def\fs{\footnotesize}
\def\phid{\phi^\dagger}

\def\red{\color{dred}}
\definecolor{dred}{RGB}{176,0,0}
\def\blue{\color{black!35!blue}}

\newtheorem{defn}{Definition}
\newtheorem{lem}{Lemma}
\newtheorem{thm}{Theorem}

\newtheorem{cor}{Corollary}
\newtheorem{rmk}{Remark}

\title[Refined Cauchy/Littlewood identities and six-vertex model partition functions III]
{Refined Cauchy/Littlewood identities and six-vertex model partition functions: III. Deformed bosons}

\author{M.~Wheeler}
\address{Michael Wheeler, School of Mathematics and Statistics, University of Melbourne, Parkville, Victoria 3010, Australia}
\email{mwheeler@ms.unimelb.edu.au}

\author{P.~Zinn-Justin}
\address{Paul Zinn-Justin, Laboratoire de Physique Th\'eorique et Hautes \'Energies, CNRS UMR 7589 and Universit\'e Pierre et Marie Curie (Paris 6), 4 place Jussieu, 75252 Paris cedex 05, France}
\email{pzinn@lpthe.jussieu.fr}

\keywords{Cauchy and Littlewood identities, symmetric functions, alternating sign matrices, six-vertex model}

\thanks{The authors are supported by ARC grant DP140102201 and ERC grant 278124 ``LIC''.
They would like to acknowledge hospitality and support from the Galileo Galilei Institute,
where part of this work was carried out during the program ``Statistical Mechanics, Integrability and Combinatorics".}

\begin{document}

\begin{abstract}
We study Hall--Littlewood polynomials using an integrable lattice model of $t$-deformed bosons. Working with row-to-row transfer matrices, we review the construction of Hall--Littlewood polynomials (of the $A_n$ root system) within the framework of this model. Introducing appropriate double-row transfer matrices, we extend this formalism to 
Hall--Littlewood polynomials based on the $BC_n$ root system, and obtain a new combinatorial formula for them. We then apply our methods to prove a series of refined Cauchy and Littlewood identities involving Hall--Littlewood polynomials. The last two of these identities are new, and relate infinite sums over hyperoctahedrally symmetric 
Hall--Littlewood polynomials with partition functions of the six-vertex model on finite domains.
\end{abstract}

\maketitle

\section{Introduction}

This paper is a continuation of the work initiated in \cite{bw,bwz-j}, where a number of identities involving 
Hall--Littlewood polynomials were conjectured, and in some cases proved. These identities generalize Cauchy and Littlewood formulae \cite{mac}, in which an infinite sum over Hall--Littlewood polynomials on the left hand side is equated with a finite product on the right hand side. The generalization is characterized by {\bf 1.} A small modification of the summand on the left hand side, often featuring a refining parameter; {\bf 2.} The right hand side being replaced with a partition function of the six-vertex model on some finite domain, rather than a simple product as in the original identity. In the current paper we give a method which allows us to prove all such identities in a unified way, and to obtain new ones.

The new feature in this paper, compared with our earlier ones \cite{bw,bwz-j}, is that we use an integrable model of deformed bosons as a tool for studying Hall--Littlewood polynomials. We refer to this model as the {\it $t$-boson model,} with $t$ playing the role of a deformation parameter\footnote{In the literature this model is normally called the {\it $q$-boson model,} but considering that the deformation parameter of the model is the same as the one-parameter deformation of Schur to Hall--Littlewood polynomials, we find that it makes more sense to name it the $t$-boson model. This also avoids confusion with the $q$ parameter from Macdonald polynomials.}. The $t$-boson model was introduced in \cite{bb} as a discretization of the Bose gas model, and studied in the framework of the algebraic Bethe Ansatz (see \cite{bik}, and references therein). Its relation to symmetric functions was discovered by Tsilevich in \cite{tsi}, where it was shown that the off-shell wavefunctions of the model are equal to Hall--Littlewood polynomials. The model has also been studied on the semi-infinite lattice, with a boundary term in its Hamiltonian, leading to connections with hyperoctahedrally symmetric Hall--Littlewood polynomials \cite{vde2,vde1}. Quite recently, Borodin has studied a further one-parameter deformation of the model and used it to define a family of rational symmetric functions which generalize Hall--Littlewood polynomials \cite{bor} (more precisely, the framework adopted in \cite{bor} is that of higher-spin vertex models \cite{man}, rather than bosonic lattice models). More recently, higher rank solutions of the $RLL$ equation (involving commuting copies of the $t$-boson algebra) have finally permitted Macdonald polynomials to be expressed in the setting of quantum integrable models \cite{cdgw,dgw}.

The $t$-boson model is a powerful device for proving Cauchy and Littlewood identities of the type studied in \cite{bw,bwz-j}. The main feature making it adaptable to this task is that its integrability comes from the $R$ matrix of the six-vertex model. This means that certain expressions involving Hall--Littlewood polynomials can be related, using only elementary steps within the Yang--Baxter algebra of the model, with partition functions of the six-vertex model. This straightforward approach suffices as a general technique to derive all the Cauchy and Littlewood identities in this paper, and the steps can be summarized as follows: 

{\bf 1.} Write down a partition function in the $t$-boson model, which can be identified with the left hand side of the identity that we wish to prove (an infinite sum over 
Hall--Littlewood polynomials). We refer to this as a {\it bosonic} partition function;

 {\bf 2.} Multiply the bosonic partition function (from its right edge) by a product of $R$ matrices of the six-vertex model, in such a way that it remains invariant. This invariance is easily deduced using freezing arguments. The vertices thus introduced will be called an {\it auxiliary lattice,} since each lattice line corresponds with an auxiliary vector space in the $t$-boson model. In performing this step, we obtain a non-trivial bosonic partition function which is attached to a trivial auxiliary lattice;
 
 {\bf 3.} Use the integrability of the $t$-boson model (the intertwining or $RLL$ equation) to relocate all $R$ matrices so that they multiply the bosonic partition function from its left edge. This repositioning of the auxiliary lattice has two consequences. First, it causes bosonic partition function to trivialize, so that it produces only an overall multiplicative factor. Second, the auxiliary lattice becomes a non-trivial partition function in the six-vertex model;
 
 {\bf 4.} Evaluate the resulting six-vertex model partition function explicitly as a determinant/Pfaffian (if possible). Three of the partition functions that we consider admit such an evaluation, but in our final two identities we obtain partition functions which do not seem to be expressible via a simple formula. In this way, we obtain the right hand side of the identity that we wish to prove. 

Our paper is split over seven sections. In Section \ref{sec:model} we define the $t$-boson model, identify its space of states with the space of partitions, and explain its integrability via the six-vertex model. We also discuss integrable boundary conditions (a solution of the reflection equation), which is necessary in our study of the hyperoctahedrally symmetric Hall--Littlewood polynomials. In Section \ref{sec:HL} we review the result of \cite{tsi}, expressing Hall--Littlewood polynomials as matrix elements in the $t$-boson model, and reinterpret this at the level of lattice paths. We shall also derive another known expression for Hall--Littlewood polynomials, as a sum over the symmetric group \cite{mac}, using the integrability of the $t$-boson model and working in the $F$ basis \cite{ms,abfr}. We then repeat these steps for Hall--Littlewood polynomials with hyperoctahedral symmetry, using the integrable boundaries of Section \ref{sec:model}. In this way we obtain a new realization of these polynomials in terms of lattice paths (which appears to be simpler than the recently obtained branching rule for Koornwinder polynomials \cite{vde3}), and recover a known sum formula originally due to Venkateswaran \cite{ven}.

The remainder of the paper deals with Cauchy and Littlewood identities for Hall--Littlewood polynomials. In Section \ref{sec:cauchy} we prove the refined Cauchy identity
\begin{multline}
\label{1}
\sum_{\lambda}
\prod_{i=1}^{m_0(\lambda)}
(1 - u t^{i})
b_{\lambda}(t)
P_{\lambda}(x_1,\dots,x_n;t)
P_{\lambda}(y_1,\dots,y_n;t)
=
\mathcal{Z}_{\rm DW}(x_1,\dots,x_n;y_1,\dots,y_n;t;u)
\\
=
\frac{\prod_{i,j=1}^{n} (1- t x_i y_j)}
{\prod_{1 \leq i<j \leq n} (x_i-x_j)(y_i-y_j)}
\det_{1 \leq i,j \leq n}
\left[
\frac{1-ut + (u-1)t x_i y_j}{(1-x_i y_j) (1-t x_i y_j)}
\right]
\end{multline}
where the sum is taken over all partitions of length\footnote{In this paper we shall consider parts of size zero as contributing to the length of a partition. For example, $\lambda = (5,4,2,0,0)$ shall have length $\ell(\lambda) = 5$. While this convention is non-standard, it allows us to place parts of size zero on a common footing with non-zero parts, which is very convenient for our purposes.} $\ell(\lambda) = n$, $P_{\lambda}(x_1,\dots,x_n;t)$ denotes a Hall--Littlewood polynomial, and
\begin{align*}
b_{\lambda}(t) = \prod_{j=1}^{\infty} \prod_{i=1}^{m_j(\lambda)} (1-t^i),
\quad
m_j(\lambda) = \#\{k: \lambda_k = j\},
\quad
m_0(\lambda) = \ell(\lambda) - \sum_{j=1}^{\infty} m_j(\lambda),
\end{align*}
using the method outlined above. This identity is the $q=0$ specialization of an analogous result for Macdonald polynomials, due to Warnaar \cite{war}. The $u=1$ specialization of \eqref{1} was already considered in \cite{bw}, and in \cite{bwz-j} an independent proof was obtained using the Izergin--Korepin technique. As an interesting by-product of the method adopted in the present paper, we obtain a realization of the right hand side of \eqref{1} as a ``hybrid'' partition function $\mathcal{Z}_{\rm DW}$ of the six-vertex model with an extra ``bosonic'' column. By setting $u=1$, this hybrid partition function reduces precisely to the domain wall partition function \cite{kor,ize}.

In Section \ref{sec:little} we prove the refined Littlewood identity
\begin{multline}
\label{2}
\sum_{\lambda' \ {\rm even}}
\prod_{i=1}^{m_0(\lambda)/2}
(1-u t^{2i-1})
b^{\rm ev}_{\lambda}(t)
P_{\lambda}(x_1,\dots,x_{2n};t)
=
\mathcal{Z}_{\rm OS}(x_1,\dots,x_{2n};t;u)
\\
=
\prod_{1 \leq i<j \leq 2n}
\frac{(1-t x_i x_j)}{(x_i - x_j)}
\pf_{1\leq i < j \leq 2n}
\left[
\frac{(x_i - x_j) (1-ut + (u-1)t x_i x_j)}
{(1-x_i x_j) (1-t x_i x_j)}
\right]
\end{multline}
where the sum is over all partitions with $\ell(\lambda)=2n$, whose conjugate $\lambda'$ has even parts (equivalently all part-multiplicities $m_j(\lambda)$ are even), and 
\begin{align*}
b^{\rm ev}_{\lambda}(t) = \prod_{j=1}^{\infty} \prod_{i=1}^{m_j(\lambda)/2} (1-t^{2i-1}),
\end{align*}
by adapting our general method to Littlewood-type sums. The $u=1$ case of \eqref{2} was conjectured in \cite{bw}, before it was proved for general $u$ in \cite{bwz-j}, again by Izergin--Korepin techniques. A $q$ analogue of \eqref{2} was also conjectured in \cite{bwz-j}, and proved very shortly afterward by Rains in \cite{rai2}. Similarly to in Section \ref{sec:cauchy}, we are able to interpret the right hand side of \eqref{2} as a hybrid partition function $\mathcal{Z}_{\rm OS}$ in the six-vertex model. When we specialize $u=1$, this hybrid partition function reduces to a partition function of off-diagonally symmetric alternating sign matrices \cite{kup2}.

In Section \ref{sec:refl-cauchy} we turn to Hall--Littlewood polynomials 
$K_{\lambda}(x_1^{\pm1},\dots,x_n^{\pm1};t)$ with hyperoctahedral symmetry\footnote{We will use the terms ``hyperoctahedral'' and ``$BC_n$'' Hall--Littlewood polynomials interchangeably throughout the paper.}, and prove the reflecting Cauchy identity
\begin{multline}
\label{3}
\sum_{\lambda} 
\prod_{i=1}^{m_0(\lambda)} 
(1-t^i)
b_{\lambda}(t)
K_{\lambda}(x_1^{\pm1},\dots,x_n^{\pm1};t) 
P_{\lambda}(y_1,\dots,y_n;t) 
=
\mathcal{Z}_{\rm U}(x_1^{\pm1},\dots,x_n^{\pm1};y_1,\dots,y_n;t)
\\
=
\frac{
\prod_{i,j=1}^{n} 
(1-t x_i y_j) (1-t \b{x}_i y_j) 
}
{
\prod_{1\leq i<j \leq n} (x_i-x_j) (y_i-y_j)(1 - \b{x}_i \b{x}_j)(1 - t y_i y_j)
}
\det_{1\leq i,j \leq n}
\left[
\frac{(1-t)}{(1- x_i y_j)(1-\b{x}_i y_j)(1-t x_i y_j)(1-t \b{x}_i y_j)}
\right]
\end{multline}
where the sum is over all partitions with $\ell(\lambda) = n$, and $\b{x} := 1/x$, which was conjectured in \cite{bw}. A $u$ refinement of \eqref{3}, featuring ``lifted'' Hall--Littlewood polynomials (lifted Koornwinder polynomials \cite{rai} at $q=0$), was subsequently conjectured in \cite{bwz-j}. Although we hoped this more general conjecture would be accessible by the techniques of this paper, we are unable to prove it for the moment, as we have not succeeded to write lifted Hall--Littlewood polynomials as matrix products in the $t$-boson model. Since the refining parameter is absent from \eqref{3}, we do not have an interpretation of the right hand side as a hybrid partition function in the sense of the previous two examples. As written, the right hand side is equal to the partition function $\mathcal{Z}_{\rm U}$ of the six-vertex model with a reflecting or U-turn boundary \cite{tsu}.

In Sections \ref{sec:doub-refl} and \ref{sec:refl-little} we present two new identities involving hyperoctahedrally symmetric Hall--Littlewood polynomials. The first is the doubly reflecting Cauchy identity
\begin{align}
\label{4}
\sum_{\lambda}
z^{|\lambda|} 
\prod_{i=1}^{m_0(\lambda)}
(1-t^i)
b_{\lambda}(t)
K_{\lambda}(x_1^{\pm1},\dots,x_n^{\pm1};t) 
K_{\lambda}(y_1^{\pm1},\dots,y_n^{\pm1};t) 
=
\mathcal{Z}_{\rm UU}(x_1^{\pm1},\dots,x_n^{\pm1};y_1^{\pm1},\dots,y_n^{\pm1};t;z),
\end{align}
summed over partitions with $\ell(\lambda) = n$, and the second is the reflecting Littlewood identity
\begin{align}
\label{5}
\sum_{\lambda'\ {\rm even}}
z^{|\lambda|/2} 
\prod_{i=1}^{m_0(\lambda)/2}
(1-t^{2i-1})
b^{\rm ev}_{\lambda}(t)
K_{\lambda}(x_1^{\pm1},\dots,x_{2n}^{\pm1};t)
=
\mathcal{Z}_{\rm UO}(x_1^{\pm1},\dots,x_{2n}^{\pm1};t;z).
\end{align}  
summed over partitions with $\ell(\lambda) = 2n$. These identities are distinct from the earlier ones in the paper in two ways. First, they contain another refining parameter $z$, which cannot be absorbed into the $x$ or $y$ variables due to the inhomogeneity of the 
hyperoctahedral Hall--Littlewood polynomials. This parameter serves to regularize the left hand side of \eqref{4} and \eqref{5}, since neither sum is well defined (even as a formal power series) at $z=1$. Second, the right hand sides of \eqref{4} and \eqref{5} do not have a simple evaluation in terms of determinants or Pfaffians, although they are very closely related to partition functions $\mathcal{Z}_{\rm UU}$ and $\mathcal{Z}_{\rm UO}$ corresponding with certain other symmetry classes of ASMs \cite{kup2}.

Finally, in Appendix \ref{app:fbasis} we review some basic facts about factorizing $F$ matrices \cite{ms,abfr}, while in Appendix \ref{app:lagrange} we recall elementary Lagrange interpolation techniques which allow us to evaluate the various partition functions appearing in the paper.

\section{Integrable model of deformed bosons}
\label{sec:model}

\subsection{Definition of physical space and partition states}
 
We consider a semi-infinite one-dimensional lattice, with sites labelled by non-negative integers. Each site $i \geq 0$ is occupied by $m_i$ particles, and these occupation numbers can take any non-negative integer value. We construct a vector space $V$ by taking linear combinations of all possible fillings of this one-dimensional lattice:
\begin{align}
\label{space}
V 
=
{\rm Span} 
\left\{
\ket{m_0}_0
\otimes
\ket{m_1}_1
\otimes
\ket{m_2}_2
\otimes
\cdots
\right\},
\qquad
m_i \geq 0,\ \forall\ i \geq 0.
\end{align}
The vector space $V$ will be the representation space for the operators which we subsequently consider. The basis elements $\bigotimes_{i=0}^{\infty} \ket{m_i}_i$ can be trivially identified with partitions. Let 
$\lambda$ be a partition with $m_i(\lambda)$ parts of size $i$, for all $i \geq 0$ (we allow partitions to have parts of size 0, and distinguish between partitions with different numbers of size 0 parts, even if they are otherwise identical). Then we define $\ket{\lambda} \in V$ as follows:
\begin{align*}
\ket{\lambda}
=
\ket{m_0(\lambda)}_0
\otimes
\ket{m_1(\lambda)}_1
\otimes
\ket{m_2(\lambda)}_2
\otimes
\cdots.
\end{align*}
This simple correspondence between partitions and states in $V$ can be interpreted as projecting a Young diagram onto the semi-infinite line. For example, we identify $\lambda = (5,3,3,1,0)$ with a filling in the following way:
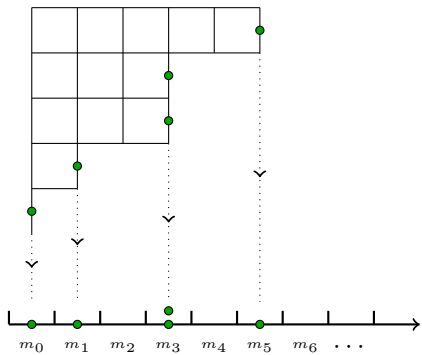
\begin{center}
\begin{tikzpicture}[scale=0.6]
%
\draw (0,4) -- (5,4);
\draw (0,3) -- (5,3);
\draw (0,2) -- (3,2);
\draw (0,1) -- (3,1);
\draw (0,0) -- (1,0);
\draw (0,-1) -- (0,4);
\draw (1,0) -- (1,4);
\draw (2,1) -- (2,4);
\draw (3,1) -- (3,4);
\draw (4,3) -- (4,4);
\draw (5,3) -- (5,4);
\gbull{0}{-0.5}{0.09};
\gbull{1}{0.5}{0.09};
\gbull{3}{1.5}{0.09};\gbull{3}{2.5}{0.09};
\gbull{5}{3.5}{0.09};
\draw[dotted,arrow=0.5] (0,-1) -- (0,-2.5);
\draw[dotted,arrow=0.5] (1,0) -- (1,-2.5);
\draw[dotted,arrow=0.5] (3,1) -- (3,-2.5);
\draw[dotted,arrow=0.5] (5,3) -- (5,-2.5);
\draw[thick,arrow=1] (-0.5,-3) -- (8.5,-3);
\foreach\x in {1,...,9}{
\draw[thick] (-1.5+\x,-3) -- (-1.5+\x,-2.7);
}
\gbull{0}{-3}{0.09};
\gbull{1}{-3}{0.09};
\gbull{3}{-3}{0.09};\gbull{3}{-2.7}{0.09};
\gbull{5}{-3}{0.09};
\node at (0,-3.5) {$\sss m_0$};
\node at (1,-3.5) {$\sss m_1$};
\node at (2,-3.5) {$\sss m_2$};
\node at (3,-3.5) {$\sss m_3$};
\node at (4,-3.5) {$\sss m_4$};
\node at (5,-3.5) {$\sss m_5$};
\node at (6,-3.5) {$\sss m_6$};
\node at (7,-3.5) {$\cdots$};
\end{tikzpicture}
\end{center}
An alternative way of projecting is to first rotate the Young diagram anti-clockwise by 135\degree\ (the ``Russian'' convention for drawing partitions) and then project onto the infinite integral lattice. This approach leads to a fermionic or ``Maya diagram'' representation of a partition, which will not be considered in this work.

\subsection{Inhomogeneous physical space}

A generalization of the vector space \eqref{space}, which allows non-integer fillings at each site of the lattice, will be important in this work. One can define a vector space $V(\alpha_0,\alpha_1,\dots)$ parametrized by complex variables $\alpha_i$ as follows:
\begin{align}
\label{alpha-space}
V(\alpha_0,\alpha_1,\alpha_2,\dots) 
=
{\rm Span} 
\left\{
\ket{m_0+ \alpha_0}_0
\otimes
\ket{m_1+ \alpha_1}_1 
\otimes
\ket{m_2+ \alpha_2}_2 
\otimes
\cdots
\right\},
\qquad
m_i \in \mathbb{Z},\ \alpha_i \in \mathbb{C},\ \forall\ i \geq 0.
\end{align}
The parameters $\alpha_i$ can be regarded as shifts or inhomogeneities, and we now allow negative occupation numbers $m_i$ at each site, meaning that the basis vectors of \eqref{alpha-space} are identified with generalized partitions which admit negative part-multiplicities $m_i(\lambda)$.

In theory it would be possible to keep all parameters $\alpha_i$ general, but for our purposes it is only necessary to preserve the first of these, by setting $\alpha_0 = \alpha$ and $\alpha_i = 0$ for all $ i \geq 1$.
This slight generalization of \eqref{space} will be useful in our study of refined Cauchy and Littlewood identities, with the refining parameter being directly related to $\alpha$. 

\subsection{$t$-boson algebra and Fock representation}

One of the principal tools in this work is the $t$-boson algebra, generated freely by the elements 
$\phi,\phid$ modulo the commutation relation
\begin{align}
\label{t-boson}
\phi \phid - t \phid \phi = 1-t.
\end{align}
We adopt the Fock representation of the algebra \eqref{t-boson}, viewing $\phi$ and $\phid$ as annihilation and creation operators (respectively) on the vector space \eqref{space}. In particular, we take infinitely many commuting copies of the algebra \eqref{t-boson} (distinguished by subscripts), which act non-trivially in $i^{\rm th}$ factor of $V$ as follows:
\begin{align*}
\phi_i \ket{m}_i
=
(1-t^m)
\ket{m-1}_i,
\qquad
\phid_i \ket{m}_i
=
\ket{m+1}_i.
\end{align*}
As a special case of these relations, we have $\phi_i \ket{0}_i = 0$, ensuring the closure of \eqref{space} under the action of the $t$-boson algebra (negative occupation numbers are not possible). 

Similarly, when acting on the factors of the more general vector space \eqref{alpha-space}, we take the representation
\begin{align*}
\phi_i \ket{m+\alpha}_i
=
(1-t^{m+\alpha})
\ket{m+\alpha-1}_i,
\qquad
\phid_i \ket{m+\alpha}_i
=
\ket{m+\alpha+1}_i.
\end{align*}
Since the coefficient $(1-t^{m+\alpha})$ is generically non-vanishing, there is no highest/lowest weight vector, and we are obliged to allow $m_i \in \mathbb{Z}$ to ensure the closure of $V(\alpha_0,\alpha_1,\dots)$ under the action of the algebra.

\subsection{$R$ matrix, Yang--Baxter and unitarity relations}

Before presenting the integrable model to be studied in this paper, we give the $R$ matrix which underpins its integrability:
\begin{align}
\label{Rmat}
R_{ab}(x/y)
=
\begin{pmatrix}
\frac{1-t x/y}{1-x/y} & 0 & 0 & 0 \\
0 & t & \frac{(1-t)x/y}{1-x/y} & 0 \\
0 & \frac{1-t}{1-x/y} & 1 & 0 \\
0 & 0 & 0 & \frac{1-t x/y}{1-x/y} 
\end{pmatrix}_{ab}
\in 
{\rm End}(W_a \otimes W_b).
\end{align}
This is the $R$ matrix of the six-vertex model, in the multiplicative (trigonometric) parametrization. It acts in the tensor product of two-dimensional auxiliary spaces\footnote{Throughout the paper, we use the letter $W$ for vector spaces which are ``auxiliary'' (two-dimensional) and reserve $V$ for the bosonic (infinite-dimensional) spaces.} $W_a$ and $W_b$, and its non-zero components can be represented graphically as follows:
\begin{equation*}
\begin{tabular}{cccccc}
\begin{tikzpicture}[scale=0.6]
\draw[thick, smooth,arrow=0.25,arrow=0.75] (-1,0) -- (1,0);
\node[label={left: \fs \red $x$}] at (-1,0) {};
\draw[thick, smooth,arrow=0.25,arrow=0.75] (0,-1) -- (0,1);
\node[label={below: \fs \red $y$}] at (0,-1) {};
\end{tikzpicture}
&
\begin{tikzpicture}[scale=0.6]
\draw[thick, smooth,arrow=0.25,arrow=0.75] (1,0) -- (-1,0);
\node[label={left: \fs \red $x$}] at (-1,0) {};
\draw[thick, smooth,arrow=0.25,arrow=0.75] (0,1) -- (0,-1);
\node[label={below: \fs \red $y$}] at (0,-1) {};
\end{tikzpicture}
&
\begin{tikzpicture}[scale=0.6]
\draw[thick, smooth,arrow=0.25,arrow=0.75] (-1,0) -- (1,0);
\node[label={left: \fs \red $x$}] at (-1,0) {};
\draw[thick, smooth,arrow=0.25,arrow=0.75] (0,1) -- (0,-1);
\node[label={below: \fs \red $y$}] at (0,-1) {};
\end{tikzpicture}
&
\begin{tikzpicture}[scale=0.6]
\draw[thick, smooth,arrow=0.25,arrow=0.75] (1,0) -- (-1,0);
\node[label={left: \fs \red $x$}] at (-1,0) {};
\draw[thick, smooth,arrow=0.25,arrow=0.75] (0,-1) -- (0,1);
\node[label={below: \fs \red $y$}] at (0,-1) {};
\end{tikzpicture}
&
\begin{tikzpicture}[scale=0.6]
\draw[thick, smooth,arrow=0.5] (1,0) -- (0,0);
\draw[thick, smooth,arrow=0.5] (-1,0) -- (0,0);
\node[label={left: \fs \red $x$}] at (-1,0) {};
\draw[thick, smooth,arrow=0.5] (0,0) -- (0,-1);
\draw[thick, smooth,arrow=0.5] (0,0) -- (0,1);
\node[label={below: \fs \red $y$}] at (0,-1) {};
\end{tikzpicture}
&
\begin{tikzpicture}[scale=0.6]
\draw[thick, smooth,arrow=0.5] (0,0) -- (1,0);
\draw[thick, smooth,arrow=0.5] (0,0) -- (-1,0);
\node[label={left: \fs \red $x$}] at (-1,0) {};
\draw[thick, smooth,arrow=0.5] (0,-1) -- (0,0);
\draw[thick, smooth,arrow=0.5] (0,1) -- (0,0);
\node[label={below: \fs \red $y$}] at (0,-1) {};
\end{tikzpicture}
\\
\quad
$\frac{1-t z}{1-z}$
&
\quad
$\frac{1-t z}{1-z}$
&
\quad
$t$
&
\quad
$1$
&
\quad
$\frac{(1-t)z}{1-z}$
&
\quad
$\frac{1-t}{1-z}$
\end{tabular}
\end{equation*}
where we abbreviate the ratio of horizontal and vertical spectral parameters by $x/y \equiv z$. Since we will ultimately be interested in lattice paths formed by particles, it is more convenient for our purposes to draw these six vertices as
\begin{equation}
\label{six-v}
\begin{tabular}{cccccc}
\begin{tikzpicture}[scale=0.6]
\draw[thick, smooth] (-1,0) -- (1,0);
\node[label={left: \fs \red $x$}] at (-1,0) {};
\draw[thick, smooth] (0,-1) -- (0,1);
\node[label={below: \fs \red $y$}] at (0,-1) {};
\ebull{-0.5}{0}{0.09}; \ebull{0.5}{0}{0.09};
\ebull{0}{-0.5}{0.09}; \ebull{0}{0.5}{0.09};
\end{tikzpicture}
&
\begin{tikzpicture}[scale=0.6]
\draw[thick, smooth] (1,0) -- (-1,0);
\node[label={left: \fs \red $x$}] at (-1,0) {};
\draw[thick, smooth] (0,1) -- (0,-1);
\node[label={below: \fs \red $y$}] at (0,-1) {};
\gbull{-0.5}{0}{0.09}; \gbull{0.5}{0}{0.09};
\gbull{0}{-0.5}{0.09}; \gbull{0}{0.5}{0.09};
\end{tikzpicture}
&
\begin{tikzpicture}[scale=0.6]
\draw[thick, smooth] (-1,0) -- (1,0);
\node[label={left: \fs \red $x$}] at (-1,0) {};
\draw[thick, smooth] (0,1) -- (0,-1);
\node[label={below: \fs \red $y$}] at (0,-1) {};
\ebull{-0.5}{0}{0.09}; \ebull{0.5}{0}{0.09};
\gbull{0}{-0.5}{0.09}; \gbull{0}{0.5}{0.09};
\end{tikzpicture}
&
\begin{tikzpicture}[scale=0.6]
\draw[thick, smooth] (1,0) -- (-1,0);
\node[label={left: \fs \red $x$}] at (-1,0) {};
\draw[thick, smooth] (0,-1) -- (0,1);
\node[label={below: \fs \red $y$}] at (0,-1) {};
\gbull{-0.5}{0}{0.09}; \gbull{0.5}{0}{0.09};
\ebull{0}{-0.5}{0.09}; \ebull{0}{0.5}{0.09};
\end{tikzpicture}
&
\begin{tikzpicture}[scale=0.6]
\draw[thick, smooth] (1,0) -- (0,0);
\draw[thick, smooth] (-1,0) -- (0,0);
\node[label={left: \fs \red $x$}] at (-1,0) {};
\draw[thick, smooth] (0,0) -- (0,-1);
\draw[thick, smooth] (0,0) -- (0,1);
\node[label={below: \fs \red $y$}] at (0,-1) {};
\ebull{-0.5}{0}{0.09}; \gbull{0.5}{0}{0.09};
\gbull{0}{-0.5}{0.09}; \ebull{0}{0.5}{0.09};
\end{tikzpicture}
&
\begin{tikzpicture}[scale=0.6]
\draw[thick, smooth] (0,0) -- (1,0);
\draw[thick, smooth] (0,0) -- (-1,0);
\node[label={left: \fs \red $x$}] at (-1,0) {};
\draw[thick, smooth] (0,-1) -- (0,0);
\draw[thick, smooth] (0,1) -- (0,0);
\node[label={below: \fs \red $y$}] at (0,-1) {};
\gbull{-0.5}{0}{0.09}; \ebull{0.5}{0}{0.09};
\ebull{0}{-0.5}{0.09}; \gbull{0}{0.5}{0.09};
\end{tikzpicture}
\\
\quad
$\frac{1-t z}{1-z}$
&
\quad
$\frac{1-t z}{1-z}$
&
\quad
$t$
&
\quad
$1$
&
\quad
$\frac{(1-t)z}{1-z}$
&
\quad
$\frac{1-t}{1-z}$
\end{tabular}
\end{equation}
where we have replaced a left or upward arrow with a hole, and a right or downward arrow with a particle. Throughout the paper, we use red labels not only to indicate the spectral parameter on a given line, but also to specify the orientation of that line. In all cases, a line is oriented {\it towards} its spectral parameter label.

The $R$ matrix \eqref{Rmat} satisfies the Yang--Baxter equation (graphical version in parentheses)
\begin{align*}
R_{ab}(x/y) R_{ac}(x/z) R_{bc}(y/z)
=
R_{bc}(y/z) R_{ac}(x/z) R_{ab}(x/y),
\quad\quad
\left(
\begin{tikzpicture}[scale=0.7,baseline=0.5cm]
\draw[smooth] (0,1.5) arc (-155:-90:2);
\draw[smooth] (0,0) arc (155:90:2);
\draw[smooth] ({2*cos(25)},{1.5-(2-2*sin(25))})--({2*cos(25)+0.5},{1.5-(2-2*sin(25))});
\draw[smooth] ({2*cos(25)},{(2-2*sin(25))})--({2*cos(25)+0.5},{(2-2*sin(25))});
\draw[smooth] (1.5,-0.25)--(1.5,1.75);
\node[left] at (0,1.5) {\fs \red $x$};
\node[left] at (0,0) {\fs \red $y$};
\node[below] at (1.5,-0.25) {\fs \red $z$};
\end{tikzpicture}
\quad
=
\quad
\begin{tikzpicture}[scale=0.7,baseline=0.5cm]
\draw[smooth] (10,1.5) arc (-25:-90:2);
\draw[smooth] (10,0) arc (25:90:2);
\draw[smooth] ({10-2*cos(25)},{1.5-(2-2*sin(25))})--({10-2*cos(25)-0.5},{1.5-(2-2*sin(25))});
\draw[smooth] ({10-2*cos(25)},{(2-2*sin(25))})--({10-2*cos(25)-0.5},{(2-2*sin(25))});
\draw[smooth] (8.5,-0.25)--(8.5,1.75);
\node[left] at ({10-2*cos(25)-0.5},{1.5-(2-2*sin(25))}) {\fs \red $y$};
\node[left] at ({10-2*cos(25)-0.5},{(2-2*sin(25))}) {\fs \red $x$};
\node[below] at (8.5,-0.25) {\fs \red $z$};
\end{tikzpicture}
\right),
\end{align*}
and the unitarity relation
\begin{align}
\label{unitarity}
R_{ab}(x/y)
R_{ba}(y/x)
=
\frac{(y-tx)(x-ty)}{(y-x)(x-y)},
\quad\quad
\left(
\begin{tikzpicture}[scale=0.7,baseline=0.5cm]
\draw[smooth] (0,1.5) arc (-155:-90:2);
\draw[smooth] (0,0) arc (155:90:2);
\draw[smooth] ({2*cos(25)+1.8},1.5) arc (-25:-90:2);
\draw[smooth] ({2*cos(25)+1.8},0) arc (25:90:2);
\node[left] at (0,1.5) {\fs \red $x$};
\node[left] at (0,0) {\fs \red $y$};
\end{tikzpicture}
\quad
=
{\fs \frac{(y-tx)(x-ty)}{(y-x)(x-y)} }
\times
\begin{tikzpicture}[scale=0.7,baseline=0.5cm]
\draw[smooth] (0,1.2) -- (2,1.2);
\draw[smooth] (0,0.3) -- (2,0.3);
\node[left] at (0,1.2) {\fs \red $x$};
\node[left] at (0,0.3) {\fs \red $y$};
\end{tikzpicture}
\right).
\end{align}

\subsection{Definition of $L$ operators}

An integrable model of $t$-bosons can be constructed from the $L$ matrix
\begin{align}
\label{Lmat}
L_a(x)
=
L_a(x,\phi,\phid)
=
\begin{pmatrix}
x & x \phid \\
\phi & 1
\end{pmatrix}_a
=
\begin{tikzpicture}[baseline=0]
\node at (-0.7,0) {$\ss \red x$};
\dark{-0.5}{-0.5};
\end{tikzpicture}
\end{align}
which we represent graphically by the tile in the final equality. It is an operator acting in $W_a \otimes V$, and its four components in the auxiliary space $W_a$ are denoted by assigning holes/particles to the vertical edges of the tile:
\begin{equation}
\label{tiles}
\begin{tikzpicture}[baseline=0]
\matrix[column sep={2.5cm,between origins},row sep=0.25cm]{
\dark{-0.5}{-0.5}; 
\draw[gbline] (0,-0.5) node[below] {$\ss m$} -- (0,0.5) node[above] {$\ss m$};
\ebull{-0.5}{0}{0.05};\ebull{0.5}{0}{0.05};
& 
\dark{-0.5}{-0.5}; 
\draw[gbline] (0,-0.5) node[below] {$\ss m+1$} -- (0,0.5) node[above] {$\ss m$};
\draw[gline] (0.05,-0.5) -- (0.05,0) -- (0.5,0);
\ebull{-0.5}{0}{0.05};\gbull{0.5}{0}{0.05};
& 
\dark{-0.5}{-0.5}; 
\draw[gbline] (0,-0.5) node[below] {$\ss m-1$} -- (0,0.5) node[above] {$\ss m$};
\draw[gline] (-0.05,0.5) -- (-0.05,0) -- (-0.5,0);
\gbull{-0.5}{0}{0.05};\ebull{0.5}{0}{0.05};
& 
\dark{-0.5}{-0.5};
\draw[gbline] (0,-0.5) node[below] {$\ss m$} -- (0,0.5) node[above] {$\ss m$};
\draw[gline] (-0.5,0) -- (0.5,0);
\gbull{-0.5}{0}{0.05};\gbull{0.5}{0}{0.05};
\\
\node{$x$}; & \node{$x$}; & \node{$(1-t^m)$}; & \node{$1$};
\\
};
\end{tikzpicture}
\end{equation}
where we write the associated Boltzmann weight under each tile. The $L$ matrix satisfies the {\it intertwining equation,} which is essential to the derivation of all Cauchy and Littlewood identities in this paper: 
\begin{align}
\label{rll}
R_{ab}(x/y) L_a(x) L_b(y)
=
L_b(y) L_a(x) R_{ab}(x/y)
\quad
\text{in End}(W_a \otimes W_b \otimes V),
\quad
\left(
\begin{tikzpicture}[baseline=0,scale=0.7]
\draw[thick] (0,0.5) -- (1,-0.5) -- (2,-0.5);
\draw[thick] (0,-0.5) -- (1,0.5) -- (2,0.5);
\dark{1}{-1};\dark{1}{0};
\node at (-0.3,0.5) {$\ss \red x$};
\node at (-0.3,-0.5) {$\ss \red y$};
\end{tikzpicture}
\ \  
=
\
\begin{tikzpicture}[baseline=-0,scale=0.7]
\draw[thick] (-1,0.5) -- (0,0.5) -- (1,-0.5);
\draw[thick] (-1,-0.5) -- (0,-0.5) -- (1,0.5);
\dark{-1}{-1};\dark{-1}{0};
\node at (-1.2,0.5) {$\ss \red x$};
\node at (-1.2,-0.5) {$\ss \red y$};
\end{tikzpicture}
\right),
\end{align}
where $R_{ab}(x/y)$ is given by \eqref{Rmat}.

It will be convenient to define an $L^{*}$ matrix $L^{*}(x)$ which is trivially related to the first:
\begin{align}
\label{Lstar}
L^{*}_a(x)
=
x L_a(\b{x})
=
\begin{pmatrix}
1 & \phid \\
x \phi & x
\end{pmatrix}_a
=
\begin{tikzpicture}[baseline=0]
\node at (-0.7,0) {$\ss \red \b{x}$};
\light{-0.5}{-0.5};
\end{tikzpicture}
\end{align}
where $\b{x} := 1/x$, a notation that we use ubiquitously throughout the paper. The difference between the $L$ matrices \eqref{Lmat} and \eqref{Lstar} is of course only superficial, but the redistribution of $x$ weights turns out to be crucial. To emphasize this difference in the Boltzmann weights, we prescribe the entries of \eqref{Lstar} their own graphical notation:
\begin{equation}
\label{tiles*}
\begin{tikzpicture}[baseline=0]
\matrix[column sep={2.5cm,between origins},row sep=0.25cm]{
\cellbb 
\draw[gbline] (0,-0.5) node[below] {$\ss m$} -- (0,0.5) node[above] {$\ss m$};
\ebull{-0.5}{0}{0.05};\ebull{0.5}{0}{0.05};
& 
\cellbb 
\draw[gbline] (0,-0.5) node[below] {$\ss m+1$} -- (0,0.5) node[above] {$\ss m$};
\draw[gline] (0.05,-0.5) -- (0.05,0) -- (0.5,0);
\ebull{-0.5}{0}{0.05};\gbull{0.5}{0}{0.05};
& 
\cellbb 
\draw[gbline] (0,-0.5) node[below] {$\ss m-1$} -- (0,0.5) node[above] {$\ss m$};
\draw[gline] (-0.05,0.5) -- (-0.05,0) -- (-0.5,0);
\gbull{-0.5}{0}{0.05};\ebull{0.5}{0}{0.05};
& 
\cellbb
\draw[gbline] (0,-0.5) node[below] {$\ss m$} -- (0,0.5) node[above] {$\ss m$};
\draw[gline] (-0.5,0) -- (0.5,0);
\gbull{-0.5}{0}{0.05};\gbull{0.5}{0}{0.05};
\\
\\
\node{$1$}; & \node{$1$}; & \node{$(1-t^m)x$}; & \node{$x$}; 
\\
};
\end{tikzpicture}
\end{equation}
using a lighter shading for these tiles. In view of the simple relation between the $L^{*}$ matrix and the original one \eqref{Lmat}, we immediately deduce the intertwining equations
\begin{align}
\label{rll*}
R_{ab}(xy) L_a(x) L^{*}_b(y)
&=
L^{*}_b(y) L_a(x) R_{ab}(xy),
\quad
\left(
\begin{tikzpicture}[baseline=0,scale=0.7]
\draw[thick] (0,0.5) -- (1,-0.5) -- (2,-0.5);
\draw[thick] (0,-0.5) -- (1,0.5) -- (2,0.5);
\dark{1}{-1};\light{1}{0};
\node at (-0.3,0.5) {$\ss \red x$};
\node at (-0.3,-0.5) {$\ss \red \b{y}$};
\end{tikzpicture}
\ \  
=
\
\begin{tikzpicture}[baseline=-0,scale=0.7]
\draw[thick] (-1,0.5) -- (0,0.5) -- (1,-0.5);
\draw[thick] (-1,-0.5) -- (0,-0.5) -- (1,0.5);
\light{-1}{-1};\dark{-1}{0};
\node at (-1.2,0.5) {$\ss \red x$};
\node at (-1.2,-0.5) {$\ss \red \b{y}$};
\end{tikzpicture}
\right),
\end{align}

\begin{align}
\label{rl*l*}
R_{ab}(y/x) L^{*}_a(x) L^{*}_b(y)
=
L^{*}_b(y) L^{*}_a(x) R_{ab}(y/x),
\quad
\left(
\begin{tikzpicture}[baseline=0,scale=0.7]
\draw[thick] (0,0.5) -- (1,-0.5) -- (2,-0.5);
\draw[thick] (0,-0.5) -- (1,0.5) -- (2,0.5);
\light{1}{-1};\light{1}{0};
\node at (-0.3,0.5) {$\ss \red \b{x}$};
\node at (-0.3,-0.5) {$\ss \red \b{y}$};
\end{tikzpicture}
\ \  
=
\
\begin{tikzpicture}[baseline=-0,scale=0.7]
\draw[thick] (-1,0.5) -- (0,0.5) -- (1,-0.5);
\draw[thick] (-1,-0.5) -- (0,-0.5) -- (1,0.5);
\light{-1}{-1};\light{-1}{0};
\node at (-1.2,0.5) {$\ss \red \b{x}$};
\node at (-1.2,-0.5) {$\ss \red \b{y}$};
\end{tikzpicture}
\right).
\end{align}

\subsection{Row-to-row transfer matrices}

We introduce two types of row-to-row transfer matrices, obtained by taking the product of infinitely many $L$ matrices\footnote{Some care is needed to properly define such an infinite product. We note that \eqref{T} and \eqref{T*} have the same action on any finite element of $V$ as
\begin{align*}
\prod_{i=0}^{M} L_a(x,\phi_i,\phid_i) 
\prod_{j=M+1}^{\infty} \begin{pmatrix} x & 0 \\ 0 & 1 \end{pmatrix}_a
\quad
{\rm and}
\quad
\prod_{i=0}^{M} L^{*}_a(x,\phi_i,\phid_i) 
\prod_{j=M+1}^{\infty} \begin{pmatrix} 1 & 0 \\ 0 & x \end{pmatrix}_a
\end{align*}
respectively, where $M$ is chosen to be sufficiently large. So from a graphical point of view, $T_a(x)$ is the sum over all possible infinite row configurations of the tiles \eqref{tiles}, such that only the tiles \begin{tikzpicture}[scale=0.5,baseline=-0.1cm] \dark{-0.5}{-0.5};  \ebull{-0.5}{0}{0.09};\ebull{0.5}{0}{0.09}; \end{tikzpicture} and \begin{tikzpicture}[scale=0.5,baseline=-0.1cm] \dark{-0.5}{-0.5};  \gbull{-0.5}{0}{0.09};\gbull{0.5}{0}{0.09}; \end{tikzpicture} appear when sufficiently far to the right. An analogous statement applies to $T_a^{*}(x)$.}:
\begin{align}
\label{T}
T_a(x)
&=
\prod_{i=0}^{\infty}
L_a(x,\phi_i,\phid_i)
\in
{\rm End}(W_a \otimes V_0 \otimes V_1 \otimes \cdots),
\quad
\left(
T_a(x)
=
\begin{tikzpicture}[baseline=0.25cm,scale=0.7]
\node at (0.5,0.5) {$\ss \red x$};
\foreach\x in {1,...,5}{
\dark{\x}{0};
}
\node at (6.5,0.5) {$\cdots$};
\end{tikzpicture}
\ \
\right)
\\
\label{T*}
T^{*}_a(x)
&=
\prod_{i=0}^{\infty}
L^{*}_a(x,\phi_i,\phid_i)
\in
{\rm End}(W_a \otimes V_0 \otimes V_1 \otimes \cdots),
\quad
\left(
T^{*}_a(x)
=
\begin{tikzpicture}[baseline=0.25cm,scale=0.7]
\node at (0.5,0.5) {$\ss \red \b{x}$};
\foreach\x in {1,...,5}{
\light{\x}{0};
}
\node at (6.5,0.5) {$\cdots$};
\end{tikzpicture}
\ \ 
\right)
\end{align}
where we have emphasized that each $L$ matrix depends on a different copy of the 
$t$-boson algebra, and the product is ordered from left to right as the index increases. $T_a(x)$ and 
$T^{*}_a(x)$ are $2\times 2$ matrices, with operator entries which act on the physical space 
$V=\bigotimes_{i=0}^{\infty} V_i$. Note that their entries are only well defined if we assume that the spectral parameter satisfies $|x| < 1$, which ensures that terms with unbounded degree in $x$ are equal to zero. We find that
\begin{align}
\label{mat-form-dark}
T_a(x)
=
\begin{pmatrix}
0 & T_{+}(x)
\medskip
\\
0 & T_{-}(x)
\end{pmatrix}_a
=
\begin{pmatrix}
\begin{tikzpicture}[baseline=0.25cm,scale=0.5]
\node at (0.5,0.5) {$\ss \red x$};
\foreach\x in {1,...,5}{
\dark{\x}{0};
}
\ebull{1}{0.5}{0.09}; \ebull{6}{0.5}{0.09};
\node at (6.75,0.5) {$\cdots$};
\end{tikzpicture}
& 
\begin{tikzpicture}[baseline=0.25cm,scale=0.5]
\node at (0.5,0.5) {$\ss \red x$};
\foreach\x in {1,...,5}{
\dark{\x}{0};
}
\ebull{1}{0.5}{0.09}; \gbull{6}{0.5}{0.09};
\node at (6.75,0.5) {$\cdots$};
\end{tikzpicture}
\ \
\medskip
\\
\begin{tikzpicture}[baseline=0.25cm,scale=0.5]
\node at (0.5,0.5) {$\ss \red x$};
\foreach\x in {1,...,5}{
\dark{\x}{0};
}
\gbull{1}{0.5}{0.09}; \ebull{6}{0.5}{0.09};
\node at (6.75,0.5) {$\cdots$};
\end{tikzpicture} 
& 
\begin{tikzpicture}[baseline=0.25cm,scale=0.5]
\node at (0.5,0.5) {$\ss \red x$};
\foreach\x in {1,...,5}{
\dark{\x}{0};
}
\gbull{1}{0.5}{0.09}; \gbull{6}{0.5}{0.09};
\node at (6.75,0.5) {$\cdots$};
\end{tikzpicture}
\ \ 
\end{pmatrix}_a
\\ \nonumber \\
\label{mat-form-light}
T^{*}_a(x)
=
\begin{pmatrix}
T^{*}_{+}(x) & 0
\medskip
\\
T^{*}_{-}(x) & 0
\end{pmatrix}_a
=
\begin{pmatrix}
\begin{tikzpicture}[baseline=0.25cm,scale=0.5]
\node at (0.5,0.5) {$\ss \red \b{x}$};
\foreach\x in {1,...,5}{
\light{\x}{0};
}
\ebull{1}{0.5}{0.09}; \ebull{6}{0.5}{0.09};
\node at (6.75,0.5) {$\cdots$};
\end{tikzpicture}
& 
\begin{tikzpicture}[baseline=0.25cm,scale=0.5]
\node at (0.5,0.5) {$\ss \red \b{x}$};
\foreach\x in {1,...,5}{
\light{\x}{0};
}
\ebull{1}{0.5}{0.09}; \gbull{6}{0.5}{0.09};
\node at (6.75,0.5) {$\cdots$};
\end{tikzpicture}
\ \ 
\medskip
\\
\begin{tikzpicture}[baseline=0.25cm,scale=0.5]
\node at (0.5,0.5) {$\ss \red \b{x}$};
\foreach\x in {1,...,5}{
\light{\x}{0};
}
\gbull{1}{0.5}{0.09}; \ebull{6}{0.5}{0.09};
\node at (6.75,0.5) {$\cdots$};
\end{tikzpicture}
& 
\begin{tikzpicture}[baseline=0.25cm,scale=0.5]
\node at (0.5,0.5) {$\ss \red \b{x}$};
\foreach\x in {1,...,5}{
\light{\x}{0};
}
\gbull{1}{0.5}{0.09}; \gbull{6}{0.5}{0.09};
\node at (6.75,0.5) {$\cdots$};
\end{tikzpicture}
\ \ 
\end{pmatrix}_a
\end{align}
where $T_{\pm}(x)$ and $T^{*}_{\pm}(x)$ are non-trivial operator-valued power series in $x$. We define, further, the sums of matrix entries:
\begin{align*}
\mathcal{T}(x) = T_{+}(x) + T_{-}(x),
\qquad
\mathcal{T}^{*}(x) = T^{*}_{+}(x) + T^{*}_{-}(x).
\end{align*}
The vanishing conditions of certain matrix elements in \eqref{mat-form-dark} and \eqref{mat-form-light} will be very important in our subsequent proofs, since they allow us to eliminate unwanted terms.


\subsection{Boundary covector, reflection and fish equations}

In this work we will not consider the most general solution of the reflection equation, but restrict ourselves to the constant boundary covector
\begin{align*}
\langle K |_{a\b{a}}
=
\begin{pmatrix}
1 & 0
\end{pmatrix}_a
\otimes
\begin{pmatrix}
0 & 1
\end{pmatrix}_{\b{a}}
-
t
\begin{pmatrix}
0 & 1
\end{pmatrix}_a
\otimes
\begin{pmatrix}
1 & 0
\end{pmatrix}_{\b{a}},
\qquad
\langle K |_{a\b{a}}
\in
W^{*}_a \otimes W^{*}_{\b{a}}.
\end{align*}
As a 4-dimensional row-vector this has two non-zero components, which are represented by the U-turn vertices 
\begin{align*}
\bra{K}_{a\b{a}}
\begin{pmatrix}
1 \\ 0
\end{pmatrix}_{a}
\otimes
\begin{pmatrix}
0 \\ 1
\end{pmatrix}_{\b{a}}
=
\begin{tikzpicture}[scale=0.4,baseline=9]
\begin{scope}[xscale=-1,yscale=1]
\draw[smooth] ({2*sqrt(2) + 3 + 4},0) arc (-90:90:1);
\draw[smooth] ({2*sqrt(2) + 3 + 3.5},2)--({2*sqrt(2) + 3 + 4},2);
\draw[smooth] ({2*sqrt(2) + 3 + 3.5},0)--({2*sqrt(2) + 3 + 4},0);
\node at ({2*sqrt(2) + 3 + 5},1) {$\bullet$};
\ebull{{2*sqrt(2) + 6.5}}{2}{0.14};
\gbull{{2*sqrt(2) + 6.5}}{0}{0.14};
\end{scope}
\end{tikzpicture}
=
1,
\qquad
\bra{K}_{a\b{a}}
\begin{pmatrix}
0 \\ 1
\end{pmatrix}_{a}
\otimes
\begin{pmatrix}
1 \\ 0
\end{pmatrix}_{\b{a}}
=
\begin{tikzpicture}[scale=0.4,baseline=9]
\begin{scope}[xscale=-1,yscale=1]
\draw[smooth] ({2*sqrt(2) + 3 + 4},0) arc (-90:90:1);
\draw[smooth] ({2*sqrt(2) + 3 + 3.5},2)--({2*sqrt(2) + 3 + 4},2);
\draw[smooth] ({2*sqrt(2) + 3 + 3.5},0)--({2*sqrt(2) + 3 + 4},0);
\node at ({2*sqrt(2) + 3 + 5},1) {$\bullet$};
\gbull{{2*sqrt(2) + 6.5}}{2}{0.14};
\ebull{{2*sqrt(2) + 6.5}}{0}{0.14};
\end{scope}
\end{tikzpicture}
=
-t.
\end{align*}
The boundary covector satisfies the reflection equation:
\begin{align}
\label{reflection}
\langle K |_{a \b{a}} 
\langle K |_{b \b{b}}\ 
R_{\b{a} b}(1/x/y)
R_{ab}(x/y)
=
\langle K |_{b \b{b}}
\langle K |_{a \b{a}}\
R_{\b{b} a}(1/x/y)
R_{\b{b} \b{a}}(x/y),
\quad
\left(
\begin{tikzpicture}[scale=0.30,baseline=1cm]
\begin{scope}[xscale=-1,yscale=1]
\draw[smooth] (17,0)--(21,0);
\draw[smooth] (17,4)--(21,4);
\draw[smooth] (17,6)--(21,6);
\draw[smooth] (21,4) arc (-90:90:1); 
\draw[smooth] (21,0) arc (-90:90:1); 
\draw[smooth] (20,2) arc (270:180:1);
\draw[smooth] (19,8)--(19,3);
\draw[smooth] (20,2)--(21,2);
\node at (22,1) {$\bullet$};
\node at (22,5) {$\bullet$};
\node at (20,4.4) {$\ss \red \b{x}$};
\node at (20,6.4) {$\ss \red x$};
\node at (18.7,2) {$\ss \red y$};
\end{scope}
\end{tikzpicture}
\quad
=
\quad
\begin{tikzpicture}[scale=0.30,rotate=90,baseline=-1.5cm]
\begin{scope}[xscale=-1,yscale=1]
\draw[smooth] (2,2)--(6,2);
\draw[smooth] (2,4)--(6,4);
\draw[smooth] (2,8)--(6,8);
\draw[smooth] (6,2) arc (-90:90:1); 
\draw[smooth] (6,6) arc (-90:90:1);
\draw[smooth] (4,5) arc (-180:-270:1);
\draw[smooth] (4,5)--(4,0);
\draw[smooth] (5,6)--(6,6);
\node at (7,3) {$\bullet$};
\node at (7,7) {$\bullet$};
\node at (5.5,2.4) {$\ss \red \b{x}$};
\node at (5.5,4.4) {$\ss \red x$};
\node at (4,6.5) {$\ss \red \b{y}$};
\end{scope}
\end{tikzpicture}
\right)
\end{align}
The next lemma is a consequence of the reflection equation \eqref{reflection} and unitarity \eqref{unitarity}, and will be useful at several stages in the paper.

\begin{lem}
\label{lem:vanishing}
The following partition function is identically zero for all $n \geq 1$, and all possible configurations of holes/particles along the top edge of the lattice:
\begin{align}
\label{vanishing-lem}
\begin{tikzpicture}[scale=0.8,baseline=0,rotate=45]
\foreach\x in {1,...,2}{
\draw[fline] (-4.5+0.5*\x,4+0.5*\x) -- (-1+0.5*\x,0.5+0.5*\x);
\gbull{-1+0.5*\x}{0.5+0.5*\x}{0.07}
}
\begin{scope}[xscale=1,yscale=-1,yshift=-8cm]
\foreach\x in {1,...,6}{
\draw[fline] (-4.5+0.5*\x,4+0.5*\x) -- (-3+0.5*\x,2.5+0.5*\x);
}
\end{scope}
\foreach\x in {0}{
\draw[smooth] (-4+\x,4.5+\x) arc (225:45:0.353553);
\node at (-4+\x,5+\x) {\fs$\bullet$};
}
\begin{scope}[xscale=1,yscale=-1,yshift=-8cm]
\foreach\x in {0,...,2}{
\draw[smooth] (-4+\x,4.5+\x) arc (225:45:0.353553);
\node at (-4+\x,5+\x) {\fs$\bullet$};
}
\end{scope}
\node at (-4,5.5) {$\ss \red x$};
\node at (-4.5,5) {$\ss \red \b{x}$};
\node at (-4.5,3) {$\ss \red y_n$};
\node at (-4,2.5) {$\ss \red \b{y}_n$};
\node at (-3.2,1.8) {\tiny\red $\cdots$};
\node at (-2.5,1) {$\ss \red y_1$};
\node at (-2,0.5) {$\ss \red \b{y}_1$};
\node at (0.5,0.5) {$=0$};
\end{tikzpicture}
\end{align}
\end{lem}

\begin{proof}
We consider the leftmost $2 \times 2$ block of the lattice, which is given algebraically by
\begin{align*}
\begin{tikzpicture}[scale=0.6,baseline=0.3cm,rotate=45]
\foreach\x in {1,...,2}{
\draw[fline] (-4.5+0.5*\x,4+0.5*\x) -- (-3+0.5*\x,2.5+0.5*\x);
}
\begin{scope}[xscale=1,yscale=-1,yshift=-8cm]
\foreach\x in {1,...,2}{
\draw[fline] (-4.5+0.5*\x,4+0.5*\x) -- (-3+0.5*\x,2.5+0.5*\x);
}
\end{scope}
\foreach\x in {0}{
\draw[smooth] (-4+\x,4.5+\x) arc (225:45:0.353553);
\node at (-4+\x,5+\x) {\fs$\bullet$};
}
\begin{scope}[xscale=1,yscale=-1,yshift=-8cm]
\foreach\x in {0}{
\draw[smooth] (-4+\x,4.5+\x) arc (225:45:0.353553);
\node at (-4+\x,5+\x) {\fs$\bullet$};
}
\end{scope}
\node at (-4,5.5) {$\ss \red x$};
\node at (-4.5,5) {$\ss \red \b{x}$};
\node at (-4.5,3) {$\ss \red y_n$};
\node at (-4,2.5) {$\ss \red \b{y}_n$};
\end{tikzpicture}
=
\langle K |_{a \b{a}}  \langle K |_{n \b{n}} 
R_{\b{a} n}(1/x/y_n) R_{a n}(x/y_n) 
R_{\b{a} \b{n}} (y_n/x) R_{a \b{n}} (x y_n).
\end{align*}
Using the reflection equation \eqref{reflection}, followed by two applications of the unitarity relation \eqref{unitarity}, this becomes
\begin{align*}
\begin{tikzpicture}[scale=0.6,baseline=0.3cm,rotate=45]
\foreach\x in {1,...,2}{
\draw[fline] (-4.5+0.5*\x,4+0.5*\x) -- (-3+0.5*\x,2.5+0.5*\x);
}
\begin{scope}[xscale=1,yscale=-1,yshift=-8cm]
\foreach\x in {1,...,2}{
\draw[fline] (-4.5+0.5*\x,4+0.5*\x) -- (-3+0.5*\x,2.5+0.5*\x);
}
\end{scope}
\foreach\x in {0}{
\draw[smooth] (-4+\x,4.5+\x) arc (225:45:0.353553);
\node at (-4+\x,5+\x) {\fs$\bullet$};
}
\begin{scope}[xscale=1,yscale=-1,yshift=-8cm]
\foreach\x in {0}{
\draw[smooth] (-4+\x,4.5+\x) arc (225:45:0.353553);
\node at (-4+\x,5+\x) {\fs$\bullet$};
}
\end{scope}
\node at (-4,5.5) {$\ss \red x$};
\node at (-4.5,5) {$\ss \red \b{x}$};
\node at (-4.5,3) {$\ss \red y_n$};
\node at (-4,2.5) {$\ss \red \b{y}_n$};
\end{tikzpicture}
&=
\langle K |_{n \b{n}} \langle K |_{a \b{a}}
R_{\b{n} a}(1/x/y_n) R_{\b{n} \b{a}}(x/y_n) 
R_{\b{a} \b{n}} (y_n/x) R_{a \b{n}} (x y_n)
\\[-0.5cm]
&=
\frac{(1-t x y_n)(1-t x \b{y}_n)(1-t \b{x} y_n)(1-t \b{x} \b{y}_n)}
{(1-x y_n)(1-x \b{y}_n)(1-\b{x} y_n)(1-\b{x} \b{y}_n)}
\langle K |_{n \b{n}} \langle K |_{a \b{a}},
\end{align*}
and we have trivialized the leftmost block, and can proceed to the next one to the right. Iterating this $n$ times, we eventually end up with the U-turn vertex 
$
\begin{tikzpicture}[scale=0.4,baseline=9]
\begin{scope}[xscale=-1,yscale=1]
\draw[smooth] ({2*sqrt(2) + 3 + 4},0) arc (-90:90:1);
\draw[smooth] ({2*sqrt(2) + 3 + 3.5},2)--({2*sqrt(2) + 3 + 4},2);
\draw[smooth] ({2*sqrt(2) + 3 + 3.5},0)--({2*sqrt(2) + 3 + 4},0);
\node at ({2*sqrt(2) + 3 + 5},1) {$\bullet$};
\gbull{{2*sqrt(2) + 6.5}}{2}{0.14};
\gbull{{2*sqrt(2) + 6.5}}{0}{0.14};
\end{scope}
\end{tikzpicture}
$ 
which is zero.
\end{proof}

One further important relation of the boundary covector is the \textit{fish equation}:
\begin{align}
\label{fish}
\bra{K}_{a\b{a}} R_{a\b{a}}(x^2)
=
\frac{x^2-t}{1-x^2}
\bra{K}_{\b{a}a},
\qquad
\left(
\begin{tikzpicture}[scale=0.4,baseline=2cm]
\def\dy{4}
\begin{scope}[xscale=-1]
\draw[smooth] (-0.5,0+\dy)--(0,0+\dy);
\draw[smooth] (-0.5,2+\dy)--(0,2+\dy);
\draw[smooth] (0,2+\dy) arc (90:45:1);
\draw[smooth] (0,0+\dy) arc (-90:-45:1);
\draw[smooth] ({sqrt(2)/2},{1+sqrt(2)/2+\dy})--({3*sqrt(2)/2},{1-sqrt(2)/2+\dy});
\draw[smooth] ({sqrt(2)/2},{1-sqrt(2)/2+\dy})--({3*sqrt(2)/2},{1+sqrt(2)/2+\dy});
\draw[smooth] ({2*sqrt(2)},2+\dy) arc (90:135:1);
\draw[smooth] ({2*sqrt(2)},0+\dy) arc (-90:-135:1);
\draw[smooth] ({2*sqrt(2)},2+\dy) arc (90:-90:1);
\node at ({2*sqrt(2)+1},1+\dy) {$\bullet$};
\node[label={left: \fs \red $\b{x}$}] at (3, 0+\dy) {};
\node[label={left: \fs \red $x$}] at (3, 2+\dy) {};
\end{scope}
\end{tikzpicture}
\quad
=
\frac{x^2-t}{1-x^2}\ \ 
\times
\quad
\begin{tikzpicture}[scale=0.4,baseline=2cm]
\def\dy{4}
\begin{scope}[xscale=-1]
\draw[smooth] ({2*sqrt(2) + 3 + 4},2+\dy) arc (90:-90:1);
\draw[smooth] ({2*sqrt(2) + 3 + 2.5},2+\dy)--({2*sqrt(2) + 3 + 4},2+\dy);
\draw[smooth] ({2*sqrt(2) + 3 + 2.5},0+\dy)--({2*sqrt(2) + 3 + 4},0+\dy);
\node at ({2*sqrt(2) + 3 + 5},1+\dy) {$\bullet$};
\node[label={left: \fs \red $x$}] at (10, 0+\dy) {};
\node[label={left: \fs \red $\b{x}$}] at (10, 2+\dy) {};
\end{scope}
\end{tikzpicture}
\right).
\end{align}

\subsection{Boundary $B$ operator}

Yet another type of $L$ matrix will be used in this work. We call it a {\it boundary} operator or $B$ matrix, since it will always appear in contexts where a boundary covector is also present, but irrespective of its name it should simply be considered as another solution of the intertwining equation. It is given by
\begin{align*}
B_a^{(\gamma)}(x)
=
\begin{pmatrix}
t^{\mathcal{N}} & \gamma x \phid t^{\mathcal{N}}
\\
\phi & 1 - \gamma x t^{\mathcal{N}}
\end{pmatrix}_a
=
\begin{tikzpicture}[baseline=0]
\node at (-0.7,0) {$\ss \red x$};
\bdark{-0.5}{-0.5};
\node at (-0.25,0) {$\ss \red \gamma$};
\end{tikzpicture}
\end{align*}
where $\mathcal{N}$ is the particle-number operator, $[\mathcal{N},\phi] = -\phi$, $[\mathcal{N},\phid] = \phid$, and satisfies the intertwining equation
\begin{align}
\label{boundary-rll}
R_{ab}(x/y)
B_a^{(\gamma)}(x)
B_b^{(\gamma)}(y)
=
B_b^{(\gamma)}(y)
B_a^{(\gamma)}(x)
R_{ab}(x/y),
\quad
\left(
\begin{tikzpicture}[baseline=0,scale=0.7]
\draw[thick] (0,0.5) -- (1,-0.5);
\draw[thick] (0,-0.5) -- (1,0.5);
\bdark{1}{-1};\bdark{1}{0};
\node at (-0.3,0.5) {$\ss \red x$};
\node at (-0.3,-0.5) {$\ss \red y$};
\node at (1.25,0.5) {$\ss \red \gamma$};
\node at (1.25,-0.5) {$\ss \red \gamma$};
\end{tikzpicture}
\ \  
=
\
\begin{tikzpicture}[baseline=-0,scale=0.7]
\draw[thick] (0.5,0.5) -- (1.5,-0.5);
\draw[thick] (0.5,-0.5) -- (1.5,0.5);
\bdark{0}{-1};\bdark{0}{0};
\node at (-0.2,0.5) {$\ss \red x$};
\node at (-0.2,-0.5) {$\ss \red y$};
\node at (0.25,0.5) {$\ss \red \gamma$};
\node at (0.25,-0.5) {$\ss \red \gamma$};
\end{tikzpicture}
\right).
\end{align}  
Analogously to above, we also define a $B^{*}$ matrix, whose spectral parameter is inverted:
\begin{align*}
B^{*(\gamma)}_a(x)
=
\begin{pmatrix}
x t^{\mathcal{N}} & \gamma \phid t^{\mathcal{N}}
\\
x \phi & x - \gamma t^{\mathcal{N}}
\end{pmatrix}_a
=
\begin{tikzpicture}[baseline=0]
\node at (-0.7,0) {$\ss \red \b{x}$};
\blight{-0.5}{-0.5};
\node at (-0.25,0) {$\ss \red \gamma$};
\end{tikzpicture}
\end{align*}
which satisfies the intertwining equations
\begin{align*}
R_{ab}(xy) B^{(\gamma)}_a(x) B^{*(\gamma)}_b(y)
&=
B^{*(\gamma)}_b(y) B^{(\gamma)}_a(x) R_{ab}(xy),
\quad
\left(
\begin{tikzpicture}[baseline=0,scale=0.7]
\draw[thick] (0,0.5) -- (1,-0.5);
\draw[thick] (0,-0.5) -- (1,0.5);
\bdark{1}{-1};\blight{1}{0};
\node at (-0.3,0.5) {$\ss \red x$};
\node at (-0.3,-0.5) {$\ss \red \b{y}$};
\node at (1.25,0.5) {$\ss \red \gamma$};
\node at (1.25,-0.5) {$\ss \red \gamma$};
\end{tikzpicture}
\ \  
=
\
\begin{tikzpicture}[baseline=-0,scale=0.7]
\draw[thick] (0.5,0.5) -- (1.5,-0.5);
\draw[thick] (0.5,-0.5) -- (1.5,0.5);
\blight{0}{-1};\bdark{0}{0};
\node at (-0.2,0.5) {$\ss \red x$};
\node at (-0.2,-0.5) {$\ss \red \b{y}$};
\node at (0.25,0.5) {$\ss \red \gamma$};
\node at (0.25,-0.5) {$\ss \red \gamma$};
\end{tikzpicture}
\right),
\\
R_{ab}(y/x) B^{*(\gamma)}_a(x) B^{*(\gamma)}_b(y)
&=
B^{*(\gamma)}_b(y) B^{*(\gamma)}_a(x) R_{ab}(y/x),
\quad
\left(
\begin{tikzpicture}[baseline=0,scale=0.7]
\draw[thick] (0,0.5) -- (1,-0.5);
\draw[thick] (0,-0.5) -- (1,0.5);
\blight{1}{-1};\blight{1}{0};
\node at (-0.3,0.5) {$\ss \red \b{x}$};
\node at (-0.3,-0.5) {$\ss \red \b{y}$};
\node at (1.25,0.5) {$\ss \red \gamma$};
\node at (1.25,-0.5) {$\ss \red \gamma$};
\end{tikzpicture}
\ \  
=
\
\begin{tikzpicture}[baseline=-0,scale=0.7]
\draw[thick] (0.5,0.5) -- (1.5,-0.5);
\draw[thick] (0.5,-0.5) -- (1.5,0.5);
\blight{0}{-1};\blight{0}{0};
\node at (-0.2,0.5) {$\ss \red \b{x}$};
\node at (-0.2,-0.5) {$\ss \red \b{y}$};
\node at (0.25,0.5) {$\ss \red \gamma$};
\node at (0.25,-0.5) {$\ss \red \gamma$};
\end{tikzpicture}
\right).
\end{align*}
One can combine the boundary covector and boundary operators in the following way\footnote{We use negative indices to label the physical space of boundary $L$ matrices, to distinguish them from ordinary $L$ matrices, which act on sites labelled by non-negative integers.}:
\begin{align}
\label{extended-boundary1}
\bra{K(x;\gamma,\delta)}_{a\b{a}}
&=
\bra{K}_{a\b{a}}
B_{\b{a}}^{(\gamma)}(\b{x},\phi_{-2},\phid_{-2})
B_{\b{a}}^{(\delta)}(\b{x},\phi_{-1},\phid_{-1})
B_{a}^{(\gamma)}(x,\phi_{-2},\phid_{-2})
B_{a}^{(\delta)}(x,\phi_{-1},\phid_{-1})
\\
&=
\begin{tikzpicture}[scale=0.8,baseline=0.5cm]
\begin{scope}[xscale=-1,yscale=1]
\draw[smooth] (1,0) arc (-90:90:0.5);
\draw[smooth] (0.5,1)--(1,1);
\draw[smooth] (0.5,0)--(1,0);
\node at (1.5,0.5) {$\bullet$};
\bdark{-0.5}{-0.5};\bdark{0}{-0.5};
\bdark{-0.5}{0.5};\bdark{0}{0.5};
\node at (1.5,0) {$\ss \red \b{x}$};
\node at (1.5,1) {$\ss \red x$};
\node at (-0.27,0) {$\ss \red \delta$};
\node at (0.23,0) {$\ss \red \gamma$};
\node at (-0.27,1) {$\ss \red \delta$};
\node at (0.23,1) {$\ss \red \gamma$};
\end{scope}
\end{tikzpicture}
\\
\label{extended-boundary2}
\bra{K^{*}(x;\gamma,\delta)}_{a\b{a}}
&=
\bra{K}_{a\b{a}}
B^{*(\gamma)}_{\b{a}}(\b{x},\phi_{-2},\phid_{-2})
B^{*(\delta)}_{\b{a}}(\b{x},\phi_{-1},\phid_{-1})
B^{*(\gamma)}_{a}(x,\phi_{-2},\phid_{-2})
B^{*(\delta)}_{a}(x,\phi_{-1},\phid_{-1})
\\
&=
\begin{tikzpicture}[scale=0.8,baseline=0.5cm]
\begin{scope}[xscale=-1,yscale=1]
\draw[smooth] (1,0) arc (-90:90:0.5);
\draw[smooth] (0.5,1)--(1,1);
\draw[smooth] (0.5,0)--(1,0);
\node at (1.5,0.5) {$\bullet$};
\blight{-0.5}{-0.5};\blight{0}{-0.5};
\blight{-0.5}{0.5};\blight{0}{0.5};
\node at (1.5,0) {$\ss \red x$};
\node at (1.5,1) {$\ss \red \b{x}$};
\node at (-0.27,0) {$\ss \red \delta$};
\node at (0.23,0) {$\ss \red \gamma$};
\node at (-0.27,1) {$\ss \red \delta$};
\node at (0.23,1) {$\ss \red \gamma$};
\end{scope}
\end{tikzpicture}
\end{align}
with both $\bra{K(x;\gamma,\delta)}_{a\b{a}}, \bra{K^{*}(x;\gamma,\delta)}_{a\b{a}} \in 
W_a^{*} \otimes W_{\b{a}}^{*} \otimes {\rm End}(\widehat{V})$, where we abbreviate the tensor product of boundary bosonic spaces by $\widehat{V} = V_{-2} \otimes V_{-1}$.


\subsection{Double-row transfer matrices}

We define double-row transfer matrices as follows:
\begin{align*}
\mathbb{T}_{a\b{a}}(x;\gamma,\delta)
=
\bra{K(x;\gamma,\delta)}_{a \b{a}}
\prod_{i=0}^{\infty}
L_{a}(\b{x},\phi_i,\phid_i)
\prod_{j=0}^{\infty}
L_{a}(x,\phi_j,\phid_j)
&=
\begin{tikzpicture}[scale=0.8,baseline=0.5cm]
\begin{scope}[xscale=-1,yscale=1]
\draw[smooth] (1,0) arc (-90:90:0.5);
\draw[smooth] (0.5,1)--(1,1);
\draw[smooth] (0.5,0)--(1,0);
\node at (1.5,0.5) {$\bullet$};
\bdark{-0.5}{-0.5};\bdark{0}{-0.5};
\bdark{-0.5}{0.5};\bdark{0}{0.5};
\foreach\x in {1,...,5}{
\dark{-\x-0.5}{0.5};\dark{-\x-0.5}{-0.5};
}
\node at (-6,1) {$\cdots$};
\node at (-6,0) {$\cdots$};
\node at (1.5,0) {$\ss \red \b{x}$};
\node at (1.5,1) {$\ss \red x$};
\node at (-0.27,0) {$\ss \red \delta$};
\node at (0.23,0) {$\ss \red \gamma$};
\node at (-0.27,1) {$\ss \red \delta$};
\node at (0.23,1) {$\ss \red \gamma$};
\end{scope}
\end{tikzpicture}
\\
\mathbb{T}_{a\b{a}}^{*}(x;\gamma,\delta)
=
\bra{K^{*}(x;\gamma,\delta)}_{a \b{a}}
\prod_{i=0}^{\infty}
L^{*}_{a}(\b{x},\phi_i,\phid_i)
\prod_{j=0}^{\infty}
L^{*}_{a}(x,\phi_j,\phid_j)
&=
\begin{tikzpicture}[scale=0.8,baseline=0.5cm]
\begin{scope}[xscale=-1,yscale=1]
\draw[smooth] (1,0) arc (-90:90:0.5);
\draw[smooth] (0.5,1)--(1,1);
\draw[smooth] (0.5,0)--(1,0);
\node at (1.5,0.5) {$\bullet$};
\blight{-0.5}{-0.5};\blight{0}{-0.5};
\blight{-0.5}{0.5};\blight{0}{0.5};
\foreach\x in {1,...,5}{
\light{-\x-0.5}{0.5};\light{-\x-0.5}{-0.5};
}
\node at (-6,1) {$\cdots$};
\node at (-6,0) {$\cdots$};
\node at (1.5,0) {$\ss \red x$};
\node at (1.5,1) {$\ss \red \b{x}$};
\node at (-0.27,0) {$\ss \red \delta$};
\node at (0.23,0) {$\ss \red \gamma$};
\node at (-0.27,1) {$\ss \red \delta$};
\node at (0.23,1) {$\ss \red \gamma$};
\end{scope}
\end{tikzpicture}
\end{align*}
which are both operators in $W_a^{*} \otimes W_{\b{a}}^{*} \otimes {\rm End} (\widehat{V} \otimes V)$. The name ``double-row transfer matrix'' is somewhat misleading, and it should be borne in mind that these are really 
operator-valued covectors in $W_a^{*} \otimes W_{\b{a}}^{*}$. We shall make use of the following components:
\begin{align*}
\mathbb{T}_{--}(x;\gamma,\delta)
&=
\begin{tikzpicture}[scale=0.8,baseline=0.5cm]
\begin{scope}[xscale=-1,yscale=1]
\draw[smooth] (1,0) arc (-90:90:0.5);
\draw[smooth] (0.5,1)--(1,1);
\draw[smooth] (0.5,0)--(1,0);
\node at (1.5,0.5) {$\bullet$};
\bdark{-0.5}{-0.5};\bdark{0}{-0.5};
\bdark{-0.5}{0.5};\bdark{0}{0.5};
\foreach\x in {1,...,5}{
\dark{-\x-0.5}{0.5};\dark{-\x-0.5}{-0.5};
}
\node at (-6,1) {$\cdots$};
\node at (-6,0) {$\cdots$};
\node at (1.5,0) {$\ss \red \b{x}$};
\node at (1.5,1) {$\ss \red x$};
\node at (-0.27,0) {$\ss \red \delta$};
\node at (0.23,0) {$\ss \red \gamma$};
\node at (-0.27,1) {$\ss \red \delta$};
\node at (0.23,1) {$\ss \red \gamma$};
\gbull{-5.5}{0}{0.07};
\gbull{-5.5}{1}{0.07};
\end{scope}
\end{tikzpicture}
\\
\mathbb{T}_{++}^{*}(x;\gamma,\delta)
&=
\begin{tikzpicture}[scale=0.8,baseline=0.5cm]
\begin{scope}[xscale=-1,yscale=1]
\draw[smooth] (1,0) arc (-90:90:0.5);
\draw[smooth] (0.5,1)--(1,1);
\draw[smooth] (0.5,0)--(1,0);
\node at (1.5,0.5) {$\bullet$};
\blight{-0.5}{-0.5};\blight{0}{-0.5};
\blight{-0.5}{0.5};\blight{0}{0.5};
\foreach\x in {1,...,5}{
\light{-\x-0.5}{0.5};\light{-\x-0.5}{-0.5};
}
\node at (-6,1) {$\cdots$};
\node at (-6,0) {$\cdots$};
\node at (1.5,0) {$\ss \red x$};
\node at (1.5,1) {$\ss \red \b{x}$};
\node at (-0.27,0) {$\ss \red \delta$};
\node at (0.23,0) {$\ss \red \gamma$};
\node at (-0.27,1) {$\ss \red \delta$};
\node at (0.23,1) {$\ss \red \gamma$};
\ebull{-5.5}{0}{0.07};
\ebull{-5.5}{1}{0.07};
\end{scope}
\end{tikzpicture}
\end{align*}
which are to be understood as configurations which have only the tiles \begin{tikzpicture}[scale=0.5,baseline=-0.1cm] \dark{-0.5}{-0.5};  \gbull{-0.5}{0}{0.09};\gbull{0.5}{0}{0.09}; \end{tikzpicture} and \begin{tikzpicture}[scale=0.5,baseline=-0.1cm] \light{-0.5}{-0.5};  \ebull{-0.5}{0}{0.09};\ebull{0.5}{0}{0.09}; \end{tikzpicture}, respectively, sufficiently far to the right.

\section{Hall--Littlewood polynomials}
\label{sec:HL}

\subsection{$A_n$ Hall--Littlewood polynomials: lattice paths}

An interesting connection between the $t$-boson model and Hall--Littlewood polynomials was first pointed out in \cite{tsi}. There it was shown that Hall--Littlewood polynomials are the Bethe wavefunctions (more specifically, the expansion coefficients of Bethe vectors on the basis $\ket{\lambda}$, in the absence of Bethe equations) of the $t$-boson model. This connection was also pursued in \cite{korff}, where an explicit bijection was found between $t$-weighted semi-standard Young tableaux (that evaluate Hall--Littlewood polynomials) and the lattice paths considered here. The approach adopted in \cite{korff} also gave rise to combinatorial formulae for the $q$-Whittaker functions, and cylindrical versions of both the Hall--Littlewood and $q$-Whittaker functions.

Below we state the result of \cite{tsi,korff} in terms of our notations. For the purposes of stating the theorem, the physical space will now start at the first, rather than the zeroth lattice site: $V = \bigotimes_{i=1}^{\infty} V_i$. This means that, in what follows, the row-to-row transfer matrices \eqref{T} and \eqref{T*} have their product starting at $i=1$ instead of $i=0$. We will only adopt this convention momentarily, returning to the larger lattice immediately after the theorem.

\begin{thm}
\label{thm:hl}
The Hall--Littlewood polynomials $P_{\lambda}(x;t)$ and $Q_{\lambda}(x;t)$ are given by
\begin{align}
\label{P}
P_{\lambda}(x_1,\dots,x_n;t)
&=
\bra{\lambda}
\mathcal{T}(x_n)
\dots
\mathcal{T}(x_1)
\ket{0}
\\
\label{Q}
Q_{\lambda}(x_1,\dots,x_n;t)
&=
\bra{0}
\mathcal{T}^{*}(x_1)
\dots
\mathcal{T}^{*}(x_n)
\ket{\lambda}
\end{align}
where $\ket{\lambda} = \bigotimes_{i=1}^{\infty} \ket{m_i(\lambda)}_i$, and similarly for the dual state $\bra{\lambda}$.
\end{thm}

\begin{proof}
We begin with the proof \eqref{Q}. Inserting a complete set of states before the final transfer matrix, we have the branching formula
\begin{align*}
\bra{0}
\mathcal{T}^{*}(x_1)
\dots
\mathcal{T}^{*}(x_n)
\ket{\lambda}
=
\sum_{\mu}
\bra{0}
\mathcal{T}^{*}(x_1)
\dots
\mathcal{T}^{*}(x_{n-1})
\ket{\mu}
\bra{\mu}
\mathcal{T}^{*}(x_n)
\ket{\lambda}.
\end{align*}
Comparing this with the branching rule for $Q_{\lambda}(x;t)$,
\begin{align*}
Q_{\lambda}(x_1,\dots,x_n;t)
=
\sum_{\mu}
Q_{\mu}(x_1,\dots,x_{n-1};t)
Q_{\lambda/\mu}(x_n;t),
\end{align*}
it suffices to show that $\bra{\mu} \mathcal{T}^{*}(x) \ket{\lambda} = Q_{\lambda/\mu}(x;t)$, the skew Hall--Littlewood $Q$ polynomial in a single variable. Following \cite{mac}, 
$Q_{\lambda/\mu}(x;t)$ is given explicitly by
\begin{align}
\label{Q-branch}
Q_{\lambda/\mu}(x;t)
=
\left\{
\begin{array}{ll}
x^{|\lambda|-|\mu|}
\varphi_{\lambda/\mu}(t),
&
\lambda \succeq \mu,
\\
\\
0,
&
\text{otherwise},
\end{array}
\right.
\qquad
\varphi_{\lambda/\mu}(t)
:=
\prod_{i \geq 1: m_i(\mu) + 1 = m_i(\lambda)}
(1-t^{m_i(\lambda)}),
\end{align}
where we use the notation $\lambda \succeq \mu$ for interlacing partitions, {\it i.e.} partitions which satisfy 
$\lambda_i \geq \mu_i \geq \lambda_{i+1}$ for all $i \geq 1$. Using the graphical interpretation of $\bra{\mu} \mathcal{T}^{*}(x) \ket{\lambda}$, one can easily deduce that it has the form of \eqref{Q-branch}. Since lattice paths cannot cross, it is immediate that $\bra{\mu} \mathcal{T}^{*}(x) \ket{\lambda} \neq 0$ if and only if $\lambda \succeq \mu$. We obtain a weight of $x$ for every horizontal step taken by a line, and the number of such steps is precisely the difference in the weights of the partitions $\lambda$ and $\mu$. Finally, in the transition from $\lambda$ to $\mu$, we get a weight of $(1-t^m)$ every time a particle departs from a group of $m$ particles at a certain site.

The proof of \eqref{P} is very similar. We begin in the same way, with the branching formula
\begin{align*}
\bra{\lambda}
\mathcal{T}(x_n)
\dots
\mathcal{T}(x_1)
\ket{0}
=
\sum_{\mu}
\bra{\lambda}
\mathcal{T}(x_n)
\ket{\mu}
\bra{\mu}
\mathcal{T}(x_{n-1})
\dots
\mathcal{T}(x_1)
\ket{0},
\end{align*}
and compare it with the branching rule for $P_{\lambda}(x;t)$:
\begin{align*}
P_{\lambda}(x_1,\dots,x_n;t)
=
\sum_{\mu}
P_{\lambda/\mu}(x_n;t)
P_{\mu}(x_1,\dots,x_{n-1};t),
\end{align*}
where the one-variable skew Hall--Littlewood $P$ polynomial is given by
\begin{align}
\label{P-branch}
P_{\lambda/\mu}(x;t)
=
\left\{
\begin{array}{ll}
x^{|\lambda|-|\mu|}
\psi_{\lambda/\mu}(t),
&
\lambda \succeq \mu,
\\
\\
0,
&
\text{otherwise},
\end{array}
\right.
\qquad
\psi_{\lambda/\mu}(t)
:=
\prod_{i \geq 1: m_i(\mu) = m_i(\lambda)+1}
(1-t^{m_i(\mu)}).
\end{align}
Hence we wish to show that $\bra{\lambda} \mathcal{T}(x_n) \ket{\mu}$ is given by \eqref{P-branch}. To do that, we observe that the tiles in \eqref{tiles} can be drawn instead as
\begin{equation}
\label{tiles2}
\begin{tikzpicture}[baseline=0]
\matrix[column sep={2.5cm,between origins},row sep=0.25cm]{
\dark{-0.5}{-0.5};
\draw[bbline] (0,-0.5) node[below] {$\ss m$} -- (0,0.5) node[above] {$\ss m$};
\draw[bline] (-0.5,0) -- (0.5,0);
\bbull{-0.5}{0}{0.05};\bbull{0.5}{0}{0.05};
& 
\dark{-0.5}{-0.5}; 
\draw[bbline] (0,-0.5) node[below] {$\ss m+1$} -- (0,0.5) node[above] {$\ss m$};
\draw[bline] (-0.05,-0.5) -- (-0.05,0) -- (-0.5,0);
\bbull{-0.5}{0}{0.05};\ebull{0.5}{0}{0.05};
& 
\dark{-0.5}{-0.5}; 
\draw[bbline] (0,-0.5) node[below] {$\ss m-1$} -- (0,0.5) node[above] {$\ss m$};
\draw[bline] (0.05,0.5) -- (0.05,0) -- (0.5,0);
\ebull{-0.5}{0}{0.05};\bbull{0.5}{0}{0.05};
&
\dark{-0.5}{-0.5};
\draw[bbline] (0,-0.5) node[below] {$\ss m$} -- (0,0.5) node[above] {$\ss m$};
\ebull{-0.5}{0}{0.05};\ebull{0.5}{0}{0.05};
\\
\node{$x$}; & \node{$x$}; & \node{$(1-t^m)$}; & \node{$1$}; 
\\
};
\end{tikzpicture}
\end{equation}
where we have complemented particles with holes, and vice versa, on the auxiliary (horizontal) space of each tile\footnote{Whenever we perform this complementation, we indicate it by colouring particles in blue rather than green. In other words, we make the substitution $\part \rightarrow \hole$ and $\hole \rightarrow \bpart$ for the labels of each auxiliary space.}. Using this alternative graphical interpretation of $\bra{\lambda} \mathcal{T}(x_n) \ket{\mu}$, we get as before a weight of $x$ for every horizontal step taken. In the transition from $\lambda$ to $\mu$, we get a weight of $(1-t^m)$ every time a particle joins $m-1$ others at a certain site.
\end{proof}

Let us return to the larger physical space, $V = \bigotimes_{i=0}^{\infty} V_i$. Consider partitions augmented by a zeroth site, with a non-integral shift $\alpha$:
\begin{align*}
\ket{\lambda;\alpha}
= 
\ket{m_0(\lambda)+\alpha}_0 
\otimes 
\ket{m_1(\lambda)}_1
\otimes
\ket{m_2(\lambda)}_2 
\otimes
\cdots
\end{align*}
and similarly for dual states $\bra{\lambda;\alpha}$. The following result can be easily deduced from Theorem \ref{thm:hl}.
\begin{lem}
The Hall--Littlewood polynomials admit the alternative expressions
\label{cor:hl}
\begin{align}
\label{hl-P}
\prod_{i=1}^{n}
x_i
P_{\lambda}(x_1,\dots,x_n;t)
&=
\bra{\lambda;\alpha}
T_{+}(x_n)
\dots
T_{+}(x_1)
\ket{0;\alpha}
\\
\label{hl-Q}
\prod_{j=1}^{m_0(\lambda)}
(1-t^{\alpha+j})
\prod_{i=1}^{n}
x_i
Q_{\lambda}(x_1,\dots,x_n;t)
&=
\bra{0;\alpha}
T^{*}_{-}(x_1)
\dots
T^{*}_{-}(x_n)
\ket{\lambda;\alpha}
\end{align}
\end{lem}

\begin{proof}

We start with the proof of \eqref{hl-P}. Using its graphical form \eqref{mat-form-dark}, each $T_{+}(x_i)$ operator has a hole at its left edge and a green particle at its right edge (which is at infinity). Alternatively, after complementing particles with holes and vice versa, as we did in the proof of Theorem \ref{thm:hl}, each $T_{+}(x_i)$ operator has a blue particle at its left edge and a hole at its right edge. A typical configuration of the right hand side of \eqref{hl-P} is thus
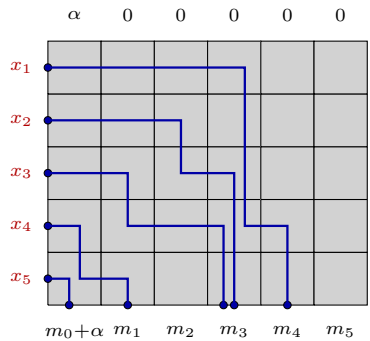
\begin{equation*}
\begin{tikzpicture}[scale=0.7]
\foreach\y in {1,...,5}{
\foreach\x in {1,...,6}{
\dark{\x}{\y}
}
\bbull{1}{\y+0.5}{0.07};
}
\bbull{1.4}{1}{0.07}; \bbull{2.5}{1}{0.07}; 
\bbull{4.3}{1}{0.07}; \bbull{4.5}{1}{0.07}; \bbull{5.5}{1}{0.07};
\draw[bline] (1,1.5) -- (1.4,1.5) -- (1.4,1);
\draw[bline] (1,2.5) -- (1.6,2.5) -- (1.6,1.5) -- (2.5,1.5) -- (2.5,1);
\draw[bline] (1,3.5) -- (2.5,3.5) -- (2.5,2.5) -- (4.3,2.5) -- (4.3,1);
\draw[bline] (1,4.5) -- (3.5,4.5) -- (3.5,3.5) -- (4.5,3.5) -- (4.5,1);
\draw[bline] (1,5.5) -- (4.7,5.5) -- (4.7,2.5) -- (5.5,2.5) -- (5.5,1);
\foreach\y in {1,...,5}{
\node at (0.5,-\y+6.5) {$\ss \red x_{\y}$};
}
\node at (1.5,6.5) {$\ss \alpha$};
\node at (2.5,6.5) {$\ss 0$};
\node at (3.5,6.5) {$\ss 0$};
\node at (4.5,6.5) {$\ss 0$};
\node at (5.5,6.5) {$\ss 0$};
\node at (6.5,6.5) {$\ss 0$};
\node at (1.5,0.5) {$\ss m_0+\alpha$};
\node at (2.5,0.5) {$\ss m_1$};
\node at (3.5,0.5) {$\ss m_2$};
\node at (4.5,0.5) {$\ss m_3$};
\node at (5.5,0.5) {$\ss m_4$};
\node at (6.5,0.5) {$\ss m_5$};
\end{tikzpicture}
\end{equation*}
which is shown here for $\lambda = (4,3,3,1,0)$. Consulting the Boltzmann weights \eqref{tiles2}, we see that in general $\left(\substack{n \\ \\ m_0(\lambda)}\right)$ configurations of the zeroth column are possible (corresponding with choosing the positions of the holes on the left edge of the first column), and every such configuration has the overall weight $\prod_{i=1}^{n} x_i$. Hence we can strip away the zeroth column at the expense of the factor $\prod_{i=1}^{n} x_i$ and restrict the lattice to $\bigotimes_{i=1}^{\infty} V_i$, with free boundary conditions at the left edge, which is precisely equation \eqref{P} for $P_{\lambda}(x_1,\dots,x_n;t)$.

The proof of \eqref{hl-Q} is completely analogous. From \eqref{mat-form-light}, each $T^{*}_{-}(x_i)$ operator has a green particle at its left edge and a hole at its right edge (situated at infinity). A typical configuration of the right hand side of \eqref{hl-Q} is thus
\begin{equation*}
\begin{tikzpicture}[scale=0.7]
\foreach\y in {1,...,5}{
\foreach\x in {1,...,6}{
\light{\x}{\y}
}
\gbull{1}{\y+0.5}{0.07};
}
\gbull{1.3}{6}{0.07}; \gbull{1.5}{6}{0.07};
\gbull{3.5}{6}{0.07}; \gbull{4.4}{6}{0.07}; \gbull{5.5}{6}{0.07};
\draw[gline] (1,5.5) -- (1.3,5.5) -- (1.3,6);
\draw[gline] (1,4.5) -- (1.5,4.5) -- (1.5,6);
\draw[gline] (1,3.5) -- (1.7,3.5) -- (1.7,4.5) -- (2.5,4.5) -- (2.5,5.5) -- (3.5,5.5) -- (3.5,6);
\draw[gline] (1,2.5) -- (3.5,2.5) -- (3.5,3.5) -- (4.4,3.5) -- (4.4,6);
\draw[gline] (1,1.5) -- (4.6,1.5) -- (4.6,4.5) -- (5.5,4.5) -- (5.5,6);
\foreach\y in {1,...,5}{
\node at (0.5,\y+0.5) {$\ss \red \b{x}_{\y}$};
}
\node at (1.5,6.5) {$\ss m_0+\alpha$};
\node at (2.5,6.5) {$\ss m_1$};
\node at (3.5,6.5) {$\ss m_2$};
\node at (4.5,6.5) {$\ss m_3$};
\node at (5.5,6.5) {$\ss m_4$};
\node at (6.5,6.5) {$\ss m_5$};
\node at (1.5,0.5) {$\ss \alpha$};
\node at (2.5,0.5) {$\ss 0$};
\node at (3.5,0.5) {$\ss 0$};
\node at (4.5,0.5) {$\ss 0$};
\node at (5.5,0.5) {$\ss 0$};
\node at (6.5,0.5) {$\ss 0$};
\end{tikzpicture}
\end{equation*}
which is shown here for $\lambda = (4,3,2,0,0)$. As before, $\left(\substack{n \\ \\ m_0(\lambda)}\right)$ configurations of the zeroth column are possible, and in all configurations we have a group of $m_0(\lambda)$ particles starting at the top of the column and branching away from this group one by one to the left edge of the lattice. From the Boltzmann weights \eqref{tiles*}, we see that every configuration of this column gives rise to the total weight $\prod_{j=1}^{m_0(\lambda)} (1-t^{\alpha+j}) \prod_{i=1}^{n} x_i$, and can otherwise be neglected. Restricting the lattice to $\bigotimes_{i=1}^{\infty} V_i$, with a free boundary at the left edge, from \eqref{Q} we obtain $Q_{\lambda}(x_1,\dots,x_n;t)$.

\end{proof}

\subsection{$A_n$ Hall--Littlewood polynomials: sum over symmetric group}
\label{sec:a_n-sym}

Hall--Littlewood polynomials were originally defined, not combinatorially as in the previous subsection, but as a sum over the symmetric group \cite{mac}. Our next aim is to show that such summation formulae can also be recovered very naturally in the framework of the $t$-boson model. We begin by defining the column-to-column transfer matrix:
\begin{align*}
S^{[i]}(x_1,\dots,x_n)
=
L_{n}(x_n,\phi_i,\phid_i)
\dots
L_{1}(x_1,\phi_i,\phid_i)
\in
{\rm End}(W_{1} \otimes \cdots \otimes W_{n} \otimes V_i),
\end{align*}
which is a $2^n \times 2^n$ matrix formed by taking the tensor product of $n$ $L$ matrices with different auxiliary spaces, whose entries act non-trivially in $V_i$. Consider the components of $S^{[i]}(x_1,\dots,x_n)$ in the bosonic space $V_i$:
\begin{align*}
S^{[l,m]}(x_1,\dots,x_n)
\equiv
\bra{l}_i
L_{n}(x_n,\phi_i,\phid_i)
\dots
L_{1}(x_1,\phi_i,\phid_i)
\ket{m}_i
=
\begin{tikzpicture}[scale=0.5,baseline=(mid.base)]
\node (mid) at (1,3) {};
\node at (0.5,1.5) {$\ss \red x_{n}$};
\node at (0.5,3.3) {$\red \vdots$};
\node at (0.5,4.5) {$\ss \red x_{1}$};
\foreach\y in {1,...,4}{
\dark{1}{\y};
}
\node at (1.5,0.5) {$\ss l$};
\node at (1.5,5.5) {$\ss m$};
\end{tikzpicture}\ \ .
\end{align*}
The Hall--Littlewood polynomials can be expressed in terms of these, as follows:
\begin{align}
\label{P-columns1}
\prod_{i=1}^{n} x_i
P_{\lambda}(x_1,\dots,x_n;t)
=
\bra{\hole,\dots,\hole}
\prod_{i=0}^{\infty}
S^{[m_i(\lambda),0]}(x_1,\dots,x_n)
\ket{\part,\dots,\part},
\end{align}
where the product is ordered from left to right as the index $i$ increases, and where we use the shorthand
\begin{align*}
\bra{\hole,\dots,\hole}
=
\bigotimes_{i=1}^{n} \begin{pmatrix} 1 & 0 \end{pmatrix}_i
\quad
{\rm and}
\quad
\ket{\part,\dots,\part}
=
\bigotimes_{i=1}^{n} \begin{pmatrix} 0 \\ 1 \end{pmatrix}_i
\end{align*} 
to encode the states at the left and right edges of the lattice. Indeed this is simply a re-writing of equation \eqref{hl-P}, in which we decompose the lattice not in terms of rows, but in terms of columns. We now study the operators $S^{[i]}(x_1,\dots,x_n)$ under conjugation by $F$ matrices \cite{ms}, and refer the reader to Appendix \ref{app:fbasis} for more details on this change of basis. We are interested in the twisted transfer column-to-column matrix, given by
\begin{align*}
\widetilde{S}^{[i]}_{1\dots n}(x_1,\dots,x_n)
=
F_{1\dots n}(x_1,\dots,x_n)
S^{[i]}_{1 \dots n}(x_1,\dots,x_n)
F^{-1}_{1\dots n}(x_1,\dots,x_n),
\end{align*}
which has the following symmetry property
\begin{align}
\label{symmetry}
\widetilde{S}^{[i]}_{1\dots n}(x_1,\dots,x_n)
=
\widetilde{S}^{[i]}_{\sigma(1) \dots \sigma(n)}
\left(x_{\sigma(1)},\dots,x_{\sigma(n)}\right)
\end{align}
for any permutation $\sigma \in S_n$. Given that\footnote{The identities \eqref{eigenvector} are well-known properties of factorizing $F$ matrices, that follow directly from the expressions \eqref{F} and \eqref{F-inv} for $F$ and its inverse, as well as the eigenvector equations
\begin{align*}
\begin{pmatrix} 1 & 0 \end{pmatrix}_a
\otimes
\begin{pmatrix} 1 & 0 \end{pmatrix}_b
\mathcal{R}_{ab}(x/y)
=
\begin{pmatrix} 1 & 0 \end{pmatrix}_a
\otimes
\begin{pmatrix} 1 & 0 \end{pmatrix}_b,
\quad
\mathcal{R}_{ab}(x/y)
\begin{pmatrix} 0 \\ 1 \end{pmatrix}_a
\otimes
\begin{pmatrix} 0 \\ 1 \end{pmatrix}_b
=
\begin{pmatrix} 0 \\ 1 \end{pmatrix}_a
\otimes
\begin{pmatrix} 0 \\ 1 \end{pmatrix}_b,
\end{align*}
which can easily be verified as properties of $\mathcal{R}$ matrix \eqref{renorm-R}. 
}
\begin{align}
\label{eigenvector}
\bra{\hole,\dots,\hole}
F(x_1,\dots,x_n)
=
\bra{\hole,\dots,\hole},
\qquad
F^{-1}(x_1,\dots,x_n)
\ket{\part,\dots,\part}
=
\ket{\part,\dots,\part},
\end{align}
we can clearly conjugate each matrix appearing in \eqref{P-columns1} by $F$, to obtain
\begin{align}
\label{P-columns2}
\prod_{i=1}^{n} x_i
P_{\lambda}(x_1,\dots,x_n;t)
=
\bra{\hole,\dots,\hole}
\prod_{i=0}^{\infty}
\widetilde{S}^{[m_i(\lambda),0]}(x_1,\dots,x_n)
\ket{\part,\dots,\part}.
\end{align}
The motivation for this change of basis is that, in view of the symmetry \eqref{symmetry}, one expects to be able to write explicit formulae for each $\widetilde{S}^{[m,0]}(x_1,\dots,x_n)$ and to calculate the right hand side of 
\eqref{P-columns2} directly. To proceed in that direction, we have the following result.

\begin{thm}
\label{thm:twisting}
For all $0 \leq m \leq n$, one has
\begin{align}
\label{twist-explicit}
\widetilde{S}^{[m,0]}(x_1,\dots,x_n)
=
\sum_{\substack{I \subset \{1,\ldots,n\} \\ |I|=m}}
\bigotimes_{i \in I}
\begin{pmatrix} 
0 & x_i
\\ 
0 & 0
\end{pmatrix}_i
\ \ 
\bigotimes_{j \not\in I}
\begin{pmatrix} 
x_j
\prod_{k \in I} \left( \frac{x_j-t x_k}{x_j-x_k} \right) & 0
\\ 
0 & 1
\end{pmatrix}_j
\end{align}
where the sum is taken over all subsets $I$ of $\{1,\dots,n\}$ of cardinality $m$.
\end{thm}

\begin{proof}
Using the results of Appendix \ref{app:fbasis}, we can write the components of $\widetilde{S}^{[m,0]}(x_1,\dots,x_n)$ as follows:
\begin{align}
\Big[
\widetilde{S}^{[m,0]}_{1\dots n}(x_1,\dots,x_n)
\Big]_{i_1\dots i_n}^{j_1\dots j_n}
&=
\Big[
\mathcal{R}^{1\dots n}_{\sigma}
S^{[m,0]}_{1\dots n}(x_1,\dots,x_n)
\mathcal{R}^{\rho}_{1\dots n}
\prod_{1\leq k<l \leq n}
\Delta_{kl}^{-1}(x_k,x_l)
\Big]_{i_1\dots i_n}^{j_1\dots j_n}
\nonumber
\\
&=
\Big[
S^{[m,0]}_{\sigma(1) \dots \sigma(n)}\left(x_{\sigma(1)},\dots,x_{\sigma(n)}\right)
\mathcal{R}^{1\dots n}_{\sigma}
\mathcal{R}^{\rho}_{1\dots n}
\prod_{1\leq k<l \leq n}
\Delta_{kl}^{-1}(x_k,x_l)
\Big]_{i_1\dots i_n}^{j_1\dots j_n}
\nonumber
\\
&=
\Big[
S^{[m,0]}_{\sigma(1) \dots \sigma(n)}\left(x_{\sigma(1)},\dots,x_{\sigma(n)}\right)
\mathcal{R}^{1\dots n}_{\sigma}
\mathcal{R}^{\rho}_{1\dots n}
\Big]_{i_1\dots i_n}^{j_1\dots j_n}
\prod_{1\leq k<l \leq n}
b^{-1}_{j_k,j_l}(x_k,x_l)
\label{...1}
\end{align}
where the permutations $\sigma$ and $\rho$ are such that $i_{\sigma(1)} \geq \cdots \geq i_{\sigma(n)}$ and $j_{\rho(1)} \leq \cdots \leq j_{\rho(n)}$. To complete the calculation, we remark that graphically\footnote{We will use the right hand side of \eqref{running-example} as a sufficiently general running example, preferring this to a rigorous algebraic proof, which would be unnecessarily tedious.} 
\begin{align}
\label{running-example}
\Big[
S^{[m,0]}_{\sigma(1) \dots \sigma(n)}\left(x_{\sigma(1)},\dots,x_{\sigma(n)}\right)
\mathcal{R}^{1\dots n}_{\sigma}
\mathcal{R}^{\rho}_{1\dots n}
\Big]_{i_1\dots i_n}^{j_1\dots j_n}
=
\begin{tikzpicture}[scale=0.6,baseline=(mid.base)]
\node (mid) at (1,4) {};
\draw[fline] plot [smooth] coordinates {(2,6.5) (2.5,6.5) (5.5,5.5) (7.5,3.5) (8.5,3.5)};
\draw[fline] plot [smooth] coordinates {(2,5.5) (2.5,5.5) (5.5,4.5) (7.5,2.5) (8.5,2.5)};
\draw[fline] plot [smooth] coordinates {(2,4.5) (2.5,4.5) (5.5,6.5) (7.5,4.5) (8.5,4.5)};
\draw[fline] plot [smooth] coordinates {(2,3.5) (2.5,3.5) (5.5,2.5) (7.5,5.5) (8.5,5.5)};
\draw[fline] plot [smooth] coordinates {(2,2.5) (2.5,2.5) (5.5,3.5) (7.5,6.5) (8.5,6.5)};
\draw[fline] (2,1.5) -- (8.5,1.5);
\draw[dotted] (5.5,1) -- (5.5,7);
\foreach\y in {1,...,6}{
\dark{1}{\y};
}
\node at (1.5,0.5) {$\ss m$};
\node at (1.5,7.5) {$\ss 0$};
\foreach\x in {1,...,3}{
\ebull{1}{0.5+\x}{0.09};
\gbull{1}{3.5+\x}{0.09};
}
\foreach\x in {1,...,4}{
\gbull{8.5}{0.5+\x}{0.09};
}
\foreach\x in {5,...,6}{
\ebull{8.5}{0.5+\x}{0.09};
}
\node at (0.3,6.5) {$\ss \red x_{\sigma(1)} $};
\node at (0.3,4.2) {$\red \vdots$};
\node at (0.3,1.5) {$\ss \red x_{\sigma(n)} $};
\node at (4,1) {$\ss \red R^{1\dots n}_{\sigma}$};
\node at (7,1) {$\ss \red R^{\rho}_{1\dots n}$};
\end{tikzpicture}
\end{align}
in which the states on the left edge of the column are ordered, with holes occupying the lowest edges followed by particles occupying the highest ones. This follows from the fact that 
$i_{\sigma(1)} \geq \cdots \geq i_{\sigma(n)}$. Similarly, at the right edge of the product of $\mathcal{R}$ matrices, the lowest edges are occupied by particles, while the highest ones are empty. This comes from the fact that $j_{\rho(1)} \leq \cdots \leq j_{\rho(n)}$.
Since the occupation number at the top of the column is 0, the tiles with a particle at their left edge freeze, and we obtain
\begin{align}
\label{...2}
\Big[
S^{[m,0]}_{\sigma(1) \dots \sigma(n)}\left(x_{\sigma(1)},\dots,x_{\sigma(n)}\right)
\mathcal{R}^{1\dots n}_{\sigma}
\mathcal{R}^{\rho}_{1\dots n}
\Big]_{i_1\dots i_n}^{j_1\dots j_n}
=
\begin{tikzpicture}[scale=0.6,baseline=(mid.base)]
\node (mid) at (1,4) {};
\draw[fline] plot [smooth] coordinates {(3.5,6.5) (5.5,5.5) (7.5,3.5) (8.5,3.5)};
\draw[fline] plot [smooth] coordinates {(3.5,5.5) (5.5,4.5) (7.5,2.5) (8.5,2.5)};
\draw[fline] plot [smooth] coordinates {(3.5,4.5) (5.5,6.5) (7.5,4.5) (8.5,4.5)};
\draw[fline] plot [smooth] coordinates {(2,3.5) (2.5,3.5) (5.5,2.5) (7.5,5.5) (8.5,5.5)};
\draw[fline] plot [smooth] coordinates {(2,2.5) (2.5,2.5) (5.5,3.5) (7.5,6.5) (8.5,6.5)};
\draw[fline] (2,1.5) -- (8.5,1.5);
\draw[dotted] (5.5,1) -- (5.5,7);
\foreach\y in {1,...,3}{
\dark{1}{\y};
}
\node at (1.5,0.5) {$\ss m$};
\node at (1.5,4.5) {$\ss 0$};
\foreach\x in {1,...,3}{
\ebull{1}{0.5+\x}{0.09};
\gbull{3.5}{3.5+\x}{0.09};
}
\foreach\x in {1,...,4}{
\gbull{8.5}{0.5+\x}{0.09};
}
\foreach\x in {5,...,6}{
\ebull{8.5}{0.5+\x}{0.09};
}
\node at (2,5.5) {$\ss \red x_k: i_k = \part $};
\draw[dred] (3.2,4) -- (3,4) -- (3,7) -- (3.2,7);
\node at (-0.5,2.5) {$\ss \red x_l: i_l = \hole $};
\draw[dred] (0.7,1) -- (0.5,1) -- (0.5,4) -- (0.7,4);
\node at (10,6) {$\ss \red x_l: j_l = \hole$};
\draw[dred] (8.8,5.2) -- (9,5.2) -- (9,7) -- (8.8,7);
\node at (10,3) {$\ss \red x_k: j_k = \part$};
\draw[dred] (8.8,1) -- (9,1) -- (9,4.8) -- (8.8,4.8);
\node at (4,1) {$\ss \red R^{1\dots n}_{\sigma}$};
\node at (7,1) {$\ss \red R^{\rho}_{1\dots n}$};
\end{tikzpicture}
\end{align}
where we have indicated groups of lines/rapidities with a common edge state. Studying the right hand side of \eqref{...2} and using particle-conservation arguments, it is not hard to conclude that it vanishes if $i_k = \part$, $j_k=\hole$ for some $1\leq k \leq n$ (such a situation would give rise to a vertex which is not among the six in \eqref{six-v}). Furthermore, all non-vanishing components have the common factor $\prod_{k: i_k = \part} \prod_{l: j_l = \hole} (x_l-x_k)/(x_l-t x_k)$, due to the fact that 
\begin{align*}
\begin{tikzpicture}[scale=0.6,baseline=-1]
\draw[thick, smooth] (1,0) -- (-1,0);
\node[label={left: \fs \red $x$}] at (-1,0) {};
\draw[thick, smooth] (0,1) -- (0,-1);
\node[label={below: \fs \red $y$}] at (0,-1) {};
\gbull{-0.5}{0}{0.09}; \gbull{0.5}{0}{0.09};
\gbull{0}{-0.5}{0.09}; \gbull{0}{0.5}{0.09};
\end{tikzpicture}
=
1,
\qquad
\begin{tikzpicture}[scale=0.6,baseline=-1]
\draw[thick, smooth] (1,0) -- (-1,0);
\node[label={left: \fs \red $x$}] at (-1,0) {};
\draw[thick, smooth] (0,-1) -- (0,1);
\node[label={below: \fs \red $y$}] at (0,-1) {};
\gbull{-0.5}{0}{0.09}; \gbull{0.5}{0}{0.09};
\ebull{0}{-0.5}{0.09}; \ebull{0}{0.5}{0.09};
\end{tikzpicture}
=
\frac{1-x/y}{1-t x/y},
\end{align*}
in the normalization of the $\mathcal{R}$ matrix. Eliminating the vertices which give rise to this overall factor, we find that
\begin{multline*}
\Big[
S^{[m,0]}_{\sigma(1) \dots \sigma(n)}\left(x_{\sigma(1)},\dots,x_{\sigma(n)}\right)
\mathcal{R}^{1\dots n}_{\sigma}
\mathcal{R}^{\rho}_{1\dots n}
\Big]_{i_1\dots i_n}^{j_1\dots j_n}
=
\\[-1cm]
\delta_{0,\#\{k:i_k=\part, j_k=\hole\}}
\prod_{k: i_k = \part}\ \prod_{l: j_l = \hole} 
\frac{(x_l-x_k)}{(x_l-t x_k)}
\times
\left(
\begin{tikzpicture}[scale=0.6,baseline=(mid.base)]
\node (mid) at (1,2.5) {};
\draw[fline] plot [smooth] coordinates {(2,3.5) (2.5,3.5) (4.5,2.5)};
\draw[fline] plot [smooth] coordinates {(2,2.5) (2.5,2.5) (4.5,3.5)};
\draw[fline] (2,1.5) -- (4.5,1.5);
\foreach\y in {1,...,3}{
\dark{1}{\y};
}
\node at (1.5,0.5) {$\ss m$};
\node at (1.5,4.5) {$\ss 0$};
\foreach\x in {1,...,3}{
\ebull{1}{0.5+\x}{0.09};
}
\gbull{4.5}{1.5}{0.09};
\ebull{4.5}{2.5}{0.09};
\ebull{4.5}{3.5}{0.09};
\node at (-0.5,2.5) {$\ss \red x_l: i_l = \hole $};
\draw[dred] (0.7,1) -- (0.5,1) -- (0.5,4) -- (0.7,4);
\node at (6,3) {$\ss \red x_l: j_l = \hole$};
\draw[dred] (4.8,2.2) -- (5,2.2) -- (5,4) -- (4.8,4);
\node at (6,1.25) {$\ss \red \substack{x_k: i_k = \hole \\ \quad j_k = \part}$};
\draw[dred] (4.8,1) -- (5,1) -- (5,2) -- (4.8,2);
\end{tikzpicture}
\right).
\end{multline*}
For the term remaining in parentheses, we have
\begin{align}
\begin{tikzpicture}[scale=0.6,baseline=(mid.base)]
\node (mid) at (1,2.5) {};
\draw[fline] plot [smooth] coordinates {(2,3.5) (2.5,3.5) (4.5,2.5)};
\draw[fline] plot [smooth] coordinates {(2,2.5) (2.5,2.5) (4.5,3.5)};
\draw[fline] (2,1.5) -- (4.5,1.5);
\foreach\y in {1,...,3}{
\dark{1}{\y};
}
\node at (1.5,0.5) {$\ss m$};
\node at (1.5,4.5) {$\ss 0$};
\foreach\x in {1,...,3}{
\ebull{1}{0.5+\x}{0.09};
}
\gbull{4.5}{1.5}{0.09};
\ebull{4.5}{2.5}{0.09};
\ebull{4.5}{3.5}{0.09};
\node at (-0.5,2.5) {$\ss \red x_l: i_l = \hole $};
\draw[dred] (0.7,1) -- (0.5,1) -- (0.5,4) -- (0.7,4);
\node at (6,3) {$\ss \red x_l: j_l = \hole$};
\draw[dred] (4.8,2.2) -- (5,2.2) -- (5,4) -- (4.8,4);
\node at (6,1.25) {$\ss \red \substack{x_k: i_k = \hole \\ \quad j_k = \part}$};
\draw[dred] (4.8,1) -- (5,1) -- (5,2) -- (4.8,2);
\end{tikzpicture}
=
\begin{tikzpicture}[scale=0.6,baseline=(mid.base)]
\node (mid) at (1,2.5) {};
\draw[fline] plot [smooth] coordinates {(1,3.5) (0.5,3.5) (-1.5,2.5)};
\draw[fline] plot [smooth] coordinates {(1,2.5) (0.5,2.5) (-1.5,3.5)};
\draw[fline] (1,1.5) -- (-1.5,1.5);
\foreach\y in {1,...,3}{
\dark{1}{\y};
}
\node at (1.5,0.5) {$\ss m$};
\node at (1.5,4.5) {$\ss 0$};
\foreach\x in {1,...,3}{
\ebull{-1.5}{0.5+\x}{0.09};
}
\gbull{2}{1.5}{0.09};
\ebull{2}{2.5}{0.09};
\ebull{2}{3.5}{0.09};
\node at (-3,2.5) {$\ss \red x_l: i_l = \hole $};
\draw[dred] (-1.8,1) -- (-2,1) -- (-2,4) -- (-1.8,4);
\node at (3.5,3) {$\ss \red x_l: j_l = \hole$};
\draw[dred] (2.3,2.2) -- (2.5,2.2) -- (2.5,4) -- (2.3,4);
\node at (3.5,1.25) {$\ss \red \substack{x_k: i_k = \hole \\ \quad j_k = \part}$};
\draw[dred] (2.3,1) -- (2.5,1) -- (2.5,2) -- (2.3,2);
\end{tikzpicture}
&=
\begin{tikzpicture}[scale=0.6,baseline=(mid.base)]
\node (mid) at (1,2.5) {};
\foreach\y in {1,...,3}{
\dark{1}{\y};
}
\node at (1.5,0.5) {$\ss m$};
\node at (1.5,4.5) {$\ss 0$};
\foreach\x in {1,...,3}{
\ebull{1}{0.5+\x}{0.09};
}
\gbull{2}{1.5}{0.09};
\ebull{2}{2.5}{0.09};
\ebull{2}{3.5}{0.09};
\node at (-0.5,2.5) {$\ss \red x_l: i_l = \hole $};
\draw[dred] (0.7,1) -- (0.5,1) -- (0.5,4) -- (0.7,4);
\node at (3.5,3) {$\ss \red x_l: j_l = \hole$};
\draw[dred] (2.3,2.2) -- (2.5,2.2) -- (2.5,4) -- (2.3,4);
\node at (3.5,1.25) {$\ss \red \substack{x_k: i_k = \hole \\ \quad j_k = \part}$};
\draw[dred] (2.3,1) -- (2.5,1) -- (2.5,2) -- (2.3,2);
\end{tikzpicture}
\\
&=
\delta_{m,\#\{k:i_k=\hole, j_k=\part\}} 
\prod_{l: i_l=\hole} x_l,
\nonumber
\end{align}
where the first equality follows from the intertwining equation \eqref{rll}, the second equality from the fact that $\begin{tikzpicture}[scale=0.6,baseline=-1]
\draw[thick, smooth] (1,0) -- (-1,0);
\node[label={left: \fs \red $x$}] at (-1,0) {};
\draw[thick, smooth] (0,1) -- (0,-1);
\node[label={below: \fs \red $y$}] at (0,-1) {};
\ebull{-0.5}{0}{0.09}; \ebull{0.5}{0}{0.09};
\ebull{0}{-0.5}{0.09}; \ebull{0}{0.5}{0.09};
\end{tikzpicture}
=1$, and the third equality from the admissible tiles \eqref{tiles}. In summary, we have shown that
\begin{multline*}
\Big[
S^{[m,0]}_{\sigma(1) \dots \sigma(n)}\left(x_{\sigma(1)},\dots,x_{\sigma(n)}\right)
\mathcal{R}^{1\dots n}_{\sigma}
\mathcal{R}^{\rho}_{1\dots n}
\Big]_{i_1\dots i_n}^{j_1\dots j_n}
=
\\
\delta_{0,\#\{k:i_k=\part, j_k=\hole\}}
\delta_{m,\#\{k:i_k=\hole, j_k=\part\}} 
\prod_{l: i_l=\hole} x_l
\prod_{k: i_k = \part}\ \prod_{l: j_l = \hole} 
\frac{(x_l-x_k)}{(x_l-t x_k)}.
\end{multline*}
Finally, combining with equation \eqref{...1} and the fact that
\begin{align*}
\prod_{1\leq k<l \leq n}
b^{-1}_{j_k,j_l}(x_k,x_l)
=
\prod_{k: j_k = \part}\ \prod_{l: j_l = \hole} 
\frac{(x_l-t x_k)}{(x_l-x_k)},
\end{align*}
there is some cancellation and we obtain
\begin{align}
\label{twist-comp}
\Big[
\widetilde{S}^{[m,0]}_{1\dots n}(x_1,\dots,x_n)
\Big]_{i_1\dots i_n}^{j_1\dots j_n}
=
\delta_{0,\#\{k:i_k=\part, j_k=\hole\}}
\delta_{m,\#\{k:i_k=\hole, j_k=\part\}} 
\prod_{l: i_l=\hole} x_l
\prod_{k: \left\{ \substack{i_k = \hole \\ j_k = \part} \right.}\ \prod_{l: j_l = \hole} 
\frac{(x_l-t x_k)}{(x_l-x_k)}.
\end{align}
It is immediate that the $2^n \times 2^n$ matrix \eqref{twist-explicit} has the same components as those given by 
\eqref{twist-comp}.
\end{proof}

Returning to \eqref{P-columns2} and using the explicit form \eqref{twist-explicit} of the twisted column-to-column transfer matrices, we now recover the sum formula that was the aim of this subsection.
\begin{lem}
\label{lem:P-sym}
Hall--Littlewood polynomials are given by
\begin{align}
\label{P-sym}
P_{\lambda}(x_1,\dots,x_n;t)
=
\sum_{\sigma \in S_n/S^{\lambda}_n}
\
\sigma
\left(
\prod_{i=1}^{n}
x_{i}^{\lambda_i}
\prod_{i,j:\lambda_i>\lambda_j}
\frac{(x_{i}-tx_{j})}{(x_{i}-x_{j})}
\right)
\end{align}
where $S_n$ is the ordinary symmetric group, and $S_n^{\lambda}$ is the subgroup of $S_n$ which stabilizes $\lambda$.
\end{lem}

\begin{proof}
For $i > \lambda_1$ (the largest part of $\lambda$), we clearly have $m_i(\lambda) = 0$. Given that 
$\widetilde{S}^{[0,0]}(x_1,\dots,x_n) \ket{\part,\dots,\part} = \ket{\part,\dots,\part}$, the product in \eqref{P-columns2} truncates and we are able to write
\begin{align}
\label{truncated}
\prod_{i=1}^{n} x_i
P_{\lambda}(x_1,\dots,x_n;t)
=
\bra{\hole,\dots,\hole}
\prod_{i=0}^{\lambda_1}
\widetilde{S}^{[m_i(\lambda),0]}(x_1,\dots,x_n)
\ket{\part,\dots,\part}.
\end{align}
We now use the explicit form \eqref{twist-explicit} of the twisted column-to-column transfer matrices. 
Since \eqref{twist-explicit} contains only spin-raising operators (in our language, operators which map $\ket{\part} \rightarrow \ket{\hole}$ in the auxiliary spaces), one can deduce that
\begin{align}
\label{symmet}
\prod_{i=1}^{n} x_i
P_{\lambda}(x_1,\dots,x_n;t)
=
{\rm Sym}
\left\{
\prod_{i=1}^{n}
x_i^{\lambda_i+1}
\prod_{i,j: \lambda_i > \lambda_j}
\frac{(x_i-tx_j)}{(x_i-x_j)}
\right\},
\end{align}
where the term in braces is the simplest possible one that arises in \eqref{truncated}. Namely, it comes from the unique configuration
\begin{align*}
\begin{tikzpicture}[scale=0.7,baseline=0]
\foreach\y in {1,...,5}{
\foreach\x in {1,...,5}{
\dark{\x}{\y}
}
}
\ebull{1}{1.5}{0.07}; \ebull{1}{2.5}{0.07}; \ebull{1}{3.5}{0.07}; \ebull{1}{4.5}{0.07}; \ebull{1}{5.5}{0.07};
\gbull{6}{1.5}{0.07}; \gbull{6}{2.5}{0.07}; \gbull{6}{3.5}{0.07}; \gbull{6}{4.5}{0.07}; \gbull{6}{5.5}{0.07};
\gbull{1.5}{1}{0.07};
\gbull{2.5}{1}{0.07}; 
\gbull{4.4}{1}{0.07}; \gbull{4.6}{1}{0.07}; \gbull{5.5}{1}{0.07};
\draw[gline] (6,1.5) -- (1.5,1.5) -- (1.5,1);
\draw[gline] (6,2.5) -- (2.5,2.5) -- (2.5,1);
\draw[gline] (6,4.5) -- (4.4,4.5) -- (4.4,1);
\draw[gline] (6,3.5) -- (4.6,3.5) -- (4.6,1);
\draw[gline] (6,5.5) -- (5.5,5.5) -- (5.5,1);
\foreach\y in {1,...,5}{
\node at (0.5,-\y+6.5) {$\ss \red x_{\y}$};
}
\node at (1.5,6.5) {$\ss 0$};
\node at (2.5,6.5) {$\ss 0$};
\node at (3.5,6.5) {$\ss 0$};
\node at (4.5,6.5) {$\ss 0$};
\node at (5.5,6.5) {$\ss 0$};
\node at (1.5,0.5) {$\ss m_0$};
\node at (2.5,0.5) {$\ss m_1$};
\node at (3.5,0.5) {$\ss m_2$};
\node at (4.5,0.5) {$\ss m_3$};
\node at (5.5,0.5) {$\ss m_4$};
\node at (7,3.5) {$\equiv$}; 
\end{tikzpicture}
\
\begin{tikzpicture}[scale=0.7,baseline=0]
\foreach\y in {1,...,5}{
\foreach\x in {1,...,5}{
\dark{\x}{\y}
}
}
\bbull{1}{1.5}{0.07}; \bbull{1}{2.5}{0.07}; \bbull{1}{3.5}{0.07}; \bbull{1}{4.5}{0.07}; \bbull{1}{5.5}{0.07};
\ebull{6}{1.5}{0.07}; \ebull{6}{2.5}{0.07}; \ebull{6}{3.5}{0.07}; \ebull{6}{4.5}{0.07}; \ebull{6}{5.5}{0.07};
\bbull{1.5}{1}{0.07};
\bbull{2.5}{1}{0.07}; 
\bbull{4.4}{1}{0.07}; \bbull{4.6}{1}{0.07}; \bbull{5.5}{1}{0.07};
\draw[bline] (1,1.5) -- (1.5,1.5) -- (1.5,1);
\draw[bline] (1,2.5) -- (2.5,2.5) -- (2.5,1);
\draw[bline] (1,3.5) -- (4.4,3.5) -- (4.4,1);
\draw[bline] (1,4.5) -- (4.6,4.5) -- (4.6,1);
\draw[bline] (1,5.5) -- (5.5,5.5) -- (5.5,1);
\foreach\y in {1,...,5}{
\node at (0.5,-\y+6.5) {$\ss \red x_{\y}$};
}
\node at (1.5,6.5) {$\ss 0$};
\node at (2.5,6.5) {$\ss 0$};
\node at (3.5,6.5) {$\ss 0$};
\node at (4.5,6.5) {$\ss 0$};
\node at (5.5,6.5) {$\ss 0$};
\node at (1.5,0.5) {$\ss m_0$};
\node at (2.5,0.5) {$\ss m_1$};
\node at (3.5,0.5) {$\ss m_2$};
\node at (4.5,0.5) {$\ss m_3$};
\node at (5.5,0.5) {$\ss m_4$};
\node at (7,3.5) {$=$};
\end{tikzpicture}
\
\begin{tikzpicture}[scale=0.7,baseline=0]
\node at (11,5.5) {$x_1^{5} \frac{(x_1-t x_2)(x_1-t x_3)(x_1-t x_4)(x_1-t x_5)}
{(x_1-x_2)(x_1-x_3)(x_1-x_4)(x_1-x_5)}$};
\node at (11,4.5) {$x_2^{4} \frac{(x_2-t x_4)(x_2-t x_5)}{(x_2-x_4)(x_2-x_5)}$};
\node at (11,3.5) {$x_3^{4} \frac{(x_3-t x_4)(x_3-t x_5)}{(x_3-x_4)(x_3-x_5)}$};
\node at (11,2.5) {$x_4^{2} \frac{(x_4-t x_5)}{(x_4-x_5)}$};
\node at (11,1.5) {$x_5$};
\end{tikzpicture}
\end{align*}
shown here in the case $\lambda = (4,3,3,1,0)$, where each column is now to be understood as a twisted operator of the form \eqref{twist-explicit}. In view of \eqref{twist-explicit}, the tile $\begin{tikzpicture}[scale=0.5,baseline=-0.1cm] \dark{-0.5}{-0.5};  \gbull{-0.5}{0}{0.09};\ebull{0.5}{0}{0.09}; \end{tikzpicture} \equiv \begin{tikzpicture}[scale=0.5,baseline=-0.1cm] \dark{-0.5}{-0.5};  \ebull{-0.5}{0}{0.09};\bbull{0.5}{0}{0.09}; \end{tikzpicture}$ is now forbidden, while the Boltzmann weight of $\begin{tikzpicture}[scale=0.5,baseline=-0.1cm] \dark{-0.5}{-0.5};  \ebull{-0.5}{0}{0.09};\ebull{0.5}{0}{0.09}; \end{tikzpicture} \equiv \begin{tikzpicture}[scale=0.5,baseline=-0.1cm] \dark{-0.5}{-0.5};  \bbull{-0.5}{0}{0.09};\bbull{0.5}{0}{0.09}; \end{tikzpicture}$ is now non-local (it depends on the other tiles in its column). The sum over all possible terms in \eqref{truncated} is achieved by the symmetrization \eqref{symmet}. Cancelling the common factor $\prod_{i=1}^{n} x_i$ from both sides of \eqref{symmet}, we obtain \eqref{P-sym}. 

\end{proof}

\begin{rmk}
Hall--Littlewood polynomials can also be expressed as a sum over the full symmetric group $S_n$, without dividing out the subgroup $S^{\lambda}_n$, at the expense of an overall normalization: 
\begin{align}
\label{P-sym-2}
P_{\lambda}(x_1,\dots,x_n;t)
=
\frac{1}{v_{\lambda}(t)}
\sum_{\sigma \in S_n}
\
\sigma
\left(
\prod_{i=1}^{n}
x_{i}^{\lambda_i}
\prod_{1 \leq i<j \leq n}
\frac{(x_{i}-tx_{j})}{(x_{i}-x_{j})}
\right),
\qquad
v_{\lambda}(t)
=
\prod_{i=0}^{\infty}
\left(
\prod_{j=1}^{m_i(\lambda)}
\frac{1-t^j}{1-t}
\right).
\end{align}
For more details on the equivalence between \eqref{P-sym} and \eqref{P-sym-2}, see Section 1, Chapter III of \cite{mac}.

\end{rmk}

\subsection{$BC_n$ Hall--Littlewood polynomials: lattice paths}
\label{sec:bc-lattice-paths}

\begin{defn}
Let $BC_n$ Hall--Littlewood polynomials $K_{\lambda}(x^{\pm1};t;\gamma,\delta)$ and $L_{\lambda}(x^{\pm1};t;\gamma,\delta)$ be defined as follows: 
\begin{align}
\label{K}
\prod_{i=1}^{n}
(x_i-t \b{x}_i)
K_{\lambda}(x_1^{\pm1},\dots,x_n^{\pm1};t;\gamma,\delta)
&=
\langle \widehat{0} |
\otimes
\langle \lambda |
\mathbb{T}_{--}(x_n;\gamma,\delta)
\dots
\mathbb{T}_{--}(x_1;\gamma,\delta)
|\widehat{0} \rangle
\otimes
|0 \rangle
\\
\label{L}
\prod_{i=1}^{n}
(\b{x}_i-t x_i)
L_{\lambda}(x_1^{\pm1},\dots,x_n^{\pm1};t;\gamma,\delta)
&=
\langle \widehat{0} |
\otimes
\langle 0 |
\mathbb{T}^{*}_{++}(x_1;\gamma,\delta)
\dots
\mathbb{T}^{*}_{++}(x_n;\gamma,\delta)
|\widehat{0} \rangle
\otimes
|\lambda \rangle
\end{align}
where $|\widehat{0} \rangle = |0\rangle_{-2} \otimes |0\rangle_{-1} \in \widehat{V}$ and $|0\rangle, |\lambda \rangle \in V$, as usual.
\end{defn}

\begin{rmk}{\rm 
For the moment, we treat \eqref{K} and \eqref{L} as definitions of the hyperoctahedral Hall--Littlewood polynomial and its dual, respectively. The fact that it is valid to do so will not be demonstrated until the next subsection, where we will show that \eqref{K} is equivalent to a known formula for $BC_n$ Hall--Littlewood polynomials, in terms of a sum over the hyperoctahedral group \cite{ven}.

It is important to mention that \eqref{K} is not the most general type of $BC_n$ symmetric Hall--Littlewood polynomial, since only two boundary parameters $\gamma,\delta$ are present. As written, \eqref{K} is the $q=0$ case of a two-parameter Koornwinder polynomial, originally defined by Macdonald \cite{mac2}. A four-parameter version of \eqref{K} should also be attainable, but the obvious approach in this direction (namely augmenting the lattice to include four boundary sites, each with its own parameter) does not lead to four-parameter hyperoctahedral Hall--Littlewood polynomials as they appear in \cite{ven}. We have not attempted to resolve this issue in the current work, sticking to two boundary parameters for simplicity.
} 
\end{rmk}

\begin{rmk}{\rm 
Equations \eqref{K} and \eqref{L} allow the $BC_n$ Hall--Littlewood polynomials to be calculated as sums over lattice paths. Indeed, by replacing each double-row transfer matrix with its graphical counterpart, we typically obtain configurations of the form
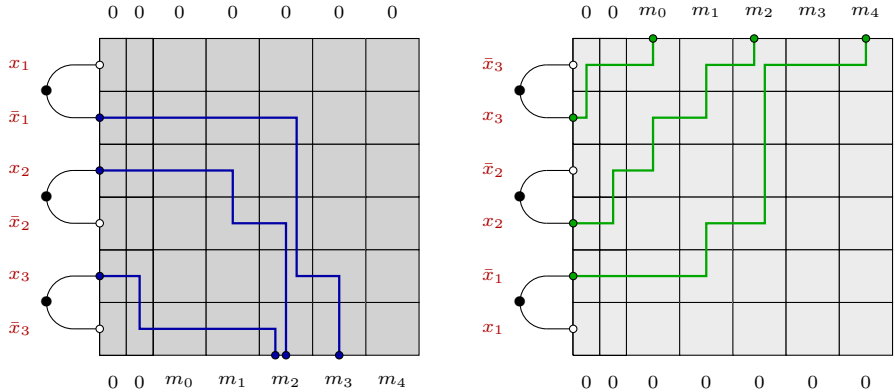
\begin{equation*}
\begin{tikzpicture}[scale=0.7]
\foreach\y in {1,...,6}{
\foreach\x in {2,...,6}{
\dark{\x}{\y};
\bdark{1}{\y};\bdark{1.5}{\y};
}
}
\foreach\y in {0.5,2.5,4.5}{
\begin{scope}[xscale=-1]
\draw[smooth] (-0.5,\y+1) arc (-90:90:0.5);
\draw[smooth] (-1,\y+2)--(-0.5,\y+2);
\draw[smooth] (-1,\y+1)--(-0.5,\y+1);
\node at (0,\y+1.5) {$\bullet$};
\end{scope}
}
\foreach\y in {1,2,3}{
\node at (-0.5,-2*\y+8.5) {$\ss \red x_{\y}$};
\node at (-0.5,-2*\y+7.5) {$\ss \red \b{x}_{\y}$};
}
\draw[bline] (1,2.5) -- (1.75,2.5) -- (1.75,1.5) -- (4.3,1.5) -- (4.3,1);
\draw[bline] (1,4.5) -- (3.5,4.5) -- (3.5,3.5) -- (4.5,3.5) -- (4.5,1);
\draw[bline] (1,5.5) -- (4.7,5.5) -- (4.7,2.5) -- (5.5,2.5) -- (5.5,1);
\ebull{1}{1.5}{0.07}; \ebull{1}{3.5}{0.07}; \ebull{1}{6.5}{0.07};
\bbull{1}{2.5}{0.07}; \bbull{1}{4.5}{0.07}; \bbull{1}{5.5}{0.07};
\bbull{4.3}{1}{0.07}; \bbull{4.5}{1}{0.07}; \bbull{5.5}{1}{0.07};
\node at (1.25,7.5) {$\ss 0$};
\node at (1.75,7.5) {$\ss 0$};
\node at (2.5,7.5) {$\ss 0$};
\node at (3.5,7.5) {$\ss 0$};
\node at (4.5,7.5) {$\ss 0$};
\node at (5.5,7.5) {$\ss 0$};
\node at (6.5,7.5) {$\ss 0$};
\node at (1.25,0.5) {$\ss 0$};
\node at (1.75,0.5) {$\ss 0$};
\node at (2.5,0.5) {$\ss m_0$};
\node at (3.5,0.5) {$\ss m_1$};
\node at (4.5,0.5) {$\ss m_2$};
\node at (5.5,0.5) {$\ss m_3$};
\node at (6.5,0.5) {$\ss m_4$};
\end{tikzpicture}
\qquad
\begin{tikzpicture}[scale=0.7]
\foreach\y in {1,...,6}{
\foreach\x in {2,...,6}{
\light{\x}{\y};
\blight{1}{\y};\blight{1.5}{\y};
}
}
\foreach\y in {0.5,2.5,4.5}{
\begin{scope}[xscale=-1]
\draw[smooth] (-0.5,\y+1) arc (-90:90:0.5);
\draw[smooth] (-1,\y+2)--(-0.5,\y+2);
\draw[smooth] (-1,\y+1)--(-0.5,\y+1);
\node at (0,\y+1.5) {$\bullet$};
\end{scope}
}
\foreach\y in {1,2,3}{
\node at (-0.5,2*\y-0.5) {$\ss \red x_{\y}$};
\node at (-0.5,2*\y+0.5) {$\ss \red \b{x}_{\y}$};
}
\draw[gline] (1,5.5) -- (1.25,5.5) -- (1.25,6.5) -- (2.5,6.5) -- (2.5,7);
\draw[gline] (1,3.5) -- (1.75,3.5) -- (1.75,4.5) -- (2.5,4.5) -- (2.5,5.5) -- (3.5,5.5) -- (3.5,6.5) -- (4.4,6.5) -- (4.4,7);
\draw[gline] (1,2.5) -- (3.5,2.5) -- (3.5,3.5) -- (4.6,3.5) -- (4.6,6.5) -- (6.5,6.5) -- (6.5,7);
\ebull{1}{1.5}{0.07}; \ebull{1}{4.5}{0.07}; \ebull{1}{6.5}{0.07};
\gbull{1}{2.5}{0.07}; \gbull{1}{3.5}{0.07}; \gbull{1}{5.5}{0.07};
\gbull{2.5}{7}{0.07}; \gbull{4.4}{7}{0.07}; \gbull{6.5}{7}{0.07};
\node at (1.25,7.5) {$\ss 0$};
\node at (1.75,7.5) {$\ss 0$};
\node at (2.5,7.5) {$\ss m_0$};
\node at (3.5,7.5) {$\ss m_1$};
\node at (4.5,7.5) {$\ss m_2$};
\node at (5.5,7.5) {$\ss m_3$};
\node at (6.5,7.5) {$\ss m_4$};
\node at (1.25,0.5) {$\ss 0$};
\node at (1.75,0.5) {$\ss 0$};
\node at (2.5,0.5) {$\ss 0$};
\node at (3.5,0.5) {$\ss 0$};
\node at (4.5,0.5) {$\ss 0$};
\node at (5.5,0.5) {$\ss 0$};
\node at (6.5,0.5) {$\ss 0$};
\end{tikzpicture}
\end{equation*}
where the left lattice is a single configuration of the right hand side of \eqref{K} with $\lambda = (3,2,2)$, and the right lattice a single configuration of the right hand side of \eqref{L} with $\lambda = (4,2,0)$.
}
\end{rmk}

\begin{rmk}{\rm 
While \eqref{K} and \eqref{L} have natural combinatorial interpretations, they do not lead immediately to branching formulae for $BC_n$ Hall--Littlewood polynomials. For example, inserting a complete set of states in \eqref{K}, we find that
\begin{align*}
\prod_{i=1}^{n}
(x_i-t \b{x}_i)
K_{\lambda}(x_1^{\pm1},\dots,x_n^{\pm1};t;\gamma,\delta)
=
\sum_{i,j,\nu}
\langle 0 |_{-2}
\otimes
\langle 0 |_{-1}
\otimes
\langle \lambda |
\mathbb{T}_{--}(x_n;\gamma,\delta)
| i \rangle_{-2}
\otimes
| j \rangle_{-1}
\otimes
|\nu \rangle
\\
\times
\langle i |_{-2}
\otimes
\langle j |_{-1}
\otimes
\langle \nu |
\mathbb{T}_{--}(x_{n-1};\gamma,\delta)
\dots
\mathbb{T}_{--}(x_1;\gamma,\delta)
| 0 \rangle_{-2}
\otimes
| 0 \rangle_{-1}
\otimes
|0\rangle,
\end{align*}
where $i,j$ are summed over all possible non-negative integer values (it is easy to deduce that the only permissible values are $i=j=0$, $i=1,j=0$ and $i=0,j=1$), and $\nu$ is summed over all possible partitions. Since one of $i,j$ can be non-zero, this does not express a $BC_n$ Hall--Littlewood polynomial as a sum over those of type $BC_{n-1}$, but rather over a more general class of hyperoctahedrally symmetric polynomials. For this reason, in the presence of non-zero boundary parameters, we cannot easily relate our formula with the $q=0$ limit of the branching formula in \cite{vde3}.

However, in the special case when both boundary parameters are zero ($\gamma=\delta = 0$), the contribution from the boundary columns becomes trivial and we find that
\begin{align*}
\prod_{i=1}^{n}
(x_i-t \b{x}_i)
K_{\lambda}(x_1^{\pm1},\dots,x_n^{\pm1};t)
=
\sum_{\nu}
\langle 0 |_{-2}
\otimes
\langle 0 |_{-1}
\otimes
\langle \lambda |
\mathbb{T}_{--}(x_n;0,0)
| 0 \rangle_{-2}
\otimes
| 0 \rangle_{-1}
\otimes
|\nu \rangle
\\
\times
\langle 0 |_{-2}
\otimes
\langle 0 |_{-1}
\otimes
\langle \nu |
\mathbb{T}_{--}(x_{n-1};0,0)
\dots
\mathbb{T}_{--}(x_1;0,0)
| 0 \rangle_{-2}
\otimes
| 0 \rangle_{-1}
\otimes
|0\rangle,
\end{align*}
leading to a genuine branching rule. Calculating $\langle 0 |_{-2}
\otimes
\langle 0 |_{-1}
\otimes
\langle \lambda |
\mathbb{T}_{--}(x_n;0,0)
| 0 \rangle_{-2}
\otimes
| 0 \rangle_{-1}
\otimes
|\nu \rangle
$ explicitly, we obtain the branching formula
\begin{align*}
K_{\lambda}(x_1^{\pm1},\dots,x_n^{\pm1};t)
=
\sum_{\nu}
K_{\lambda/\nu}(x_n^{\pm1};t)
K_{\nu}(x_1^{\pm1},\dots,x_{n-1}^{\pm1};t),
\end{align*}
with
\begin{align*}
K_{\lambda/\nu}(x^{\pm1};t)
=
\frac{1}{x-t\b{x}}
\left\{
x
\sum_{\substack{\mu: \nu \preccurlyeq \mu \preccurlyeq \lambda \\ \ell(\mu) = \ell(\lambda)}}
-t\b{x}
\sum_{\substack{\mu: \nu \preccurlyeq \mu \preccurlyeq \lambda \\ \ell(\mu) = \ell(\lambda)-1}}
\right\}
x^{2|\mu|-|\lambda|-|\nu|} 
\psi^{(0)}_{\lambda/\mu}(t) \psi^{(0)}_{\mu/\nu}(t),
\end{align*}
(a function which is in fact symmetric under $x \leftrightarrow \b{x}$) and where we have defined
\begin{align*}
\psi^{(0)}_{\lambda/\mu}(t)
:=
\prod_{i \geq 0: m_i(\mu) = m_i(\lambda)+1}
(1-t^{m_i(\mu)}),
\qquad
\ell(\mu)
:=
\sum_{j=0}^{\infty} m_j(\mu).
\end{align*}
By similar arguments, we find that (for $\gamma=\delta=0$)
\begin{align*}
L_{\lambda}(x_1^{\pm1},\dots,x_n^{\pm1};t)
=
\sum_{\nu}
L_{\nu}(x_1^{\pm1},\dots,x_{n-1}^{\pm1};t)
L_{\lambda/\nu}(x_n^{\pm1};t),
\end{align*}
with
\begin{align*}
L_{\lambda/\nu}(x^{\pm1};t)
=
\frac{1}{\b{x}-tx}
\left\{
\b{x}
\sum_{\substack{\mu: \nu \preccurlyeq \mu \preccurlyeq \lambda \\ \ell(\mu) = \ell(\lambda)}}
-t x
\sum_{\substack{\mu: \nu \preccurlyeq \mu \preccurlyeq \lambda \\ \ell(\mu) = \ell(\lambda)-1}}
\right\}
x^{|\lambda|+|\nu|-2|\mu|} 
\varphi^{(0)}_{\lambda/\mu}(t) \varphi^{(0)}_{\mu/\nu}(t),
\end{align*}
where we have defined
\begin{align*}
\varphi^{(0)}_{\lambda/\mu}(t)
:=
\prod_{i \geq 0: m_i(\mu)+1 = m_i(\lambda)}
(1-t^{m_i(\lambda)}).
\end{align*}
Now since
\begin{align*}
\varphi^{(0)}_{\lambda/\mu}(t) 
= 
\frac{\prod_{j=1}^{m_0(\lambda)} (1-t^j) b_{\lambda}(t)}
{\prod_{k=1}^{m_0(\mu)} (1-t^k) b_{\mu}(t)}  
\psi^{(0)}_{\lambda/\mu}(t),
\end{align*} 
we conclude that
\begin{align*}
L_{\lambda/\nu}(x^{\pm1};t)
=
\frac{\prod_{j=1}^{m_0(\lambda)} (1-t^j) b_{\lambda}(t)}{\prod_{k=1}^{m_0(\nu)} (1-t^k) b_{\nu}(t)}
K_{\lambda/\nu}(x^{\pm1};t),
\end{align*}
which implies that more generally
\begin{align*}
L_{\lambda}(x_1^{\pm1},\dots,x_n^{\pm1};t)
=
\prod_{j=1}^{m_0(\lambda)} (1-t^j) b_{\lambda}(t)
K_{\lambda}(x_1^{\pm1},\dots,x_n^{\pm1};t).
\end{align*}
Hence the $BC_n$ Hall--Littlewood polynomial \eqref{K} and its dual \eqref{L} are equal up to a multiplicative factor  when $\gamma = \delta = 0$, in complete analogy with the $A_n$ Hall--Littlewood polynomials \eqref{P} and \eqref{Q}.

}
\end{rmk}

\subsection{$BC_n$ Hall--Littlewood polynomials: sum over hyperoctahedral group}

Here we proceed along parallel lines to Subsection \ref{sec:a_n-sym}. Our aim is to start from \eqref{K}, which slices the lattice in terms of double rows, and instead slice it by columns. In so doing, we will find that \eqref{K} coincides with a known formula for $BC_n$ Hall--Littlewood polynomials \cite{ven}, proving that \eqref{K} is indeed a valid expression for a Koornwinder polynomial at $q=0$. We begin by defining column-to-column transfer matrices on the double-row lattice:
\begin{align*}
\mathbb{S}^{[i]}(x^{\pm1}_1,\dots,x^{\pm1}_n)
&=
L_{\b{n}}(\b{x}_n,\phi_i,\phid_i)
L_{n}(x_n,\phi_i,\phid_i) 
\dots
L_{\b{1}}(\b{x}_1,\phi_i,\phid_i)
L_{1}(x_1,\phi_i,\phid_i)
\\
\mathbb{B}^{[i]}(x^{\pm1}_1,\dots,x^{\pm1}_n;\gamma)
&=
B^{(\gamma)}_{\b{n}}(\b{x}_n,\phi_i,\phid_i)
B^{(\gamma)}_{n}(x_n,\phi_i,\phid_i) 
\dots
B^{(\gamma)}_{\b{1}}(\b{x}_1,\phi_i,\phid_i)
B^{(\gamma)}_{1}(x_1,\phi_i,\phid_i)
\end{align*}
which are $4^{n} \times 4^{n}$ matrices formed by taking the tensor product of $2n$ $L$ ($B$) matrices with different auxiliary spaces, and elements of ${\rm End}(W_1 \otimes W_{\b{1}} \otimes \cdots \otimes W_n \otimes W_{\b{n}} \otimes V_i)$. Consider the components of 
$\mathbb{S}^{[i]}(x^{\pm1}_1,\dots,x^{\pm1}_n)$ and $\mathbb{B}^{[i]}(x^{\pm1}_1,\dots,x^{\pm1}_n;\gamma)$ in the bosonic space $V_i$:
\begin{align*}
\mathbb{S}^{[l,m]}(x^{\pm1}_1,\dots,x^{\pm1}_n)
&=
\bra{l}_i
L_{\b{n}}(\b{x}_n,\phi_i,\phid_i)
L_{n}(x_n,\phi_i,\phid_i) 
\dots
L_{\b{1}}(\b{x}_1,\phi_i,\phid_i)
L_{1}(x_1,\phi_i,\phid_i)
\ket{m}_i
\\
\mathbb{B}^{[0,0]}(x^{\pm1}_1,\dots,x^{\pm1}_n;\gamma)
&=
\bra{0}_i
B^{(\gamma)}_{\b{n}}(\b{x}_n,\phi_i,\phid_i)
B^{(\gamma)}_{n}(x_n,\phi_i,\phid_i) 
\dots
B^{(\gamma)}_{\b{1}}(\b{x}_1,\phi_i,\phid_i)
B^{(\gamma)}_{1}(x_1,\phi_i,\phid_i)
\ket{0}_i
\end{align*}
which have the simple graphical interpretation
\begin{align*}
\mathbb{S}^{[l,m]}(x^{\pm1}_1,\dots,x^{\pm1}_n)
=
\begin{tikzpicture}[scale=0.6,baseline=(mid.base)]
\node (mid) at (1,3) {};
\node at (0.5,0.5) {$\ss \red \b{x}_{n}$};
\node at (0.5,1.5) {$\ss \red x_{n}$};
\node at (0.5,3.3) {$\red \vdots$};
\node at (0.5,4.5) {$\ss \red \b{x}_{1}$};
\node at (0.5,5.5) {$\ss \red x_{1}$};
\foreach\y in {0,...,5}{
\dark{1}{\y};
}
\node at (1.5,-0.5) {$\ss l$};
\node at (1.5,6.5) {$\ss m$};
\end{tikzpicture}
\qquad
{\rm and}
\qquad
\mathbb{B}^{[0,0]}(x^{\pm1}_1,\dots,x^{\pm1}_n;\gamma)
=
\begin{tikzpicture}[scale=0.6,baseline=(mid.base)]
\node (mid) at (1,3) {};
\node at (0.5,0.5) {$\ss \red \b{x}_{n}$};
\node at (0.5,1.5) {$\ss \red x_{n}$};
\node at (0.5,3.3) {$\red \vdots$};
\node at (0.5,4.5) {$\ss \red \b{x}_{1}$};
\node at (0.5,5.5) {$\ss \red x_{1}$};
\foreach\y in {0,...,5}{
\bdark{1}{\y};
}
\node at (1.25,-0.5) {$\ss 0$};
\node at (1.25,6.5) {$\ss 0$};
\node at (1.25,0.5) {$\ss \red \gamma$};
\node at (1.25,1.5) {$\ss \red \gamma$};
\node at (1.25,2.5) {$\ss \red \gamma$};
\node at (1.25,3.5) {$\ss \red \gamma$};
\node at (1.25,4.5) {$\ss \red \gamma$};
\node at (1.25,5.5) {$\ss \red \gamma$};
\end{tikzpicture}\ \ .
\end{align*}
The $BC_n$ symmetric Hall--Littlewood polynomials \eqref{K} can be expressed in terms of these, as follows:
\begin{multline}
\label{K-columns1}
\prod_{i=1}^{n} (x_i -t \b{x}_i)
K_{\lambda}(x_1^{\pm1},\dots,x_n^{\pm1};t;\gamma,\delta)
=
\\
\bra{\mathbb{K}}
\mathbb{B}^{[0,0]}(x^{\pm1}_1,\dots,x^{\pm1}_n;\gamma)
\mathbb{B}^{[0,0]}(x^{\pm1}_1,\dots,x^{\pm1}_n;\delta)
\prod_{i=0}^{\infty}
\mathbb{S}^{[m_i(\lambda),0]}(x^{\pm1}_1,\dots,x^{\pm1}_n)
\ket{\part,\dots,\part},
\end{multline}
where the left and right edge states are given by
\begin{align*}
\bra{\mathbb{K}} = \bigotimes_{k=1}^{n} \bra{K}_{k\b{k}}
\qquad
{\rm and}
\qquad
\ket{\part,\dots,\part}
=
\bigotimes_{k=1}^{n}
\begin{pmatrix} 0 \\ 1 \end{pmatrix}_k
\otimes
\begin{pmatrix} 0 \\ 1 \end{pmatrix}_{\b{k}}.
\end{align*}
As we did previously, we define twisted versions of the column-to-column transfer matrices, by conjugating with an appropriate $F$ matrix. Since we are working on a lattice with twice as many auxiliary spaces as before, the change of basis is achieved by the $F$ matrix $\mathbb{F}_{1\dots n}(x_1^{\pm1},\dots,x_n^{\pm1}) \equiv 
F_{1\b{1}\dots n \b{n}}(x_1,\b{x}_1,\dots,x_n,\b{x}_n)$. We thus define
\begin{align*}
\widetilde{\mathbb{S}}^{[i]}_{1 \dots n}
(x^{\pm1}_1,\dots,x^{\pm1}_n)
&=
\mathbb{F}_{1 \dots n}(x^{\pm1}_1,\dots,x^{\pm1}_n)
\mathbb{S}^{[i]}_{1 \dots n}
(x^{\pm1}_1,\dots,x^{\pm1}_n)
\mathbb{F}^{-1}_{1 \dots n}(x^{\pm1}_1,\dots,x^{\pm1}_n),
\\
\widetilde{\mathbb{B}}^{[i]}_{1\dots n}(x^{\pm1}_1,\dots,x^{\pm1}_n;\gamma)
&=
\mathbb{F}_{1 \dots n}(x^{\pm1}_1,\dots,x^{\pm1}_n)
\mathbb{B}^{[i]}_{1\dots n}(x^{\pm1}_1,\dots,x^{\pm1}_n;\gamma)
\mathbb{F}^{-1}_{1 \dots n}(x^{\pm1}_1,\dots,x^{\pm1}_n).
\end{align*}
Recalling that $\mathbb{F}^{-1}(x^{\pm1}_1,\dots,x^{\pm1}_n) \ket{\part,\dots,\part}
= \ket{\part,\dots,\part}$, we find that
\begin{multline}
\label{K-columns2}
\prod_{i=1}^{n} (x_i -t \b{x}_i)
K_{\lambda}(x_1^{\pm1},\dots,x_n^{\pm1};t;\gamma,\delta)
\\
\langle \widetilde{\mathbb{K}} |
\widetilde{\mathbb{B}}^{[0,0]}(x^{\pm1}_1,\dots,x^{\pm1}_n;\gamma)
\widetilde{\mathbb{B}}^{[0,0]}(x^{\pm1}_1,\dots,x^{\pm1}_n;\delta)
\prod_{i=0}^{\infty}
\widetilde{\mathbb{S}}^{[m_i(\lambda),0]}(x^{\pm1}_1,\dots,x^{\pm1}_n)
\ket{\part,\dots,\part},
\end{multline}
where $\langle \widetilde{\mathbb{K}} | = \bra{\mathbb{K}}
\mathbb{F}^{-1}(x^{\pm1}_1,\dots,x^{\pm1}_n)$ constitutes a non-trivial modification of 
$\bra{\mathbb{K}}$, and will be calculated explicitly in the following theorem.
\begin{thm}
Let us adopt the convention $x_{\b{k}} := \b{x}_k$, for all $1 \leq k \leq n$. The components of the twisted column-to-column transfer matrices are given explicitly by
\begin{align}
\label{twist-bold-S}
\widetilde{\mathbb{S}}^{[m,0]}(x^{\pm1}_1,\dots,x^{\pm1}_n)
&=
\sum_{\substack{I \subset \{1,\b{1},\ldots,n,\b{n}\} \\ |I|=m}}
\bigotimes_{i \in I}
\begin{pmatrix} 
0 & x_i
\\ 
0 & 0
\end{pmatrix}_i
\ \ 
\bigotimes_{j \not\in I}
\begin{pmatrix} 
x_j
\prod_{k \in I} \left( \frac{x_j-t x_k}{x_j-x_k} \right) & 0
\\ 
0 & 1
\end{pmatrix}_j
\\
\label{twist-bold-B}
\widetilde{\mathbb{B}}^{[0,0]}(x^{\pm1}_1,\dots,x^{\pm1}_n;\gamma)
&=
\bigotimes_{i \in \{1,\b{1},\ldots,n,\b{n}\}}
\begin{pmatrix} 
1 & 0
\\ 
0 & 1-\gamma x_i
\end{pmatrix}_i
\end{align}
while the twisted left boundary is given by
\begin{align}
\label{twist-bdry}
\prod_{i=1}^{n}
\frac{1}{(x_i - t \b{x}_i)}
\langle \widetilde{\mathbb{K}} |
&=
\sum_{\{\epsilon_1,\dots,\epsilon_n\}  \in \{\pm 1\}^n}
\prod_{k=1}^{n}
\left(
\frac{x_k^{\epsilon_k}}{1-x_k^{2\epsilon_k}}
\right)
\prod_{1\leq k < l \leq n}
\left(
\frac{1-t x_k^{\epsilon_k} x_l^{\epsilon_l}}
{1-x_k^{\epsilon_k} x_l^{\epsilon_l}}
\right)
\bigotimes_{k=1}^{n}
\langle -\epsilon_k |_k
\otimes
\langle \epsilon_k |_{\b{k}}
\end{align}
where the sum is taken over all binary strings $\{\epsilon_1,\dots,\epsilon_n\} \in \{\pm 1\}^n$, and
where by definition $\langle +1 |_k = \begin{pmatrix} 1 & 0 \end{pmatrix}_k$ and 
$\langle -1 |_k = \begin{pmatrix} 0 & 1 \end{pmatrix}_k$.
\end{thm}

\begin{proof}
We have already proved \eqref{twist-bold-S}, since it is just \eqref{twist-explicit} with $W_1 \otimes \cdots \otimes W_n \rightarrow W_1 \otimes W_{\b{1}} \otimes \cdots \otimes W_n \otimes W_{\b{n}}$ and $\{x_1,\dots,x_n\} \rightarrow \{x_1,\b{x}_1,\dots,x_n,\b{x}_n\}$. Similarly, the proof of \eqref{twist-bold-B} proceeds in much the same way as that of Theorem \ref{thm:twisting}, although the argument is even simpler considering that both the top and bottom of the column have occupation numbers 0 (in particular, this is why $\widetilde{\mathbb{B}}^{[0,0]}$ is diagonal). The most involved proof is therefore that of equation \eqref{twist-bdry}. Our approach is, once again, to consider the components of $\langle \widetilde{\mathbb{K}} |$:
\begin{align}
\nonumber
\langle \widetilde{\mathbb{K}} |^{j_1 j_{\b{1}} \dots j_n j_{\b{n}}}
&=
\Big[
\bra{\mathbb{K}}
\mathcal{R}^{\rho}_{1 \b{1} \dots n  \b{n}}
\prod_{\substack{k<l \\ k,l \in\{1,\b{1},\ldots,n,\b{n}\}}}
\Delta_{kl}^{-1}(x_k,x_l)
\Big]^{j_1 j_{\b{1}} \dots j_n j_{\b{n}}}
\\
\label{...4}
&
=
\Big[
\bra{\mathbb{K}}
\mathcal{R}^{\rho}_{1 \b{1} \dots n  \b{n}}
\Big]^{j_1 j_{\b{1}} \dots j_n j_{\b{n}}}
\times
\prod_{\substack{k<l \\ k,l \in\{1,\b{1},\ldots,n,\b{n}\}}}
b^{-1}_{j_k,j_l}(x_k,x_l)
\end{align}
where $\rho$ is a permutation of $\{1,\b{1},\ldots,n,\b{n}\}$ such that 
$j_{\rho(1)} \leq j_{\rho(\b{1})} \leq \cdots \leq j_{\rho(n)} \leq j_{\rho(\b{n})}$. The non-trivial part of this expression can be represented graphically as follows:
\begin{align}
\label{...3}
\Big[
\bra{\mathbb{K}}
\mathcal{R}^{\rho}_{1 \b{1} \dots n  \b{n}}
\Big]^{j_1 j_{\b{1}} \dots j_n j_{\b{n}}}
=
\begin{tikzpicture}[scale=0.6,baseline=(mid.base)]
\node (mid) at (1,4) {};
\foreach\x in {-1,0,1,2}{
\draw[smooth] (2,3.5+2*\x) arc (90:270:0.5);
\node at (1.5,3+2*\x) {\fs$\bullet$};
}
\node at (1.2,7.5) {$\ss \red x_1$};
\node at (1.2,6.5) {$\ss \red \b{x}_1$};
\node at (1.2,4.2) {$\red \vdots$};
\node at (1.3,1.5) {$\ss \red x_n$};
\node at (1.3,0.5) {$\ss \red \b{x}_n$};
\draw[fline] plot [smooth] coordinates {
(2,7.5) (5,7.5) (7,6.5)};
\draw[fline] plot [smooth] coordinates {
(2,6.5) (4,6.5) (5,5) (7,3.5)};
\draw[fline] plot [smooth] coordinates {
(2,5.5) (5,5.5) (7,4.5)};
\draw[fline] plot [smooth] coordinates {
(2,4.5) (4,4.5) (7,2.5)};
\draw[fline] plot [smooth] coordinates {
(2,3.5) (4.5,3.5) (7,7.5)};
\draw[fline] plot [smooth] coordinates {
(2,2.5) (4.5,2.5) (7,0.5)};
\draw[fline] (2,1.5) -- (7,1.5);
\draw[fline] plot [smooth] coordinates {(2,0.5) (4.5,1) (7,5.5)};
\foreach\x in {1,...,4}{
\gbull{7}{-0.5+\x}{0.09};
}
\foreach\x in {5,...,8}{
\ebull{7}{-0.5+\x}{0.09};
}
\end{tikzpicture}
\end{align}
where the states on the right edge of the lattice are ordered, with particles filling the lowest edges and holes situated at the highest edges. We proceed to show that this partition function is completely factorized, by stripping away its lines one by one. 

Starting with the topmost double line, there are four cases to consider. In the case $j_{1} = j_{\b{1}} = \hole$, the two ends of the top double line are attached to holes. Since the U-turn vertex always produces one hole and one particle at the left edge of the lattice, by particle-conservation arguments we conclude that \eqref{...3} vanishes in this case. Similarly when $j_{1} = j_{\b{1}} = \part$, we can use the Yang--Baxter and unitarity equations to reposition the top double line, producing a lattice configuration of the form \eqref{vanishing-lem}, which again vanishes.

The two remaining cases can be treated in parallel. For $j_{1} = \hole$, $j_{\b{1}} = \part$ the two ends of the top double line split up, one end attaching to a hole, the other attaching to a particle. Of the two possible U-turn vertex configurations, only
$
\begin{tikzpicture}[scale=0.4,baseline=9]
\begin{scope}[xscale=-1,yscale=1]
\draw[smooth] ({2*sqrt(2) + 3 + 4},0) arc (-90:90:1);
\draw[smooth] ({2*sqrt(2) + 3 + 3.5},2)--({2*sqrt(2) + 3 + 4},2);
\draw[smooth] ({2*sqrt(2) + 3 + 3.5},0)--({2*sqrt(2) + 3 + 4},0);
\node at ({2*sqrt(2) + 3 + 5},1) {$\bullet$};
\ebull{{2*sqrt(2) + 6.5}}{2}{0.14};
\gbull{{2*sqrt(2) + 6.5}}{0}{0.14};
\end{scope}
\end{tikzpicture}
$
will contribute to the value of the partition function (the other configuration will vanish, again by particle-conservation). Consulting the Boltzmann weights we find that the top double line can be removed at the expense of the factor 
$\prod_{k>1: j_k = \hole} (1-\b{x}_1 \b{x}_k)/(1-t \b{x}_1 \b{x}_k) 
\prod_{k>1: j_{\b{k}} = \hole} (x_1-x_k)/(x_1-t x_k)$. For $j_1 = \part$, $j_{\b{1}} = \hole$, the two ends of the top double line must cross somewhere. Using the fish equation \eqref{fish} with the $\mathcal{R}$ matrix normalization, this crossing can be removed at the expense of the factor $(x_1 - t \b{x}_1)/(\b{x}_1 - t x_1)$, at which point the calculation proceeds as above with $x_1$ and $\b{x}_1$ interchanged. Hence in this case we obtain the overall factor 
$$(x_1 - t \b{x}_1)/(\b{x}_1 - t x_1)
\prod_{k>1: j_k = \hole} (x_k-x_1)/(x_k-t x_1) 
\prod_{k>1: j_{\b{k}} = \hole} (1-x_1 x_k)/(1-t x_1 x_k).$$
Iterating this over all double lines, we obtain
\begin{multline*}
\prod_{i=1}^{n}
\frac{1}{(x_i - t \b{x}_i)}
\Big[
\bra{\mathbb{K}}
\mathcal{R}^{\rho}_{1 \b{1} \dots n  \b{n}}
\Big]^{j_1 j_{\b{1}} \dots j_n j_{\b{n}}}
=
\delta_{0,\#\{k:j_k = j_{\b{k}} = \hole\}}
\delta_{0,\#\{k:j_k = j_{\b{k}} = \part\}}
\\
\times
\prod_{k: j_k = \hole}
\frac{1}{(x_k - t \b{x}_k)}
\prod_{k: j_{\b{k}} = \hole}
\frac{1}{(\b{x}_k - t x_k)}
\prod_{\substack{k<l \\ j_k = \hole \\ j_l = \hole}}
\frac{(1-\b{x}_k \b{x}_l)}{(1-t \b{x}_k \b{x}_l)}
\prod_{\substack{k<l \\ j_{k} = \hole \\ j_{\b{l}} = \hole}} 
\frac{(x_k-x_l)}{(x_k-t x_l)}
\prod_{\substack{k<l \\ j_{\b{k}} = \hole \\ j_l = \hole}}
\frac{(x_l-x_k)}{(x_l-t x_k)}
\prod_{\substack{k<l \\ j_{\b{k}} = \hole \\ j_{\b{l}} = \hole}} 
\frac{(1-x_k x_l)}{(1-t x_k x_l)}.
\end{multline*} 
Combining this with \eqref{...4} and the relation
\begin{align*}
\prod_{\substack{k<l \\ k,l \in\{1,\b{1},\ldots,n,\b{n}\}}}
b^{-1}_{j_k,j_l}(x_k,x_l)
=
\prod_{\substack{k: j_k = \part \\ l: j_l = \hole}} 
\frac{(x_l-t x_k)}{(x_l-x_k)}
\prod_{\substack{k: j_k = \part \\ l: j_{\b{l}} = \hole}} 
\frac{(1-t x_k x_l)}{(1-x_k x_l)}
\prod_{\substack{k: j_{\b{k}} = \part \\ l: j_l = \hole}} 
\frac{(1-t \b{x}_k \b{x}_l)}{(1-\b{x}_k \b{x}_l)}
\prod_{\substack{k: j_{\b{k}} = \part \\ l: j_{\b{l}} = \hole}} 
\frac{(x_k-t x_l)}{(x_k-x_l)},
\end{align*}
after some cancellation of factors we find that the components of $\langle \widetilde{\mathbb{K}} |$ are given by
\begin{multline*}
\prod_{i=1}^{n}
\frac{1}{(x_i - t \b{x}_i)}
\langle \widetilde{\mathbb{K}} |^{j_1 j_{\b{1}} \dots j_n j_{\b{n}}}
=
\delta_{0,\#\{k:j_k = j_{\b{k}} = \hole\}}
\delta_{0,\#\{k:j_k = j_{\b{k}} = \part\}}
\\
\times
\prod_{k: j_k = \hole}
\frac{1}{(x_k - \b{x}_k)}
\prod_{k: j_{\b{k}} = \hole}
\frac{1}{(\b{x}_k - x_k)}
\prod_{\substack{k>l \\ j_k = \part \\ j_l = \hole}} 
\frac{(x_l-t x_k)}{(x_l-x_k)}
\prod_{\substack{k>l \\ j_k = \part \\ j_{\b{l}} = \hole}} 
\frac{(1-t x_k x_l)}{(1-x_k x_l)}
\prod_{\substack{k>l \\ j_{\b{k}} = \part \\ j_l = \hole}} 
\frac{(1-t \b{x}_k \b{x}_l)}{(1-\b{x}_k \b{x}_l)}
\prod_{\substack{k>l \\ j_{\b{k}} = \part \\ j_{\b{l}} = \hole}} 
\frac{(x_k-t x_l)}{(x_k-x_l)}
\end{multline*}
It is simple to verify that the covector \eqref{twist-bdry} has the same components as those above.

\end{proof}

\begin{lem}
$BC_n$ Hall--Littlewood polynomials are given by
\begin{multline}
\label{BC-HL1}
K_{\lambda}(x_1^{\pm1},\dots,x_n^{\pm1};t;\gamma,\delta)
=
\\
\sum_{\omega \in (S_n/S_n^{\lambda}) \ltimes \{\pm1\}^n}
\omega\left(
\prod_{i=1}^{n}
x_{i}^{\lambda_i}
\frac{(1-\gamma \b{x}_{i})(1-\delta \b{x}_{i})}{(1-\b{x}_{i}^2)}
\prod_{i,j: \lambda_i > \lambda_j}
\frac{(x_i-tx_j)}{(x_i-x_j)}
\prod_{1 \leq i<j \leq n}
\frac{(1- t \b{x}_i \b{x}_j)}{(1- \b{x}_i \b{x}_j)}
\right)
\end{multline}
where the sum is over the group of signed permutations formed by the semi-direct product of $S_n/S_n^{\lambda}$ and $\{\pm1\}^n$.
\end{lem}

\begin{proof}
The proof is similar to that of Lemma \ref{lem:P-sym}. We begin by truncating \eqref{K-columns2}, as follows:
\begin{multline}
\label{truncated-K}
\prod_{i=1}^{n} (x_i -t \b{x}_i)
K_{\lambda}(x_1^{\pm1},\dots,x_n^{\pm1};t;\gamma,\delta)
=
\\
\langle \widetilde{\mathbb{K}} |
\widetilde{\mathbb{B}}^{[0,0]}(x^{\pm1}_1,\dots,x^{\pm1}_n;\gamma)
\widetilde{\mathbb{B}}^{[0,0]}(x^{\pm1}_1,\dots,x^{\pm1}_n;\delta)
\prod_{i=0}^{\lambda_1}
\widetilde{\mathbb{S}}^{[m_i(\lambda),0]}(x^{\pm1}_1,\dots,x^{\pm1}_n)
\ket{\part,\dots,\part}.
\end{multline}
Now by combining the formulae \eqref{twist-bold-S}--\eqref{twist-bdry} in equation \eqref{truncated-K}, we deduce that
\begin{align}
\label{hyp-sym}
K_{\lambda}(x_1^{\pm1},\dots,x_n^{\pm1};t;\gamma,\delta)
=
{\rm HypSym}\left\{
\prod_{i=1}^{n}
x_{i}^{\lambda_i}
\frac{(1-\gamma \b{x}_{i})(1-\delta \b{x}_{i})}{(1-\b{x}_{i}^2)}
\prod_{i,j: \lambda_i > \lambda_j}
\frac{(x_i-tx_j)}{(x_i-x_j)}
\prod_{1 \leq i<j \leq n}
\frac{(1- t \b{x}_i \b{x}_j)}{(1- \b{x}_i \b{x}_j)}
\right\}
\end{align}
where the symmetrization is over $\{x_1,\dots,x_n\}$ and their inversions. The term in braces corresponds with the simplest configuration in \eqref{truncated-K}, namely
\begin{align*}
\begin{tikzpicture}[scale=0.7]
\foreach\y in {1,...,6}{
\foreach\x in {2,...,6}{
\dark{\x}{\y};
\bdark{1}{\y};\bdark{1.5}{\y};
}
}
\foreach\y in {0.5,2.5,4.5}{
\begin{scope}[xscale=-1]
\draw[smooth] (-0.5,\y+1) arc (-90:90:0.5);
\draw[smooth] (-1,\y+2)--(-0.5,\y+2);
\draw[smooth] (-1,\y+1)--(-0.5,\y+1);
\node at (0,\y+1.5) {$\bullet$};
\end{scope}
}
\foreach\y in {1,2,3}{
\node at (-0.5,-2*\y+8.5) {$\ss \red x_{\y}$};
\node at (-0.5,-2*\y+7.5) {$\ss \red \b{x}_{\y}$};
}
\draw[gline] (7,1.5) -- (1,1.5);
\draw[gline] (7,4.5) -- (4.4,4.5) -- (4.4,1);
\draw[gline] (7,3.5) -- (1,3.5);
\draw[gline] (7,2.5) -- (4.6,2.5) -- (4.6,1);
\draw[gline] (7,5.5) -- (1,5.5);
\draw[gline] (7,6.5) -- (6.5,6.5) -- (6.5,1);
\gbull{1}{1.5}{0.07}; \gbull{1}{3.5}{0.07}; \gbull{1}{5.5}{0.07};
\ebull{1}{2.5}{0.07}; \ebull{1}{4.5}{0.07}; \ebull{1}{6.5}{0.07};
\gbull{7}{1.5}{0.07}; \gbull{7}{3.5}{0.07}; \gbull{7}{5.5}{0.07};
\gbull{7}{2.5}{0.07}; \gbull{7}{4.5}{0.07}; \gbull{7}{6.5}{0.07};
\gbull{4.4}{1}{0.07}; \gbull{4.6}{1}{0.07}; \gbull{6.5}{1}{0.07};
\node at (1.25,7.5) {$\ss 0$};
\node at (1.75,7.5) {$\ss 0$};
\node at (2.5,7.5) {$\ss 0$};
\node at (3.5,7.5) {$\ss 0$};
\node at (4.5,7.5) {$\ss 0$};
\node at (5.5,7.5) {$\ss 0$};
\node at (6.5,7.5) {$\ss 0$};
\node at (1.25,0.5) {$\ss 0$};
\node at (1.75,0.5) {$\ss 0$};
\node at (2.5,0.5) {$\ss m_0$};
\node at (3.5,0.5) {$\ss m_1$};
\node at (4.5,0.5) {$\ss m_2$};
\node at (5.5,0.5) {$\ss m_3$};
\node at (6.5,0.5) {$\ss m_4$};
\node at (8,4) {$=$}; 
\end{tikzpicture}
\
\begin{tikzpicture}[scale=0.7,baseline=0]
\node at (1.5,3.5) {$\blue \frac{(1-t\b{x}_1 \b{x}_2)(1-t\b{x}_1 \b{x}_3)(1-t\b{x}_2 \b{x}_3)}
{(1-\b{x}_1 \b{x}_2)(1-\b{x}_1 \b{x}_3)(1-\b{x}_2 \b{x}_3)}$};
\node at (4.6,3.5) {$\blue \ss\times$};
\node at (6,5.5) {$\blue \ss\b{x}_1/(1-\b{x}_1^2)$};
\node at (6,3.5) {$\blue \ss\b{x}_2/(1-\b{x}_2^2)$};
\node at (6,1.5) {$\blue \ss\b{x}_3/(1-\b{x}_3^2)$};
\node at (7.75,3.5) {$\ss\times$};
\node at (10,6.5) {$\ss x_1^{5} \frac{(x_1-t x_2)(x_1-t x_3)}{(x_1-x_2)(x_1-x_3)}$};
\node at (10,5.5) {$\ss (1-\gamma \b{x}_1)(1-\delta \b{x}_1)$};
\node at (10,4.5) {$\ss x_2^{3}$};
\node at (10,3.5) {$\ss (1-\gamma \b{x}_2)(1-\delta \b{x}_2)$};
\node at (10,2.5) {$\ss x_3^{3}$};
\node at (10,1.5) {$\ss (1-\gamma \b{x}_3)(1-\delta \b{x}_3)$};
\end{tikzpicture}
\end{align*}
shown here for the example $\lambda = (4,2,2)$, with each column denoting a twisted column-to-column transfer matrix. We write the corresponding Boltzmann weight of this configuration alongside. The leftmost factors (in blue) come from the twisted boundary 
$\langle \widetilde{\mathbb{K}} |$, while those on the right come from the lattice itself. Notice that there is some cancellation between factors to produce exactly the term in braces in \eqref{hyp-sym}.

\end{proof}

\begin{rmk}
$BC_n$ Hall--Littlewood polynomials can also be expressed as a sum over the full hyperoctahedral group $S_n \ltimes \{\pm1\}^n$, leading to the result of \cite{ven}:
\begin{multline}
\label{BC-HL2}
K_{\lambda}(x_1^{\pm1},\dots,x_n^{\pm1};t;\gamma,\delta)
=
\\
\frac{1}{v_{\lambda}(t)}
\sum_{\omega \in S_n \ltimes \{\pm1\}^n}
\omega\left(
\prod_{i=1}^{n}
x_{i}^{\lambda_i}
\frac{(1-\gamma \b{x}_{i})(1-\delta \b{x}_{i})}{(1-\b{x}_{i}^2)}
\prod_{1 \leq i<j \leq n}
\frac{(x_i-tx_j)(1- t \b{x}_i \b{x}_j)}{(x_i-x_j)(1- \b{x}_i \b{x}_j)}
\right).
\end{multline}
As with the expressions \eqref{P-sym} and \eqref{P-sym-2}, the equivalence between \eqref{BC-HL1} and \eqref{BC-HL2} can be deduced by noticing that the larger sum in \eqref{BC-HL2} produces all terms present in \eqref{BC-HL1} up to an extra multiplicative factor $v_{\lambda}(t)$.

\end{rmk}

\section{Refined Cauchy identity}
\label{sec:cauchy}

In this section we prove a refined version of the Cauchy identity for Hall--Littlewood polynomials. This identity was originally obtained by Warnaar \cite{war} at the level of Macdonald polynomials, using the action of difference operators on the Macdonald Cauchy identity \cite{kn}. Our proof is independent of difference operators, and based on repeated application of the intertwining equations in the $t$-boson model. Furthermore our method of proof generalises to other refined identities that were considered in \cite{bw,bwz-j}, for which (to the best of our knowledge) proofs based on difference operators are not possible.  

\subsection{Lattice formulation of refined Cauchy identity}

\begin{thm}
\label{thm-cauch}
Hall--Littlewood polynomials satisfy the following refined Cauchy identity:
\begin{multline}
\label{ref-cauch}
\sum_{\lambda}
\prod_{i=1}^{m_0(\lambda)}
(1 - u t^{i})
b_{\lambda}(t)
P_{\lambda}(x_1,\dots,x_n;t)
P_{\lambda}(y_1,\dots,y_n;t)
=
\\
\frac{\prod_{i,j=1}^{n} (1- t x_i y_j)}
{\prod_{1 \leq i<j \leq n} (x_i-x_j)(y_i-y_j)}
\det_{1 \leq i,j \leq n}
\left[
\frac{1-ut + (u-1)t x_i y_j}{(1-x_i y_j) (1-t x_i y_j)}
\right],
\end{multline}
where $u$ is an arbitrary parameter.
\end{thm}

\begin{proof}
In what follows, we parametrize the refining variable as $u=t^{\alpha}$. Consider the following partition function in the $t$-boson model, whose right edge is contracted with an $n \times n$ lattice of $R$ matrix vertices:
\begin{align}
\label{pf-lhs}
\mathcal{P}(x_1,\dots,x_n;y_1,\dots,y_n;t;\alpha)
=
\begin{tikzpicture}[scale=0.6,baseline=(mid.base)]
\node (mid) at (1,4) {};
\foreach\x in {1,...,5}{
\foreach\y in {1,...,3}{
\light{\x}{\y};
}
\foreach\y in {4,...,6}{
\dark{\x}{\y};
}
}
\node at (1.5,0.5) {$\ss \alpha$};
\node at (2.5,0.5) {$\ss 0$};
\node at (3.5,0.5) {$\ss 0$};
\node at (4.5,0.5) {$\ss 0$};
\node at (5.5,0.5) {$\ss \cdots$};
\node at (1.5,7.5) {$\ss \alpha$};
\node at (2.5,7.5) {$\ss 0$};
\node at (3.5,7.5) {$\ss 0$};
\node at (4.5,7.5) {$\ss 0$};
\node at (5.5,7.5) {$\ss \cdots$};
\foreach\x in {1,...,3}{
\gbull{1}{0.5+\x}{0.09}
\ebull{1}{3.5+\x}{0.09}
}
\node at (0.6,1.5) {$\ss \red \b{y}_1$};
\node at (0.6,2.7) {$\ss \red \vdots$};
\node at (0.6,3.5) {$\ss \red \b{y}_n$};
\node at (0.6,4.5) {$\ss \red x_n$};
\node at (0.6,5.7) {$\ss \red \vdots$};
\node at (0.6,6.5) {$\ss \red x_1$};
\foreach\x in {1,...,3}{
\draw[fline] (7.5+0.5*\x,2+0.5*\x) -- (6,3.5+\x);
\draw[fline] (6,0.5+\x) -- (9.5-0.5*\x,4+0.5*\x);
\gbull{7.5+0.5*\x}{2+0.5*\x}{0.09};
\ebull{9.5-0.5*\x}{4+0.5*\x}{0.09};
}
\end{tikzpicture}
\end{align}
where the zeroth column of the lattice has initial and final occupation number $\alpha$, while all other columns have initial and final occupation number 0. The top $n$ rows of the lattice (with dark shading) correspond with $L$ matrices, while the bottom $n$ rows (with light shading) correspond with $L^{*}$ matrices. The right edge of the partition function is situated at infinity, so the only non-zero contribution comes from having the lowest $n$ edges unoccupied, and the highest $n$ edges occupied. This causes $\mathcal{P}$ to factorize into two parts:
\begin{align}
\label{pf-factorize}
\mathcal{P}(x_1,\dots,x_n;y_1,\dots,y_n;t;\alpha)
=
\begin{tikzpicture}[scale=0.6,baseline=(mid.base)]
\node (mid) at (1,4) {};
\foreach\x in {1,...,5}{
\foreach\y in {1,...,3}{
\light{\x}{\y};
}
\foreach\y in {4,...,6}{
\dark{\x}{\y};
}
}
\node at (1.5,0.5) {$\ss \alpha$};
\node at (2.5,0.5) {$\ss 0$};
\node at (3.5,0.5) {$\ss 0$};
\node at (4.5,0.5) {$\ss 0$};
\node at (5.5,0.5) {$\ss \cdots$};
\node at (1.5,7.5) {$\ss \alpha$};
\node at (2.5,7.5) {$\ss 0$};
\node at (3.5,7.5) {$\ss 0$};
\node at (4.5,7.5) {$\ss 0$};
\node at (5.5,7.5) {$\ss \cdots$};
\foreach\x in {1,...,3}{
\gbull{1}{0.5+\x}{0.09}
\ebull{1}{3.5+\x}{0.09}
}
\foreach\x in {1,...,3}{
\gbull{6}{3.5+\x}{0.09};
\ebull{6}{0.5+\x}{0.09};
}
\node at (0.6,1.5) {$\ss \red \b{y}_1$};
\node at (0.6,2.7) {$\ss \red \vdots$};
\node at (0.6,3.5) {$\ss \red \b{y}_n$};
\node at (0.6,4.5) {$\ss \red x_n$};
\node at (0.6,5.7) {$\ss \red \vdots$};
\node at (0.6,6.5) {$\ss \red x_1$};
\node at (6.5,4) {$\times$};
\foreach\x in {1,...,3}{
\draw[fline] (8.5+0.5*\x,2+0.5*\x) -- (7,3.5+\x);
\draw[fline] (7,0.5+\x) -- (10.5-0.5*\x,4+0.5*\x);
\gbull{8.5+0.5*\x}{2+0.5*\x}{0.09};
\ebull{10.5-0.5*\x}{4+0.5*\x}{0.09};
}
\foreach\x in {1,...,3}{
\gbull{7}{3.5+\x}{0.09};
\ebull{7}{0.5+\x}{0.09};
}
\end{tikzpicture}
\end{align}
where the product of $R$ matrix vertices is trivial, and gives a weight of $1$. The remaining non-trivial partition function in \eqref{pf-factorize} is the following vacuum to vacuum expectation value in the $t$-boson model:
\begin{align}
\label{hl-cauch-lhs}
\mathcal{P}(x_1,\dots,x_n;y_1,\dots,y_n;t;\alpha)
=
\bra{0;\alpha}
T^{*}_{-}(y_1) \dots T^{*}_{-}(y_n)
T_{+}(x_n) \dots T_{+}(x_1)
\ket{0;\alpha}
\end{align}
From Corollary \ref{cor:hl}, we immediately deduce that
\begin{align}
\label{hl-cauch-lhs2}
\mathcal{P}(x_1,\dots,x_n;y_1,\dots,y_n;t;\alpha)
=
\left(
\prod_{i=1}^{n} x_i y_i
\right)
\sum_{\lambda}
\prod_{j=1}^{m_0(\lambda)}
(1-t^{\alpha+j})
P_{\lambda}(x_1,\dots,x_n;t)
Q_{\lambda}(y_1,\dots,y_n;t)
\end{align}
which is the left hand side of \eqref{ref-cauch}, up to the overall factor $\prod_{i=1}^{n} x_i y_i$.

Returning to the original expression \eqref{pf-lhs}, we repeatedly apply the intertwining equation \eqref{rll*} to obtain the alternative expression
\begin{align}
\label{pf-rhs}
\mathcal{P}(x_1,\dots,x_n;y_1,\dots,y_n;t;\alpha)
=
\begin{tikzpicture}[scale=0.6,baseline=(mid.base)]
\node (mid) at (1,4) {};
\node at (-1.3,2.2) {$\ss \red \b{y}_1$};
\node at (-1.9,2.6) {$\ss \red \udiagdot$};
\node at (-2.3,3.2) {$\ss \red \b{y}_n$};
\node at (-2.2,4.8) {$\ss \red x_n$};
\node at (-1.8,5.4) {$\ss \red \ddiagdot$};
\node at (-1.2,5.8) {$\ss \red x_1$};
\foreach\x in {1,...,3}{
\draw[fline] (-2.5+0.5*\x,4+0.5*\x) -- (1,0.5+\x);
\draw[fline] (1,3.5+\x) -- (-0.5-0.5*\x,2+0.5*\x);
\ebull{-2.5+0.5*\x}{4+0.5*\x}{0.09};
\gbull{-0.5-0.5*\x}{2+0.5*\x}{0.09};
}
\foreach\x in {1,...,5}{
\foreach\y in {1,...,3}{
\dark{\x}{\y};
}
\foreach\y in {4,...,6}{
\light{\x}{\y};
}
}
\node at (1.5,0.5) {$\ss \alpha$};
\node at (2.5,0.5) {$\ss 0$};
\node at (3.5,0.5) {$\ss 0$};
\node at (4.5,0.5) {$\ss 0$};
\node at (5.5,0.5) {$\ss \cdots$};
\node at (1.5,7.5) {$\ss \alpha$};
\node at (2.5,7.5) {$\ss 0$};
\node at (3.5,7.5) {$\ss 0$};
\node at (4.5,7.5) {$\ss 0$};
\node at (5.5,7.5) {$\ss \cdots$};
\foreach\x in {1,...,3}{
\gbull{6}{0.5+\x}{0.09}
\ebull{6}{3.5+\x}{0.09}
}
\end{tikzpicture}
\end{align}
in which all $L$ matrices have been transferred to the bottom, and all $L^{*}$ matrices to the top. By particle-conservation arguments, it is easy to deduce that this partition function factorises immediately after the zeroth column:
\begin{align}
\label{rhs-factor}
\mathcal{P}(x_1,\dots,x_n;y_1,\dots,y_n;t;\alpha)
=
\begin{tikzpicture}[scale=0.6,baseline=(mid.base)]
\node (mid) at (1,4) {};
\node at (-1.3,2.2) {$\ss \red \b{y}_1$};
\node at (-1.9,2.6) {$\ss \red \udiagdot$};
\node at (-2.3,3.2) {$\ss \red \b{y}_n$};
\node at (-2.2,4.8) {$\ss \red x_n$};
\node at (-1.8,5.4) {$\ss \red \ddiagdot$};
\node at (-1.2,5.8) {$\ss \red x_1$};
\foreach\x in {1,...,3}{
\draw[fline] (-2.5+0.5*\x,4+0.5*\x) -- (1,0.5+\x);
\draw[fline] (1,3.5+\x) -- (-0.5-0.5*\x,2+0.5*\x);
\ebull{-2.5+0.5*\x}{4+0.5*\x}{0.09};
\gbull{-0.5-0.5*\x}{2+0.5*\x}{0.09};
}
\foreach\x in {1}{
\foreach\y in {1,...,3}{
\dark{\x}{\y};
}
\foreach\y in {4,...,6}{
\light{\x}{\y};
}
}
\node at (1.5,0.5) {$\ss \alpha$};
\node at (1.5,7.5) {$\ss \alpha$};
\foreach\x in {1,...,3}{
\gbull{2}{0.5+\x}{0.09}
\ebull{2}{3.5+\x}{0.09}
}
\node at (2.5,4) {$\times$};
\foreach\x in {3,...,6}{
\foreach\y in {1,...,3}{
\dark{\x}{\y};
}
\foreach\y in {4,...,6}{
\light{\x}{\y};
}
}
\node at (3.5,0.5) {$\ss 0$};
\node at (4.5,0.5) {$\ss 0$};
\node at (5.5,0.5) {$\ss 0$};
\node at (6.5,0.5) {$\ss \cdots$};
\node at (3.5,7.5) {$\ss 0$};
\node at (4.5,7.5) {$\ss 0$};
\node at (5.5,7.5) {$\ss 0$};
\node at (6.5,7.5) {$\ss \cdots$};
\foreach\x in {1,...,3}{
\draw[gline] (3,\x+0.5) -- (7,\x+0.5);
}
\foreach\x in {1,...,3}{
\gbull{3}{0.5+\x}{0.09};
\gbull{7}{0.5+\x}{0.09};
\ebull{3}{3.5+\x}{0.09};
\ebull{7}{3.5+\x}{0.09};
}
\end{tikzpicture}
\end{align}
with the remainder of the lattice (from the first column onward) being completely frozen and with total Boltzmann weight 1. To complete the proof, we show that the remaining non-trivial partition function in \eqref{rhs-factor} evaluates to the determinant on the right hand side of 
\eqref{ref-cauch}. We do this in the following lemma.

\subsection{Determinant evaluation of $Z_{\rm DW}$}

\begin{lem}
\label{lem:u-dwpf}
\begin{align}
\label{lem-u-dwpf}
\begin{tikzpicture}[scale=0.6,baseline=(mid.base)]
\node (mid) at (1,4) {};
\node at (-1.3,2.2) {$\ss \red \b{y}_1$};
\node at (-1.9,2.6) {$\ss \red \udiagdot$};
\node at (-2.3,3.2) {$\ss \red \b{y}_n$};
\node at (-2.2,4.8) {$\ss \red x_n$};
\node at (-1.8,5.4) {$\ss \red \ddiagdot$};
\node at (-1.2,5.8) {$\ss \red x_1$};
\foreach\x in {1,...,3}{
\draw[fline] (-2.5+0.5*\x,4+0.5*\x) -- (1,0.5+\x);
\draw[fline] (1,3.5+\x) -- (-0.5-0.5*\x,2+0.5*\x);
\ebull{-2.5+0.5*\x}{4+0.5*\x}{0.09};
\gbull{-0.5-0.5*\x}{2+0.5*\x}{0.09};
}
\foreach\x in {1}{
\foreach\y in {1,...,3}{
\dark{\x}{\y};
}
\foreach\y in {4,...,6}{
\light{\x}{\y};
}
}
\node at (1.5,0.5) {$\ss \alpha$};
\node at (1.5,7.5) {$\ss \alpha$};
\foreach\x in {1,...,3}{
\gbull{2}{0.5+\x}{0.09}
\ebull{2}{3.5+\x}{0.09}
}
\end{tikzpicture}
=
\frac{\prod_{i=1}^{n} (x_i y_i) \prod_{i,j=1}^{n} (1- t x_i y_j)}
{\prod_{1 \leq i<j \leq n} (x_i-x_j)(y_i-y_j)}
\det_{1 \leq i,j \leq n}
\left[
\frac{1-ut + (u-1)t x_i y_j}{(1-x_i y_j) (1-t x_i y_j)}
\right]
\end{align}
\end{lem}

\begin{proof}
We prove this using standard Lagrange interpolation techniques. Let $Z_{\rm DW}{(n)}$ denote the left hand side of \eqref{lem-u-dwpf} multiplied by\footnote{Note that $Z_{\rm DW}$ differs from the function 
$\mathcal{Z}_{\rm DW}$ in \eqref{1} by the overall normalization $\prod_{i,j=1}^{n}(1-x_i y_j)$, which makes $Z_{\rm DW}$ polynomial in its variables. We use similar notational conventions throughout the rest of the paper.} $\prod_{i,j=1}^{n} (1-x_iy_j) \prod_{k=1}^{n} \b{x}_k \b{y}_k$.  Then $Z_{\rm DW}{(n)}$ has the following properties:
\begin{enumerate}[label=\bf\arabic*.]
\item $Z_{\rm DW}{(n)}$ is symmetric in $\{x_1,\dots,x_n\}$ and $\{y_1,\dots,y_n\}$ separately.

\item $Z_{\rm DW}{(n)}$ is a polynomial in $x_n$ of degree $n$.

\item Setting $x_n = \b{y}_n$, we have the recursion relation
\begin{align*}
Z_{\rm DW}{(n)}
\Big|_{x_n = \b{y}_n}
=
(1-t)
\prod_{i=1}^{n-1}
(1-t x_i y_n)
(1-t y_i \b{y}_n)
Z_{\rm DW}{(n-1)},
\end{align*}
where $Z_{\rm DW}{(n-1)}$ is the re-normalized partition function of one size smaller.

\item Setting all variables to zero, we find
\begin{align*}
Z_{\rm DW}{(n)}
\Big|_{x_1,\dots,x_n = 0}
=
\prod_{i=1}^{n} (1-ut^i).
\end{align*}

\item The re-normalized partition function of size 1 has the explicit form
\begin{align*}
Z_{\rm DW}{(1)} = 1- ut + (u-1)t x_1 y_1.
\end{align*}
\end{enumerate}
Property {\bf 1} follows from the Yang-Baxter/intertwining equation. Property {\bf 2} comes from studying the form of the Boltzmann weights, with the further observation that the partition function in \eqref{lem-u-dwpf} has a common factor of $\prod_{i=1}^{n} x_i y_i$ (every row and column has at least one $c_{+}$ vertex). Properties {\bf 3} and {\bf 4} follow from simple freezing/combinatorial arguments in the six-vertex model, which we suppress here for the sake of brevity. Property {\bf 5} comes from a direct calculation of the partition function of size 1. Together, these five properties uniquely determine $Z_{\rm DW}{(n)}$, which can be deduced by the arguments of Appendix \ref{app:lagrange}. It can be easily verified that the right hand side of \eqref{lem-u-dwpf} satisfies the same properties.
\end{proof} 

It remains only to equate \eqref{hl-cauch-lhs2} and \eqref{lem-u-dwpf}, which gives \eqref{ref-cauch} after removing the spurious factor of $\prod_{i=1}^{n} x_i y_i$ from both sides of the equation.

\end{proof}

\subsection{Specialization to domain wall partition function}

\begin{cor}
Taking the $u=1$ ($\alpha = 0$) case of \eqref{ref-cauch}, we have
\begin{align*}
\sum_{\lambda}
\prod_{i=1}^{m_0(\lambda)}
(1 - t^{i})
b_{\lambda}(t)
P_{\lambda}(x_1,\dots,x_n;t)
P_{\lambda}(y_1,\dots,y_n;t)
=
Z_{\rm DW}(x_1,\dots,x_n;y_1,\dots,y_n;t)
\end{align*}
where $Z_{\rm DW}$ is the domain wall partition function of the six-vertex model \cite{ize,kor}.
\end{cor}

\section{Refined Littlewood identity}
\label{sec:little}

In this section we prove a refined version of a Littlewood identity for Hall--Littlewood polynomials, that was originally proved in \cite{bwz-j} by very different means. For a comprehensive study of Littlewood identities for Hall--Littlewood polynomials, we refer the reader to \cite{mac,war2} (our formula \eqref{ref-little} is a refinement of Example 3, Section 5, Chapter III of \cite{mac}) and to \cite{rai-war} for recent developments concerning boxed Littlewood formulae for Macdonald polynomials. 

\subsection{A property of the transfer matrices}

\begin{lem}
Let $\bra{\rm e;\alpha}$ be the following weighted sum over all partitions with even multiplicities:
\begin{align}
\label{e-vec}
\bra{\rm e;\alpha}
=
\sum_{\substack{
m_{i}(\lambda) \in 2 \mathbb{N} \\ m_0(\lambda) \in 2 \mathbb{Z} 
}}
\
c_{\lambda}(\alpha,t)
\bra{\lambda;\alpha},
\end{align}
where the sum is taken over all dual states
\begin{align*}
\bra{\lambda;\alpha}
=
\bra{m_0(\lambda)+\alpha}_0
\otimes
\bra{m_1(\lambda)}_1
\otimes
\bra{m_2(\lambda)}_2
\otimes
\cdots
\end{align*}
such that $m_0(\lambda) \in \{\dots,-2,0,2,\dots\} = 2\mathbb{Z}$ and $m_i(\lambda) \in \{0,2,4,\dots\} = 2\mathbb{N}$ for all $i \geq 1$, and where the summation coefficients are given by
\begin{align}
c_{\lambda}(\alpha,t) 
= 
\prod_{i=1}^{\infty}
\prod_{j=1}^{m_i(\lambda)/2} (1-t^{2j-1})
\times
\left\{
\begin{array}{cl}
\displaystyle{
\prod_{j=1}^{m_0(\lambda)/2}
(1-t^{\alpha+2j-1})
},
&
m_0(\lambda) \geq 0,
\\ \\ 
\displaystyle{
\prod_{j=1}^{-m_0(\lambda)/2}
\frac{1}{(1-t^{\alpha-2j+1})}
},
&
m_0(\lambda) \leq 0.
\end{array}
\right.
\end{align}
Then the transfer matrices satisfy 
\begin{align}
\label{lem-even}
\bra{\rm e;\alpha}
T_{-}(x)
=
\bra{\rm e;\alpha}
T^{*}_{+}(x),
\qquad
\bra{\rm e;\alpha}
T_{+}(x)
=
\bra{\rm e;\alpha}
T^{*}_{-}(x).
\end{align}

\end{lem}

\begin{proof}
We will show that for any partition state
\begin{align*}
|\mu; \alpha \rangle
=
\ket{m_0(\mu) + \alpha}_0
\otimes
\ket{m_1(\mu)}_1
\otimes
\ket{m_2(\mu)}_2
\otimes
\cdots
\end{align*}
with $m_0(\mu) \in \mathbb{Z}$ and $m_i(\mu) \in \mathbb{N}$ for all $i \geq 1$, one has
\begin{align}
\label{lem-even2}
\bra{\rm e;\alpha}
T_{-}(x)
\ket{\mu;\alpha}
=
\bra{\rm e;\alpha}
T^{*}_{+}(x)
\ket{\mu;\alpha},
\qquad
\bra{\rm e;\alpha}
T_{+}(x)
\ket{\mu;\alpha}
=
\bra{\rm e;\alpha}
T^{*}_{-}(x)
\ket{\mu;\alpha}.
\end{align}
Our starting point is the simple observation that for any partition state $\ket{\mu;\alpha}$, there is a unique way to act with $T_{\pm}(x)$/$T^{*}_{\pm}(x)$ to form a partition $\ket{\lambda;\alpha}$ whose multiplicities $m_i(\lambda)$ are all even. We denote this unique partition by 
$\ket{\mu_{\pm};\alpha}$, with $\pm$ corresponding to a partition which is bigger/smaller (respectively) than the original. Since the transfer matrix elements $T_{\pm}(x)$ act to create bigger partitions, while $T^{*}_{\pm}(x)$ generate smaller partitions, the expectation values in \eqref{lem-even2} simplify greatly and we need only show that
\begin{align}
\label{lem-even3}
c_{\mu_{+}}(\alpha,t)
\bra{\mu_{+};\alpha}
T_{-}(x)
\ket{\mu;\alpha}
&=
c_{\mu_{-}}(\alpha,t)
\bra{\mu_{-};\alpha}
T^{*}_{+}(x)
\ket{\mu;\alpha}
\\
\label{lem-even4}
c_{\mu_{+}}(\alpha,t)
\bra{\mu_{+};\alpha}
T_{+}(x)
\ket{\mu;\alpha}
&=
c_{\mu_{-}}(\alpha,t)
\bra{\mu_{-};\alpha}
T^{*}_{-}(x)
\ket{\mu;\alpha}
\end{align}
Let us focus firstly on the proof of \eqref{lem-even3}. To better illustrate this equation, we give an explicit example with $\mu = (6,4,4,2,1,1)$, $m_0(\mu) = 0$. We have $\mu_{+} = (6,6,4,4,1,1)$, $m_0(\mu_{+}) = 0$ and $\mu_{-} = (4,4,2,2,1,1)$, $m_0(\mu_{-}) = 0$. One finds that
\begin{align*}
c_{\mu_{+}}(\alpha,t) 
\times
\begin{tikzpicture}[baseline=0.25cm,scale=0.7]
\node at (-0.5,0.5) {$\ss \red x$};
\foreach\x in {0,...,7}{
\dark{\x}{0};
}
\draw[gline] (0,0.5) -- (1.3,0.5) -- (1.3,1);
\draw[gline] (1.5,0) -- (1.5,1);
\draw[gline] (1.7,0) -- (1.7,0.5) -- (2.5,0.5) -- (2.5,1);
\draw[gline] (4.4,0) -- (4.4,1); \draw[gline] (4.6,0) -- (4.6,1);
\draw[gline] (6.4,0) -- (6.4,1); \draw[gline] (6.6,0) -- (6.6,0.5) -- (8,0.5);
\gbull{0}{0.5}{0.07}; 
\gbull{1}{0.5}{0.07};
\gbull{2}{0.5}{0.07};
\ebull{3}{0.5}{0.07}; 
\ebull{4}{0.5}{0.07};
\ebull{5}{0.5}{0.07};
\ebull{6}{0.5}{0.07};
\gbull{7}{0.5}{0.07};
\gbull{8}{0.5}{0.07};
\gbull{1.3}{1}{0.07}; \gbull{1.5}{1}{0.07};
\gbull{2.5}{1}{0.07};
\gbull{4.4}{1}{0.07}; \gbull{4.6}{1}{0.07};
\gbull{6.4}{1}{0.07};
\gbull{1.5}{0}{0.07}; \gbull{1.7}{0}{0.07};
\gbull{4.4}{0}{0.07}; \gbull{4.6}{0}{0.07};
\gbull{6.4}{0}{0.07}; \gbull{6.6}{0}{0.07};
\end{tikzpicture}
&=
c_{\mu_{-}}(\alpha,t) 
\times
\begin{tikzpicture}[baseline=0.25cm,scale=0.7]
\node at (-0.5,0.5) {$\ss \red \b{x}$};
\foreach\x in {0,...,7}{
\light{\x}{0};
}
\draw[gline] (1.4,0) -- (1.4,1); \draw[gline] (1.6,0) -- (1.6,1);
\draw[gline] (2.4,0) -- (2.4,1); \draw[gline] (2.6,0) -- (2.6,0.5) -- (4.3,0.5) -- (4.3,1);
\draw[gline] (4.5,0) -- (4.5,1);
\draw[gline] (4.7,0) -- (4.7,0.5) -- (6.5,0.5) -- (6.5,1);
\ebull{0}{0.5}{0.07}; 
\ebull{1}{0.5}{0.07};
\ebull{2}{0.5}{0.07};
\gbull{3}{0.5}{0.07}; 
\gbull{4}{0.5}{0.07};
\gbull{5}{0.5}{0.07};
\gbull{6}{0.5}{0.07};
\ebull{7}{0.5}{0.07};
\ebull{8}{0.5}{0.07};
\gbull{1.4}{1}{0.07}; \gbull{1.6}{1}{0.07};
\gbull{2.4}{1}{0.07};
\gbull{4.3}{1}{0.07}; \gbull{4.5}{1}{0.07};
\gbull{6.5}{1}{0.07};
\gbull{1.4}{0}{0.07}; \gbull{1.6}{0}{0.07};
\gbull{2.4}{0}{0.07}; \gbull{2.6}{0}{0.07};
\gbull{4.5}{0}{0.07}; \gbull{4.7}{0}{0.07};
\end{tikzpicture}
\end{align*}
as required. For a general partition state $\ket{\mu;\alpha}$, it is easy to deduce that both sides of 
\eqref{lem-even3} give the same power of $x$. Furthermore, if $m_i(\mu)$ is even, then $m_i(\mu) = m_i(\mu_{+}) = m_i(\mu_{-})$, meaning that site $i$ of the lattice produces no $t$-dependent factors on either side of \eqref{lem-even3}. If $m_i(\mu)$ is odd, then either $m_i(\mu_{\pm}) = m_i(\mu) \pm 1$ or $m_i(\mu_{\pm}) = m_i(\mu) \mp 1$, and we use either the relation
\begin{align*}
\left(1-t^{\alpha_i+m_i(\mu)}\right)
\begin{tikzpicture}[baseline=0]
\dark{-0.5}{-0.5}; 
\draw[gbline] (0,-0.5) node[below] {$\ss \alpha_i+m_i(\mu)+1$} -- (0,0.5) node[above] 
{$\ss \alpha_i+m_i(\mu)$};
\draw[gline] (0.05,-0.5) -- (0.05,0) -- (0.5,0);
\end{tikzpicture}
=
\begin{tikzpicture}[baseline=0] 
\light{-0.5}{-0.5}; 
\draw[gbline] (0,-0.5) node[below] {$\ss \alpha_i+m_i(\mu)-1$} -- (0,0.5) node[above] 
{$\ss \alpha_i+m_i(\mu)$};
\draw[gline] (-0.05,0.5) -- (-0.05,0) -- (-0.5,0);
\end{tikzpicture}
\qquad
\text{or}
\qquad
\begin{tikzpicture}[baseline=0] 
\dark{-0.5}{-0.5}; 
\draw[gbline] (0,-0.5) node[below] {$\ss \alpha_i+m_i(\mu)-1$} -- (0,0.5) node[above] 
{$\ss \alpha_i+m_i(\mu)$};
\draw[gline] (-0.05,0.5) -- (-0.05,0) -- (-0.5,0);
\end{tikzpicture}
=
\left(1-t^{\alpha_i+m_i(\mu)}\right)
\begin{tikzpicture}[baseline=0]
\light{-0.5}{-0.5}; 
\draw[gbline] (0,-0.5) node[below] {$\ss \alpha_i+m_i(\mu)+1$} -- (0,0.5) node[above] 
{$\ss \alpha_i+m_i(\mu)$};
\draw[gline] (0.05,-0.5) -- (0.05,0) -- (0.5,0);
\end{tikzpicture}
\end{align*}
to compensate for the difference between $c_{\mu_{+}}(\alpha,t)$ and 
$c_{\mu_{-}}(\alpha,t)$ at that site, where $\alpha_0 = \alpha$ and $\alpha_i = 0$ for all $i \geq 1$. The proof of \eqref{lem-even4} is completely analogous.

\end{proof}

We shall represent equations \eqref{lem-even} graphically as follows:
\begin{align}
\label{lem-even5}
&
\begin{tikzpicture}[scale=0.6,baseline=(mid.base)]
\node (mid) at (1,1.3) {};
\foreach\x in {1,...,5}{
\dark{\x}{1};
}
\node at (3.5,0.5) {$\ss \bra{\rm e;\alpha}$};
\gbull{1}{1.5}{0.09};
\gbull{6}{1.5}{0.09};
\node at (0.7,1.5) {$\ss \red x$};
\end{tikzpicture}
\quad
=
\quad
\begin{tikzpicture}[scale=0.6,baseline=(mid.base)]
\node (mid) at (1,1.3) {};
\draw[cross=0.5] (1,1.5) -- (-1,1.5); 
\foreach\x in {1,...,5}{
\light{\x}{1};
}
\node at (3.5,0.5) {$\ss \bra{\rm e;\alpha}$};
\gbull{-1}{1.5}{0.09};
\ebull{6}{1.5}{0.09};
\node at (0.7,1.7) {$\ss \red \b{x}$};
\end{tikzpicture}
\\
\label{lem-even6}
&
\begin{tikzpicture}[scale=0.6,baseline=(mid.base)]
\node (mid) at (1,1.3) {};
\foreach\x in {1,...,5}{
\dark{\x}{1};
}
\node at (3.5,0.5) {$\ss \bra{\rm e;\alpha}$};
\ebull{1}{1.5}{0.09}
\gbull{6}{1.5}{0.09}
\node at (0.7,1.5) {$\ss \red x$};
\end{tikzpicture}
\quad
=
\quad
\begin{tikzpicture}[scale=0.6,baseline=(mid.base)]
\node (mid) at (1,1.3) {};
\draw[cross=0.5] (-1,1.5) -- (1,1.5); 
\foreach\x in {1,...,5}{
\light{\x}{1};
}
\node at (3.5,0.5) {$\ss \bra{\rm e;\alpha}$};
\ebull{-1}{1.5}{0.09};
\ebull{6}{1.5}{0.09};
\node at (0.7,1.7) {$\ss \red \b{x}$};
\end{tikzpicture}
\end{align}
where the cross indicates a spin-flipping operator, which converts a particle to a hole and vice versa (note that it also inverts the spectral parameter). From this it is clear that \eqref{lem-even5} and \eqref{lem-even6} are in fact components of the same equation, in which we do not explicitly specify the state on the left boundary.

\subsection{Lattice formulation of refined Littlewood identity}

\begin{thm}
\label{thm:refined-little}
Hall--Littlewood polynomials satisfy the following refined Littlewood identity:
\begin{multline}
\label{ref-little}
\sum_{
\lambda: m_i(\lambda) \in 2 \mathbb{N}
}
\
\prod_{k=1}^{m_0(\lambda)/2}
(1-u t^{2k-1})
\prod_{i=1}^{\infty}
\prod_{j=1}^{m_i(\lambda)/2}
(1-t^{2j-1})
P_{\lambda}(x_1,\dots,x_{2n};t)
=
\\
\prod_{1 \leq i<j \leq 2n}
\left(
\frac{1-t x_i x_j}{x_i - x_j}
\right)
\pf_{1\leq i < j \leq 2n}
\left[
\frac{(x_i - x_j) (1-ut + (u-1)t x_i x_j)}
{(1-x_i x_j) (1-t x_i x_j)}
\right],
\end{multline}
where the sum is taken over all partitions with even, non-negative multiplicities $m_i(\lambda)$ for all $i \geq 0$. As before, $u$ is a refining parameter whose value is arbitrary.
\end{thm}

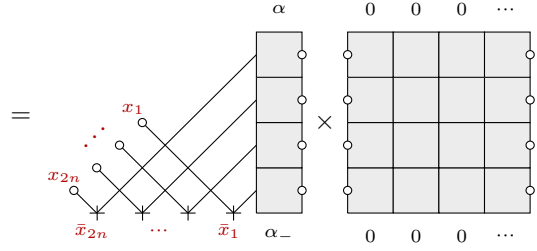
\begin{proof}
As we did previously, we parametrize $u = t^{\alpha}$. We consider the following partition in the $t$-boson model:
\begin{align}
\label{pf-little-lhs}
\mathcal{P}(x_1,\dots,x_{2n};t;\alpha)
=
\begin{tikzpicture}[scale=0.6,baseline=(mid.base)]
\node (mid) at (1,3) {};
\foreach\x in {1,...,5}{
\foreach\y in {1,...,4}{
\dark{\x}{\y};
}
}
\node at (1.5,5.5) {$\ss \alpha$};
\node at (2.5,5.5) {$\ss 0$};
\node at (3.5,5.5) {$\ss 0$};
\node at (4.5,5.5) {$\ss 0$};
\node at (5.5,5.5) {$\ss \cdots$};
\node at (3.5,0.5) {$\ss \bra{\rm e;\alpha}$};
\foreach\x in {1,...,4}{
\ebull{1}{0.5+\x}{0.09}
\gbull{6}{0.5+\x}{0.09}
}
\node at (0.5,4.5) {$\ss \red x_1$};
\node at (0.5,3.3) {$\ss \red \vdots$};
\node at (0.5,1.5) {$\ss \red x_{2n}$};
\end{tikzpicture}
\end{align}
From \eqref{hl-P} and the definition \eqref{e-vec} of $\bra{\rm e;\alpha}$ it is immediate that
\begin{align}
\label{hl-little-lhs}
\mathcal{P}(x_1,\dots,x_{2n};t;\alpha)
=
\left( \prod_{i=1}^{2n} x_i \right)
\sum_{
\lambda: m_i(\lambda) \in 2 \mathbb{N}
}
\
\prod_{k=1}^{m_0(\lambda)/2}
(1-u t^{2k-1})
\prod_{i=1}^{\infty}
\prod_{j=1}^{m_i(\lambda)/2}
(1-t^{2j-1})
P_{\lambda}(x_1,\dots,x_{2n};t),
\end{align}
which is equal to the left hand side of \eqref{ref-little}, up to the extra factor 
$\prod_{i=1}^{2n} x_i$. It is important to point out that although $\bra{{\rm e};\alpha}$ is summed over states $\lambda$ for which $m_i(\lambda)$ is negative, these states do not contribute to the final expression \eqref{hl-little-lhs}. This is easily deduced by studying the partition function \eqref{pf-little-lhs}. Indeed, if the occupation number at the base of the zeroth column is $\alpha - 2k$ for some $k \geq 1$, it is impossible to obtain the occupation number $\alpha$ at the top of this column, since no particles are incident on the left edge. 

On the other hand, applying \eqref{lem-even6} once at the base of the lattice, we find that
\begin{align*}
\mathcal{P}(x_1,\dots,x_{2n};t;\alpha)
=
\begin{tikzpicture}[scale=0.6,baseline=(mid.base)]
\node (mid) at (1,3) {};
\foreach\x in {2,...,4}{
\draw[fline] (-1.5,\x+0.5) -- (1,\x+0.5);
}
\draw[fline,cross=0.5] (-1.5,1.5) -- (1,1.5);
\foreach\x in {1,...,5}{
\foreach\y in {2,...,4}{
\dark{\x}{\y};
}
}
\foreach\x in {1,...,5}{
\light{\x}{1};
}
\node at (1.5,5.5) {$\ss \alpha$};
\node at (2.5,5.5) {$\ss 0$};
\node at (3.5,5.5) {$\ss 0$};
\node at (4.5,5.5) {$\ss 0$};
\node at (5.5,5.5) {$\ss \cdots$};
\node at (3.5,0.5) {$\ss \bra{\rm e;\alpha}$};
\foreach\x in {1,...,4}{
\ebull{-1.5}{0.5+\x}{0.09}
}
\ebull{6}{1.5}{0.09}
\foreach\x in {2,...,4}{
\gbull{6}{0.5+\x}{0.09}
}
\node at (0.5,4.7) {$\ss \red x_1$};
\node at (0.5,4.1) {$\ss \red \vdots$};
\node at (0.4,2.7) {$\ss \red x_{2n-1}$};
\node at (0.5,1.7) {$\ss \red \b{x}_{2n}$};
\end{tikzpicture}
=
\begin{tikzpicture}[scale=0.6,baseline=(mid.base)]
\node (mid) at (1,3) {};
\foreach\x in {2,...,4}{
\draw[fline] (-1.5,\x+0.5) -- (1,\x+0.5);
}
\draw[fline,cross=0.5] (-1.5,1.5) -- (1,1.5);
\foreach\x in {1,...,5}{
\foreach\y in {2,...,4}{
\dark{\x}{\y};
}
}
\foreach\x in {1,...,5}{
\light{\x}{1};
}
\node at (1.5,5.5) {$\ss \alpha$};
\node at (2.5,5.5) {$\ss 0$};
\node at (3.5,5.5) {$\ss 0$};
\node at (4.5,5.5) {$\ss 0$};
\node at (5.5,5.5) {$\ss \cdots$};
\node at (3.5,0.5) {$\ss \bra{\rm e;\alpha}$};
\draw[fline] (7,4.5) -- (7,5.5);
\draw[fline] (6,1.5) -- (7,1.5) -- (7,4.5);
\foreach\x in {2,...,4}{
\draw[fline] (8,\x+0.5) -- (6,\x+0.5);
}
\foreach\x in {1,...,4}{
\ebull{-1.5}{0.5+\x}{0.09}
}
\ebull{7}{5.5}{0.09}
\foreach\x in {2,...,4}{
\gbull{8}{0.5+\x}{0.09}
}
\node at (0.5,4.7) {$\ss \red x_1$};
\node at (0.5,4.1) {$\ss \red \vdots$};
\node at (0.4,2.7) {$\ss \red x_{2n-1}$};
\node at (0.5,1.7) {$\ss \red \b{x}_{2n}$};
\end{tikzpicture}
\end{align*}
where we have introduced a product of $R$ matrix vertices at the right edge, which are completely frozen by the vanishing property of $T(x)$ when it has no particle at its far right edge, and produce an overall Boltzmann weight of 1. Using the intertwining equation \eqref{rll*}, we transform this as follows:
\begin{align*}
\mathcal{P}(x_1,\dots,x_{2n};t;\alpha)
=
\begin{tikzpicture}[scale=0.6,baseline=(mid.base)]
\node (mid) at (1,3) {};
\foreach\x in {1,...,3}{
\draw[fline] (-1,\x+0.5) -- (1,\x+0.5);
}
\draw[fline,cross=1] (-1,0.5) -- (0,0.5); 
\draw[fline] (0,0.5) -- (0,4.5) -- (1,4.5);
\foreach\x in {1,...,5}{
\foreach\y in {1,...,3}{
\dark{\x}{\y};
}
}
\foreach\x in {1,...,5}{
\light{\x}{4};
}
\node at (1.5,5.5) {$\ss \alpha$};
\node at (2.5,5.5) {$\ss 0$};
\node at (3.5,5.5) {$\ss 0$};
\node at (4.5,5.5) {$\ss 0$};
\node at (5.5,5.5) {$\ss \cdots$};
\node at (3.5,0.5) {$\ss \bra{\rm e;\alpha}$};
\foreach\x in {1,...,4}{
\ebull{-1}{-0.5+\x}{0.09}
}
\foreach\x in {1,...,3}{
\gbull{6}{0.5+\x}{0.09}
}
\ebull{6}{4.5}{0.09}
\node at (-1.4,3.5) {$\ss \red x_1$};
\node at (-1.4,2.6) {$\ss \red \vdots$};
\node at (-1.7,1.5) {$\ss \red x_{2n-1}$};
\node at (0,0.1) {$\ss \red \b{x}_{2n}$};
\end{tikzpicture}
\end{align*}
Iterating this procedure over the remaining $T(x)$ rows, we transform them all into $T^{*}(x)$ rows, and obtain the following expression:
\begin{align*}
\mathcal{P}(x_1,\dots,x_{2n};t;\alpha)
&=
\begin{tikzpicture}[scale=0.6,baseline=(mid.base)]
\node (mid) at (1,3) {};
\foreach\x in {1,...,4}{
\draw[fline,cross=1] (1,0.5+\x) -- (1.5-\x,1);
}
\foreach\x in {1,...,4}{
\draw[fline] (-1-0.5*\x,3.5-0.5*\x) -- (1.5-\x,1);
\ebull{-1-0.5*\x}{3.5-0.5*\x}{0.09}
}
\foreach\x in {1,...,5}{
\foreach\y in {1,...,4}{
\light{\x}{\y};
}
}
\node at (1.5,5.5) {$\ss \alpha$};
\node at (2.5,5.5) {$\ss 0$};
\node at (3.5,5.5) {$\ss 0$};
\node at (4.5,5.5) {$\ss 0$};
\node at (5.5,5.5) {$\ss \cdots$};
\node at (3.5,0.5) {$\ss \bra{\rm e;\alpha}$};
\foreach\x in {1,...,4}{
\ebull{6}{0.5+\x}{0.09}
}
\node at (-1.7,3.3) {$\ss \red x_{1}$};
\node at (-2.5,2.7) {$\ss \red \ddiagdot$};
\node at (-3.2,1.8) {$\ss \red x_{2n}$};
\node at (-2.6,0.6) {$\ss \red \b{x}_{2n}$};
\node at (-1.1,0.6) {$\ss \red \cdots$};
\node at (0.4,0.6) {$\ss \red \b{x}_{1}$};
\end{tikzpicture}
\\
&=
\begin{tikzpicture}[scale=0.6,baseline=(mid.base)]
\node (mid) at (1,3) {};
\foreach\x in {1,...,4}{
\draw[fline,cross=1] (1,0.5+\x) -- (1.5-\x,1);
}
\foreach\x in {1,...,4}{
\draw[fline] (-1-0.5*\x,3.5-0.5*\x) -- (1.5-\x,1);
\ebull{-1-0.5*\x}{3.5-0.5*\x}{0.09}
}
\foreach\y in {1,...,4}{
\light{1}{\y};
}
\node at (1.5,5.5) {$\ss \alpha$};
\node at (1.5,0.5) {$\ss \alpha_{-}$};
\node at (2.5,3) {$\times$};
\foreach\x in {3,...,6}{
\foreach\y in {1,...,4}{
\light{\x}{\y};
}
}
\node at (3.5,5.5) {$\ss 0$};
\node at (4.5,5.5) {$\ss 0$};
\node at (5.5,5.5) {$\ss 0$};
\node at (6.5,5.5) {$\ss \cdots$};
\node at (3.5,0.5) {$\ss 0$};
\node at (4.5,0.5) {$\ss 0$};
\node at (5.5,0.5) {$\ss 0$};
\node at (6.5,0.5) {$\ss \cdots$};
\foreach\x in {1,...,4}{
\ebull{2}{0.5+\x}{0.09}
\ebull{3}{0.5+\x}{0.09}
\ebull{7}{0.5+\x}{0.09}
}
\node at (-1.7,3.3) {$\ss \red x_{1}$};
\node at (-2.5,2.7) {$\ss \red \ddiagdot$};
\node at (-3.2,1.8) {$\ss \red x_{2n}$};
\node at (-2.6,0.6) {$\ss \red \b{x}_{2n}$};
\node at (-1.1,0.6) {$\ss \red \cdots$};
\node at (0.4,0.6) {$\ss \red \b{x}_{1}$};
\end{tikzpicture}
\end{align*}
where the lattice factorizes into a partition function featuring the zeroth column, and a trivial region, with Boltzmann weight 1, formed by the first column onwards. The dual vector at the base of the zeroth column is given explicitly by
\begin{align*}
\bra{\alpha_{-}}
=
\sum_{k=0}^{\infty}
\prod_{j=1}^{k}
\frac{1}{1-u t^{1-2j}}
\bra{\alpha-2k}_0,
\end{align*}
where we now sum over {\it non-positive} shifts of $\alpha$. Indeed, note that the state variable at the base of the zeroth column cannot be of the form $\alpha+2k$ where $k \geq 1$, since the state variable at the top of that column is $\alpha$ and there are no particles incident on its right edge. To conclude the proof, we show that the remaining non-trivial partition function is equal to the right hand side of \eqref{ref-little}.

\subsection{Pfaffian evaluation of $Z_{\rm OS}$}

\begin{lem}
\begin{align}
\label{lem-u-osasm}
\begin{tikzpicture}[scale=0.6,baseline=(mid.base)]
\node (mid) at (1,3) {};
\foreach\x in {1,...,4}{
\draw[fline,cross=1] (1,0.5+\x) -- (1.5-\x,1);
}
\foreach\x in {1,...,4}{
\draw[fline] (-1-0.5*\x,3.5-0.5*\x) -- (1.5-\x,1);
\ebull{-1-0.5*\x}{3.5-0.5*\x}{0.09}
}
\foreach\y in {1,...,4}{
\light{1}{\y};
}
\node at (1.5,5.5) {$\ss \alpha$};
\node at (1.5,0.5) {$\ss \alpha_{-}$};
\foreach\x in {1,...,4}{
\ebull{2}{0.5+\x}{0.09}
}
\node at (-1.7,3.3) {$\ss \red x_{1}$};
\node at (-2.5,2.7) {$\ss \red \ddiagdot$};
\node at (-3.2,1.8) {$\ss \red x_{2n}$};
\node at (-2.6,0.6) {$\ss \red \b{x}_{2n}$};
\node at (-1.1,0.6) {$\ss \red \cdots$};
\node at (0.4,0.6) {$\ss \red \b{x}_{1}$};
\end{tikzpicture}
=
\left(
\prod_{i=1}^{2n}
x_i
\right)
\times
\prod_{1 \leq i<j \leq 2n}
\left(
\frac{1-t x_i x_j}{x_i - x_j}
\right)
\pf_{1\leq i < j \leq 2n}
\left[
\frac{(x_i - x_j) (1-ut + (u-1)t x_i x_j)}
{(1-x_i x_j) (1-t x_i x_j)}
\right]
\end{align}
\end{lem}

\begin{proof}
Let $Z_{\rm OS}(2n)$ denote the left hand side of \eqref{lem-u-osasm}, multiplied by 
$\prod_{1 \leq i<j \leq 2n} (1-x_i x_j) \prod_{k=1}^{2n} \b{x}_k$. We write down a list of properties which this re-normalized partition function satisfies:
\begin{enumerate}[label=\bf\arabic*.]
\item $Z_{\rm OS}(2n)$ is symmetric in $\{x_1,\dots,x_{2n}\}$.

\item $Z_{\rm OS}(2n)$ is a polynomial in $x_{2n}$ of degree $2n-1$.

\item Setting $x_{2n} = \b{x}_{2n-1}$, we have the recursion relation
\begin{align*}
Z_{\rm OS}(2n)
\Big|_{x_{2n} = \b{x}_{2n-1}}
=
(1-t)
\prod_{i=1}^{2n-2}
(1-t x_i \b{x}_{2n-1}) (1-t x_i x_{2n-1})
Z_{\rm OS}(2n-2),
\end{align*}
where $Z_{\rm OS}(2n-2)$ is the re-normalized partition function of two sizes smaller.

\item Setting all variables to zero, we have
\begin{align*}
Z_{\rm OS}(2n)
\Big|_{x_1,\dots,x_{2n} = 0}
=
\prod_{i=1}^{n} (1-ut^{2i-1}).
\end{align*}

\item The smallest possible re-normalized partition function, of size 2, is given explicitly by
\begin{align*}
Z_{\rm OS}(2)
=
1 - ut + (u-1)t x_1 x_2.
\end{align*}
 
\end{enumerate}
These properties are proved in a very similar way to those of Lemma \ref{lem:u-dwpf}. They are uniquely-determining, again by Lagrange interpolation arguments, and one can easily show that the right hand side of \eqref{lem-u-osasm} obeys them.
\end{proof}

\end{proof}

\subsection{Specialization to off-diagonally symmetric alternating sign matrices}

\begin{cor}
Taking the $u=1$ ($\alpha = 0$) case of \eqref{ref-little}, we have
\begin{align*}
\sum_{
\lambda: m_i(\lambda) \in 2 \mathbb{N}
}
\
\prod_{i=0}^{\infty}
\prod_{j=1}^{m_i(\lambda)/2}
(1-t^{2j-1})
P_{\lambda}(x_1,\dots,x_{2n};t)
=
Z_{\rm OS}(x_1,\dots,x_{2n};t)
\end{align*}
where $Z_{\rm OS}$ is the partition function of the six-vertex model corresponding with off-diagonally symmetric alternating-sign matrices \cite{kup2}.
\end{cor}

\section{Reflecting Cauchy identity}
\label{sec:refl-cauchy}

In this section we prove a result that was conjectured in \cite{bw}. It is a Cauchy-type identity, whose left hand side is an infinite sum of a product of an $A_n$ and a $BC_n$ Hall--Littlewood polynomial. The right hand side is a determinant, and equal to the partition function $Z_{\rm U}$ of U-turn ASMs \cite{tsu,kup2}.

Unlike the identities discussed in the previous sections, we do not include a refining parameter $u$ in \eqref{ref-reflect-cauch}, even though there appears to be no immediate obstacle to doing so. Indeed it is natural to follow the procedure of the preceding sections, and fix the incoming/outgoing occupation numbers of the zeroth column to $\alpha \not= 0$, in this way introducing an extra parameter $u = t^{\alpha}$. While this procedure is straightforward for $A_n$ Hall--Littlewood polynomials, with the dependence on $u$ factorizing out as in Corollary \ref{cor:hl}, for $BC_n$ Hall--Littlewood polynomials the $u$ parameter becomes embedded non-trivially in the polynomial itself (due to the presence of the boundary covectors). We thus obtain a function which is a one-parameter refinement of $BC_n$ Hall--Littlewood polynomials. 

It remains to ask whether such a refined polynomial is already known in the literature. We expected that the polynomial in question would be the $q=0$ specialization of a ``lifted'' Koornwinder polynomial as defined by Rains in \cite{rai}, with $u$ playing the role of the ``lifting'' parameter (denoted $T$ in \cite{rai}). This expectation was based on the fact that we have already conjectured a $u$-refinement of \eqref{ref-reflect-cauch} in \cite{bwz-j}, in which the lifted Koornwinder polynomials make an appearance, and that this more difficult conjecture should also be treatable in the framework of this paper. However, it turns out that this natural approach is not the correct one, and including the $u$ parameter in such a na\"ive way does not lead to lifted Koornwinder polynomials. For this reason we do not present a refined version of \eqref{ref-reflect-cauch}, and the conjecture of \cite{bwz-j} remains open.

\subsection{Lattice formulation of reflecting Cauchy identity}

\begin{thm}
$BC_n$ and $A_n$ Hall--Littlewood polynomials satisfy the following refined Cauchy identity:
\begin{multline}
\label{ref-reflect-cauch}
\sum_{\lambda} 
\prod_{i=1}^{m_0(\lambda)} 
(1-t^i)
b_{\lambda}(t)
K_{\lambda}(x^{\pm1}_1,\dots,x^{\pm1}_n;t) 
P_{\lambda}(y_1,\dots,y_n;t) 
= 
\\
\frac{\displaystyle{
\prod_{i,j=1}^{n} 
(1-t x_i y_j) (1-t \b{x}_i y_j) 
}}
{\displaystyle{
\prod_{1\leq i<j \leq n} (x_i-x_j) (y_i-y_j)(1 - \b{x}_i \b{x}_j)(1 - t y_i y_j)
}}
\det_{1\leq i,j \leq n}
\left[
\frac{(1-t)}{(1- x_i y_j)(1-\b{x}_i y_j)(1-t x_i y_j)(1-t \b{x}_i y_j)}
\right]
\end{multline}
\end{thm}

\begin{proof}
We follow a very similar procedure to the proof of Theorem \ref{thm-cauch}. Consider the following partition function in the $t$-boson model:
\begin{align}
\label{cauch-refl-1}
\mathcal{P}(x^{\pm1}_1,\dots,x^{\pm1}_n;y_1,\dots,y_n;t)
=
\begin{tikzpicture}[scale=0.6,baseline=(mid.base)]
\node (mid) at (1,4) {};
\foreach\x in {2,...,6}{
\foreach\y in {1,...,2}{
\light{\x}{\y};
}
\foreach\y in {3,...,6}{
\dark{\x}{\y};
}
}
\foreach\x in {1,...,2}{
\draw[smooth] (2,2.5+2*\x) arc (90:270:0.5);
\node at (1.5,2+2*\x) {\fs$\bullet$};
}
%
\node at (1.3,6.5) {$\ss \red x_1$};
\node at (1.3,5.5) {$\ss \red \b{x}_1$};
\node at (1.8,5.15) {$\ss \red \vdots$};
\node at (1.3,4.5) {$\ss \red x_n$};
\node at (1.3,3.5) {$\ss \red \b{x}_n$};
\node at (1.3,2.5) {$\ss \red \b{y}_n$};
\node at (1.8,2.15) {$\ss \red \vdots$};
\node at (1.3,1.5) {$\ss \red \b{y}_1$};
\node at (2.5,0.5) {$\ss 0$};
\node at (3.5,0.5) {$\ss 0$};
\node at (4.5,0.5) {$\ss 0$};
\node at (5.5,0.5) {$\ss 0$};
\node at (6.5,0.5) {$\ss \cdots$};
\node at (2.5,7.5) {$\ss 0$};
\node at (3.5,7.5) {$\ss 0$};
\node at (4.5,7.5) {$\ss 0$};
\node at (5.5,7.5) {$\ss 0$};
\node at (6.5,7.5) {$\ss \cdots$};
\foreach\x in {1,...,2}{
\gbull{2}{0.5+\x}{0.09}
}
\foreach\x in {0,1,...,3}{
\draw[fline] (8.5+0.5*\x,2+0.5*\x) -- (7,3.5+\x);
\gbull{8.5+0.5*\x}{2+0.5*\x}{0.09}
}
\foreach\x in {1,2}{
\draw[fline] (7,0.5+\x) -- (10.5-0.5*\x,4+0.5*\x);
\ebull{10.5-0.5*\x}{4+0.5*\x}{0.09}
}
\end{tikzpicture}
\end{align}
Recalling that $T^{*}(y)$ vanishes unless its right edge is unoccupied by a particle, we find that the attached lattice of $R$ matrices freezes (with total Boltzmann weight 1) and we are left with
\begin{align}
\mathcal{P}(x^{\pm1}_1,\dots,x^{\pm1}_n;y_1,\dots,y_n;t)
=
\begin{tikzpicture}[scale=0.6,baseline=(mid.base)]
\node (mid) at (1,4) {};
\foreach\x in {2,...,6}{
\foreach\y in {1,...,2}{
\light{\x}{\y};
}
\foreach\y in {3,...,6}{
\dark{\x}{\y};
}
}
\foreach\x in {1,...,2}{
\draw[smooth] (2,2.5+2*\x) arc (90:270:0.5);
\node at (1.5,2+2*\x) {\fs$\bullet$};
}
%
\node at (1.3,6.5) {$\ss \red x_1$};
\node at (1.3,5.5) {$\ss \red \b{x}_1$};
\node at (1.8,5.15) {$\ss \red \vdots$};
\node at (1.3,4.5) {$\ss \red x_n$};
\node at (1.3,3.5) {$\ss \red \b{x}_n$};
\node at (1.3,2.5) {$\ss \red \b{y}_n$};
\node at (1.8,2.15) {$\ss \red \vdots$};
\node at (1.3,1.5) {$\ss \red \b{y}_1$};
\node at (2.5,0.5) {$\ss 0$};
\node at (3.5,0.5) {$\ss 0$};
\node at (4.5,0.5) {$\ss 0$};
\node at (5.5,0.5) {$\ss 0$};
\node at (6.5,0.5) {$\ss \cdots$};
\node at (2.5,7.5) {$\ss 0$};
\node at (3.5,7.5) {$\ss 0$};
\node at (4.5,7.5) {$\ss 0$};
\node at (5.5,7.5) {$\ss 0$};
\node at (6.5,7.5) {$\ss \cdots$};
\foreach\x in {1,...,2}{
\gbull{2}{0.5+\x}{0.09}
}
\node at (7.5,4) {$\times$};
\foreach\x in {0,1,...,3}{
\draw[fline] (9.5+0.5*\x,2+0.5*\x) -- (8,3.5+\x);
\gbull{9.5+0.5*\x}{2+0.5*\x}{0.09}
}
\foreach\x in {1,2}{
\draw[fline] (8,0.5+\x) -- (11.5-0.5*\x,4+0.5*\x);
\ebull{11.5-0.5*\x}{4+0.5*\x}{0.09}
}
\foreach\x in {1,...,2}{
\gbull{2}{0.5+\x}{0.09}
\ebull{7}{0.5+\x}{0.09}
\ebull{8}{0.5+\x}{0.09}
}
\foreach\x in {1,...,4}{
\gbull{7}{2.5+\x}{0.09}
\gbull{8}{2.5+\x}{0.09}
}
\end{tikzpicture}
\end{align}
From this, it is immediate that
\begin{align}
\label{something}
\mathcal{P}(x^{\pm1}_1,\dots,x^{\pm1}_n;y_1,\dots,y_n;t)
=
\prod_{j=1}^{n} (x_j - t \b{x}_j) y_j
\sum_{\lambda}
\prod_{i=1}^{m_0(\lambda)} 
(1-t^{i})
K_{\lambda}(x^{\pm1}_1,\dots,x^{\pm1}_n;t) 
Q_{\lambda}(y_1,\dots,y_n;t).
\end{align}

Returning to the original partition function \eqref{cauch-refl-1}, we use the intertwining equation to transform it to
\begin{align}
\mathcal{P}(x^{\pm1}_1,\dots,x^{\pm1}_n;y_1,\dots,y_n;t)
=
\begin{tikzpicture}[scale=0.6,baseline=(mid.base)]
\node (mid) at (1,4) {};
\foreach\x in {1,...,4}{
\draw[fline] (-2.5+0.5*\x,4+0.5*\x) -- (1,0.5+\x);
}
\foreach\x in {2,3}{
\draw[fline] (1,3.5+\x) -- (-0.5-0.5*\x,2+0.5*\x);
\gbull{-0.5-0.5*\x}{2+0.5*\x}{0.09}
}
\foreach\x in {0,1}{
\draw[smooth] (-2+\x,4.5+\x) arc (225:45:0.353553);
\node at (-2+\x,5+\x) {\fs$\bullet$};
}
\foreach\x in {1,...,5}{
\foreach\y in {1,...,4}{
\dark{\x}{\y};
}
\foreach\y in {5,...,6}{
\light{\x}{\y};
}
}
\node at (1.5,0.5) {$\ss 0$};
\node at (2.5,0.5) {$\ss 0$};
\node at (3.5,0.5) {$\ss 0$};
\node at (4.5,0.5) {$\ss 0$};
\node at (5.5,0.5) {$\ss \cdots$};
\node at (1.5,7.5) {$\ss 0$};
\node at (2.5,7.5) {$\ss 0$};
\node at (3.5,7.5) {$\ss 0$};
\node at (4.5,7.5) {$\ss 0$};
\node at (5.5,7.5) {$\ss \cdots$};
\foreach\x in {1,2}{
\ebull{6}{4.5+\x}{0.09}
}
\foreach\x in {1,...,4}{
\gbull{6}{0.5+\x}{0.09}
}
\node at (-1,6.5) {$\ss \red x_1$};
\node at (-1.5,6) {$\ss \red \b{x}_1$};
\node at (-1.35,5.5) {$\ss \red \ddiagdot$};
\node at (-2,5.5) {$\ss \red x_n$};
\node at (-2.5,5) {$\ss \red \b{x}_n$};
\node at (-2.5,3) {$\ss \red \b{y}_n$};
\node at (-1.9,3) {$\ss \red \udiagdot$};
\node at (-2,2.5) {$\ss \red \b{y}_1$};
\end{tikzpicture}
\end{align}
Using particle conservation arguments, this partition function freezes from its zeroth column onward, as follows:
\begin{align}
\mathcal{P}(x^{\pm1}_1,\dots,x^{\pm1}_n;y_1,\dots,y_n;t)
=
\begin{tikzpicture}[scale=0.6,baseline=(mid.base)]
\node (mid) at (1,4) {};
\foreach\x in {1,...,4}{
\draw[fline] (-3.5+0.5*\x,4+0.5*\x) -- (0,0.5+\x);
}
\foreach\x in {2,3}{
\draw[fline] (0,3.5+\x) -- (-1.5-0.5*\x,2+0.5*\x);
\gbull{-1.5-0.5*\x}{2+0.5*\x}{0.09}
}
\foreach\x in {0,1}{
\draw[smooth] (-3+\x,4.5+\x) arc (225:45:0.353553);
\node at (-3+\x,5+\x) {\fs$\bullet$};
}
\node at (0.5,4) {$\times$};
\foreach\x in {1,...,5}{
\foreach\y in {1,...,4}{
\dark{\x}{\y};
}
\foreach\y in {5,...,6}{
\light{\x}{\y};
}
}
\node at (1.5,0.5) {$\ss 0$};
\node at (2.5,0.5) {$\ss 0$};
\node at (3.5,0.5) {$\ss 0$};
\node at (4.5,0.5) {$\ss 0$};
\node at (5.5,0.5) {$\ss \cdots$};
\node at (1.5,7.5) {$\ss 0$};
\node at (2.5,7.5) {$\ss 0$};
\node at (3.5,7.5) {$\ss 0$};
\node at (4.5,7.5) {$\ss 0$};
\node at (5.5,7.5) {$\ss \cdots$};
\foreach\x in {1,...,4}{
\draw[gline] (1,\x+0.5) -- (6,\x+0.5);
}
\foreach\x in {1,2}{
\ebull{0}{4.5+\x}{0.09}
\ebull{1}{4.5+\x}{0.09}
\ebull{6}{4.5+\x}{0.09}
}
\foreach\x in {1,...,4}{
\gbull{0}{0.5+\x}{0.09}
\gbull{1}{0.5+\x}{0.09}
\gbull{6}{0.5+\x}{0.09}
}
\node at (-2,6.5) {$\ss \red x_1$};
\node at (-2.5,6) {$\ss \red \b{x}_1$};
\node at (-2.35,5.5) {$\ss \red \ddiagdot$};
\node at (-3,5.5) {$\ss \red x_n$};
\node at (-3.5,5) {$\ss \red \b{x}_n$};
\node at (-3.5,3) {$\ss \red \b{y}_n$};
\node at (-2.9,3) {$\ss \red \udiagdot$};
\node at (-3,2.5) {$\ss \red \b{y}_1$};
\end{tikzpicture}
\end{align}
We evaluate the remaining, non-trivial part of this partition function in the next lemma. 

\subsection{Determinant evaluation of $Z_{\rm U}$}

\begin{lem}
\begin{multline}
\label{lem-uasm}
\begin{tikzpicture}[scale=0.8,baseline=1.5cm,rotate=45]
\foreach\x in {1,...,4}{
\draw[fline] (-3.5+0.5*\x,4+0.5*\x) -- (-2+0.5*\x,2.5+0.5*\x);
\gbull{-2+0.5*\x}{2.5+0.5*\x}{0.09};
}
\foreach\x in {2,3}{
\draw[fline] (1-0.5*\x,4.5+0.5*\x) -- (-1.5-0.5*\x,2+0.5*\x);
\gbull{-1.5-0.5*\x}{2+0.5*\x}{0.09};
\ebull{1-0.5*\x}{4.5+0.5*\x}{0.09};
}
\foreach\x in {0,1}{
\draw[smooth] (-3+\x,4.5+\x) arc (225:45:0.353553);
\node at (-3+\x,5+\x) {\fs$\bullet$};
}
\node at (-2,6.5) {$\ss \red x_1$};
\node at (-2.5,6) {$\ss \red \b{x}_1$};
\node at (-2.35,5.5) {$\ss \red \vdots$};
\node at (-3,5.5) {$\ss \red x_n$};
\node at (-3.5,5) {$\ss \red \b{x}_n$};
\node at (-3.5,3) {$\ss \red \b{y}_n$};
\node at (-2.8,3.15) {$\ss \red \cdots$};
\node at (-3,2.5) {$\ss \red \b{y}_1$};
\end{tikzpicture}
\quad
=
\quad
\frac{\displaystyle{
\prod_{i=1}^{n}
(x_i - t \b{x}_i)
y_i
\prod_{i,j=1}^{n} 
(1-t x_i y_j) (1-t \b{x}_i y_j) 
}}
{\displaystyle{
\prod_{1\leq i<j \leq n} (x_i-x_j) (y_i-y_j)(1 - \b{x}_i \b{x}_j)(1 - t y_i y_j)
}}
\\[-0.7cm]
\times
\det_{1\leq i,j \leq n}
\left[
\frac{(1-t)}{(1- x_i y_j)(1-\b{x}_i y_j)(1-t x_i y_j)(1-t \b{x}_i y_j)}
\right].
\end{multline}
\end{lem}

\begin{proof}
Let $Z_{\rm U}(n)$ be the left hand side of \eqref{lem-uasm} multiplied by 
$\prod_{i,j=1}^{n} (1-x_i y_j) (1-\b{x}_i y_j) \prod_{k=1}^{n} \b{y}_k (x_k - t\b{x}_k)^{-1}$. It satisfies the following properties:
\begin{enumerate}[label=\bf\arabic*.]
\item $Z_{\rm U}(n)$ is symmetric in $\{x_1,\dots,x_n\}$ and their inversions, and in 
$\{y_1,\dots,y_n\}$.

\item $Z_{\rm U}(n)$ is a Laurent polynomial in $x_n$ of top/bottom degree $n-1$.

\item Setting $x_n = \b{y}_n$, we have
\begin{align*}
Z_{\rm U}(n) \Big|_{x_n = \b{y}_n}
=
(1-t) \prod_{j=1}^{n-1} (1-t y_j \b{y}_n) (1-y_j y_n)
\prod_{i=1}^{n-1} (1-t x_i y_n) (1-t \b{x}_i y_n)
Z_{\rm U}(n-1),
\end{align*}
where $Z_{\rm U}(n-1)$ is the re-normalized partition function of one size smaller.

\item The smallest partition function $Z_{\rm U}(1)$ is given explicitly by
\begin{align*}
Z_{\rm U}(1) = 1-t.
\end{align*}
\end{enumerate}
As usual, these properties determine $Z_{\rm U}(n)$ uniquely, and it is a simple check to show that the right hand side of \eqref{lem-uasm} obeys them.

\end{proof}
Cancelling the common factor between the right hand sides of \eqref{something} and \eqref{lem-uasm}, we recover \eqref{ref-reflect-cauch}.
\end{proof}


\section{Doubly reflecting Cauchy identity}
\label{sec:doub-refl}

In this section we present a new Cauchy-type identity, equation \eqref{thm-uuasm}, which expresses an infinite sum of a product of two $BC_n$ Hall--Littlewood polynomials as a finite partition function in the six-vertex model. This differs from existing Cauchy identities for Koornwinder polynomials \cite{mim} (and their special cases), which are boxed (the summation is over finitely many terms). Of course infinite sums over $BC_n$ Hall--Littlewood polynomials are generally not well-defined, in view of their inhomogeneity, so we introduce an additional parameter $z$ which allows us to regulate the degree of each term in the sum.

We have not been able to evaluate the right hand side of \eqref{thm-uuasm} as a simple expression involving determinants, and some tests of small examples using symbolic algebra packages indicate that no such expression exists. This is somewhat disappointing, considering that for $z=t^{-1/2}$ it is equal to the partition function $Z_{\rm UU}$ in \cite{kup2}, and given by a product of two determinants. The difficulty stems from the need to solve the recursion relation \eqref{rec-rel} with $z$ generic, which is probably very difficult considering that \eqref{rec-rel} simplifies greatly when $z=t^{-1/2}$ (the factors on its right hand side coalesce to form perfect squares).

For simplicity, since they do not qualitatively change our result\footnote{Including the parameters $\gamma, \delta$ is achieved by replacing the boundaries in \eqref{thm-uuasm} with boundaries of the form \eqref{extended-boundary1}, \eqref{extended-boundary2}. We avoid such complications, since we are already unable to evaluate the partition function \eqref{thm-uuasm} in a simple form.}, in this section we take all Hall--Littlewood boundary parameters $\gamma = \delta = 0$.

\subsection{Lattice formulation of doubly reflecting Cauchy identity}

\begin{thm}
$BC_n$ Hall--Littlewood polynomials satisfy the following Cauchy identity:
\begin{align}
\nonumber
&\quad\quad\quad\quad
\prod_{j=1}^{n} (x_j-t\b{x}_j)(\b{y}_j-t y_j)
\times
\\[-0.8cm]
\label{thm-uuasm}
&\quad\quad\quad\quad
\sum_{\lambda}
z^{n+|\lambda|} 
\prod_{i=1}^{m_0(\lambda)}
(1-t^i)
b_{\lambda}(t)
K_{\lambda}(x_1^{\pm1},\dots,x_n^{\pm1};t) 
K_{\lambda}(y_1^{\pm1},\dots,y_n^{\pm1};t) 
=
\begin{tikzpicture}[scale=0.8,baseline=(mid.base),rotate=45]
\node (mid) at (1,1) {};
\foreach\x in {1,...,4}{
\draw[fline] (-4.5+0.5*\x,4+0.5*\x) -- (-2+0.5*\x,1.5+0.5*\x);
\gbull{-2+0.5*\x}{1.5+0.5*\x}{0.09}
}
\begin{scope}[xscale=1,yscale=-1,yshift=-8cm]
\foreach\x in {1,...,4}{
\draw[fline] (-4.5+0.5*\x,4+0.5*\x) -- (-2+0.5*\x,1.5+0.5*\x);
\ebull{-2+0.5*\x}{1.5+0.5*\x}{0.09}
}
\end{scope}
\foreach\x in {0,1}{
\draw[smooth] (-4+\x,4.5+\x) arc (225:45:0.353553);
\node at (-4+\x,5+\x) {\fs$\bullet$};
}
\begin{scope}[xscale=1,yscale=-1,yshift=-8cm]
\foreach\x in {0,1}{
\draw[smooth] (-4+\x,4.5+\x) arc (225:45:0.353553);
\node at (-4+\x,5+\x) {\fs$\bullet$};
}
\end{scope}
\node at (-3,6.5) {$\ss \red z x_1$};
\node at (-3.5,6) {$\ss \red z \b{x}_1$};
\node at (-3.4,5.6) {\tiny\red $\vdots$};
\node at (-4,5.5) {$\ss \red z x_n$};
\node at (-4.5,5) {$\ss \red z \b{x}_n$};
\node at (-4.5,3) {$\ss \red \b{y}_n$};
\node at (-4,2.5) {$\ss \red y_n$};
\node at (-3.5,2.5) {\tiny\red $\cdots$};
\node at (-3.5,2) {$\ss \red \b{y}_1$};
\node at (-3,1.5) {$\ss \red y_1$};
\end{tikzpicture}
\end{align}
where the sum is over all partitions $\lambda$ and $z$ is an arbitrary parameter.
\end{thm}

\begin{proof}
The proof is based on methods now familiar from previous sections, so we shorten our explanations. We consider the following partition function:
\begin{align}
\mathcal{P}(x_1^{\pm1},\dots,x_n^{\pm1};y_1^{\pm1},\dots,y_n^{\pm1};t;z)
=
\begin{tikzpicture}[scale=0.65,baseline=(mid.base)]
\node (mid) at (1,4) {};
\foreach\x in {1,...,5}{
\foreach\y in {0,...,3}{
\light{\x}{\y};
}
\foreach\y in {4,...,7}{
\dark{\x}{\y};
}
}
\foreach\x in {-1,0,1,2}{
\draw[smooth] (1,3.5+2*\x) arc (90:270:0.5);
\node at (0.5,3+2*\x) {\fs$\bullet$};
}
\node at (1.5,-0.5) {$\ss 0$};
\node at (2.5,-0.5) {$\ss 0$};
\node at (3.5,-0.5) {$\ss 0$};
\node at (4.5,-0.5) {$\ss 0$};
\node at (5.5,-0.5) {$\ss \cdots$};
\node at (1.5,8.5) {$\ss 0$};
\node at (2.5,8.5) {$\ss 0$};
\node at (3.5,8.5) {$\ss 0$};
\node at (4.5,8.5) {$\ss 0$};
\node at (5.5,8.5) {$\ss \cdots$};
\foreach\x in {1,...,4}{
\draw[fline] (8.5+0.5*\x,1+0.5*\x) -- (6,3.5+\x);
\draw[fline] (6,-0.5+\x) -- (10.5-0.5*\x,4+0.5*\x);
\gbull{8.5+0.5*\x}{0.5*\x+1}{0.09};
\ebull{10.5-0.5*\x}{4+0.5*\x}{0.09};
}
\node at (0.2,7.5) {$\ss \red z x_1$};
\node at (0.2,6.5) {$\ss \red z \b{x}_1$};
\node at (0.8,6.15) {$\ss \red \vdots$};
\node at (0.2,5.5) {$\ss \red z x_n$};
\node at (0.2,4.5) {$\ss \red z \b{x}_n$};
\node at (0.3,3.5) {$\ss \red \b{y}_n$};
\node at (0.3,2.5) {$\ss \red y_n$};
\node at (0.8,2.15) {$\ss \red \vdots$};
\node at (0.3,1.5) {$\ss \red \b{y}_1$};
\node at (0.3,0.5) {$\ss \red y_1$};
\end{tikzpicture}
\end{align}
Treating $z$ as complex parameter with small modulus, $|z| < 1$, the top (darkly shaded) half of the lattice will vanish unless all of the right edges are occupied by particles. This causes the attached block of $R$ matrices to freeze, and we obtain

\begin{align}
\mathcal{P}(x_1^{\pm1},\dots,x_n^{\pm1};y_1^{\pm1},\dots,y_n^{\pm1};t;z)
=
\begin{tikzpicture}[scale=0.65,baseline=(mid.base)]
\node (mid) at (1,4) {};
\foreach\x in {1,...,5}{
\foreach\y in {0,...,3}{
\light{\x}{\y};
}
\foreach\y in {4,...,7}{
\dark{\x}{\y};
}
}
\foreach\x in {-1,0,1,2}{
\draw[smooth] (1,3.5+2*\x) arc (90:270:0.5);
\node at (0.5,3+2*\x) {\fs$\bullet$};
}
\node at (1.5,-0.5) {$\ss 0$};
\node at (2.5,-0.5) {$\ss 0$};
\node at (3.5,-0.5) {$\ss 0$};
\node at (4.5,-0.5) {$\ss 0$};
\node at (5.5,-0.5) {$\ss \cdots$};
\node at (1.5,8.5) {$\ss 0$};
\node at (2.5,8.5) {$\ss 0$};
\node at (3.5,8.5) {$\ss 0$};
\node at (4.5,8.5) {$\ss 0$};
\node at (5.5,8.5) {$\ss \cdots$};
\node at (6.5,4) {$\times$};
\foreach\x in {1,...,4}{
\draw[fline] (9.5+0.5*\x,1+0.5*\x) -- (7,3.5+\x);
\draw[fline] (7,-0.5+\x) -- (11.5-0.5*\x,4+0.5*\x);
\gbull{9.5+0.5*\x}{0.5*\x+1}{0.09};
\ebull{11.5-0.5*\x}{4+0.5*\x}{0.09};
}
\foreach\x in {1,...,4}{
\gbull{6}{3.5+\x}{0.09}
\gbull{7}{3.5+\x}{0.09}
\ebull{6}{-0.5+\x}{0.09}
\ebull{7}{-0.5+\x}{0.09}
}
\node at (0.2,7.5) {$\ss \red z x_1$};
\node at (0.2,6.5) {$\ss \red z \b{x}_1$};
\node at (0.8,6.15) {$\ss \red \vdots$};
\node at (0.2,5.5) {$\ss \red z x_n$};
\node at (0.2,4.5) {$\ss \red z \b{x}_n$};
\node at (0.3,3.5) {$\ss \red \b{y}_n$};
\node at (0.3,2.5) {$\ss \red y_n$};
\node at (0.8,2.15) {$\ss \red \vdots$};
\node at (0.3,1.5) {$\ss \red \b{y}_1$};
\node at (0.3,0.5) {$\ss \red y_1$};
\end{tikzpicture}
\end{align}
Employing the results from section \ref{sec:bc-lattice-paths}, we then find that
\begin{align*}
\mathcal{P}(x_1^{\pm1},\dots,x_n^{\pm1};y_1^{\pm1},\dots,y_n^{\pm1};t;z)
=
\prod_{k=1}^{n}
(x_k - t \b{x}_k)(\b{y}_k - t y_k)
\sum_{\lambda}
z^{n+|\lambda|} 
K_{\lambda}(x^{\pm1}_1,\dots,x^{\pm1}_n;t) 
L_{\lambda}(y^{\pm1}_1,\dots,y^{\pm1}_n;t).
\end{align*}

It is now very easy to obtain the right hand side of \eqref{thm-uuasm}. By repeated application of the intertwining equation \eqref{rll*} we obtain
\begin{multline}
\mathcal{P}(x_1^{\pm1},\dots,x_n^{\pm1};y_1^{\pm1},\dots,y_n^{\pm1};t;z)
=
\\
\begin{tikzpicture}[scale=0.65,baseline=(mid.base)]
\node (mid) at (1,4) {};
\foreach\x in {1,...,4}{
\draw[fline] (-3.5+0.5*\x,4+0.5*\x) -- (1,-0.5+\x);
}
\begin{scope}[xscale=1,yscale=-1,yshift=-8cm]
\foreach\x in {1,...,4}{
\draw[fline] (-3.5+0.5*\x,4+0.5*\x) -- (1,-0.5+\x);
}
\end{scope}
\foreach\x in {0,1}{
\draw[smooth] (-3+\x,4.5+\x) arc (225:45:0.353553);
\node at (-3+\x,5+\x) {\fs$\bullet$};
}
\begin{scope}[xscale=1,yscale=-1,yshift=-8cm]
\foreach\x in {0,1}{
\draw[smooth] (-3+\x,4.5+\x) arc (225:45:0.353553);
\node at (-3+\x,5+\x) {\fs$\bullet$};
}
\end{scope}
\foreach\x in {1,...,5}{
\foreach\y in {0,...,3}{
\dark{\x}{\y};
}
\foreach\y in {4,...,7}{
\light{\x}{\y};
}
}
\node at (1.5,-0.5) {$\ss 0$};
\node at (2.5,-0.5) {$\ss 0$};
\node at (3.5,-0.5) {$\ss 0$};
\node at (4.5,-0.5) {$\ss 0$};
\node at (5.5,-0.5) {$\ss \cdots$};
\node at (1.5,8.5) {$\ss 0$};
\node at (2.5,8.5) {$\ss 0$};
\node at (3.5,8.5) {$\ss 0$};
\node at (4.5,8.5) {$\ss 0$};
\node at (5.5,8.5) {$\ss \cdots$};
\foreach\x in {1,...,4}{
\ebull{6}{3.5+\x}{0.09}
\gbull{6}{-0.5+\x}{0.09}
}
\node at (-2,6.5) {$\ss \red z x_1$};
\node at (-2.5,6) {$\ss \red z \b{x}_1$};
\node at (-2.35,5.5) {$\ss \red \ddiagdot$};
\node at (-3,5.5) {$\ss \red z x_n$};
\node at (-3.5,5) {$\ss \red z \b{x}_n$};
\node at (-3.5,3) {$\ss \red \b{y}_n$};
\node at (-3,2.5) {$\ss \red y_n$};
\node at (-2.4,2.6) {$\ss \red \udiagdot$};
\node at (-2.5,2) {$\ss \red \b{y}_1$};
\node at (-2,1.5) {$\ss \red y_1$};
\end{tikzpicture}
=
\begin{tikzpicture}[scale=0.65,baseline=(mid.base)]
\node (mid) at (1,4) {};
\foreach\x in {1,...,4}{
\draw[fline] (-4.5+0.5*\x,4+0.5*\x) -- (0,-0.5+\x);
}
\begin{scope}[xscale=1,yscale=-1,yshift=-8cm]
\foreach\x in {1,...,4}{
\draw[fline] (-4.5+0.5*\x,4+0.5*\x) -- (0,-0.5+\x);
}
\end{scope}
\foreach\x in {0,1}{
\draw[smooth] (-4+\x,4.5+\x) arc (225:45:0.353553);
\node at (-4+\x,5+\x) {\fs$\bullet$};
}
\begin{scope}[xscale=1,yscale=-1,yshift=-8cm]
\foreach\x in {0,1}{
\draw[smooth] (-4+\x,4.5+\x) arc (225:45:0.353553);
\node at (-4+\x,5+\x) {\fs$\bullet$};
}
\end{scope}
\node at (0.5,4) {$\times$};
\foreach\x in {1,...,5}{
\foreach\y in {0,...,3}{
\dark{\x}{\y};
}
\foreach\y in {4,...,7}{
\light{\x}{\y};
}
}
\node at (1.5,-0.5) {$\ss 0$};
\node at (2.5,-0.5) {$\ss 0$};
\node at (3.5,-0.5) {$\ss 0$};
\node at (4.5,-0.5) {$\ss 0$};
\node at (5.5,-0.5) {$\ss \cdots$};
\node at (1.5,8.5) {$\ss 0$};
\node at (2.5,8.5) {$\ss 0$};
\node at (3.5,8.5) {$\ss 0$};
\node at (4.5,8.5) {$\ss 0$};
\node at (5.5,8.5) {$\ss \cdots$};
\foreach\x in {0,...,3}{
\draw[gline] (1,\x+0.5) -- (6,\x+0.5);
}
\foreach\x in {1,...,4}{
\ebull{6}{3.5+\x}{0.09};
\ebull{1}{3.5+\x}{0.09};
\gbull{6}{-0.5+\x}{0.09};
\gbull{1}{-0.5+\x}{0.09};
}
\foreach\x in {1,...,4}{
\ebull{0}{3.5+\x}{0.09}
\gbull{0}{-0.5+\x}{0.09}
}
\node at (-3,6.5) {$\ss \red z x_1$};
\node at (-3.5,6) {$\ss \red z \b{x}_1$};
\node at (-3.35,5.5) {$\ss \red \ddiagdot$};
\node at (-4,5.5) {$\ss \red z x_n$};
\node at (-4.5,5) {$\ss \red z \b{x}_n$};
\node at (-4.5,3) {$\ss \red \b{y}_n$};
\node at (-4,2.5) {$\ss \red y_n$};
\node at (-3.4,2.6) {$\ss \red \udiagdot$};
\node at (-3.5,2) {$\ss \red \b{y}_1$};
\node at (-3,1.5) {$\ss \red y_1$};
\end{tikzpicture}
\end{multline}
with the second equality coming from particle-conservation. Noting that frozen tiles all have Boltzmann weight 1, we obtain precisely the partition function on the right hand side of \eqref{thm-uuasm}.

\end{proof}

\subsection{Analysis of $Z_{\rm UU}$}

Equation \eqref{thm-uuasm} is distinct from the earlier identities in this paper, in that we do not know how to evaluate the partition function on the right hand side in a simple closed form. This is due to the extra parameter $z$, which we were forced to include in the left hand side of \eqref{thm-uuasm} to make it a well-defined expression, but which prevents the easy evaluation of the right hand side. 

Nevertheless, this partition function does satisfy a list of properties which determine it uniquely. Let 
$Z_{\rm UU}(n)$ denote the right hand side of \eqref{thm-uuasm}, multiplied by
$$z^{-n}\times \prod_{i,j=1}^{n} (1- z x_i y_j) (1- z \b{x}_i y_j) (1 - z x_i \b{y}_j) (1 - z \b{x}_i \b{y}_j)
\prod_{k=1}^{n} (x_k -t \b{x}_k)^{-1} (\b{y}_k - t y_k)^{-1}.$$ The re-normalized partition function satisfies the following properties:
\begin{enumerate}[label=\bf\arabic*.]
\item $Z_{\rm UU}(n)$ is symmetric in $\{x_1,\dots,x_n\}$ and $\{y_1,\dots,y_n\}$ separately, and in their inversions.
\item $Z_{\rm UU}(n)$ is a Laurent polynomial in $x_n$ with top/bottom degree of at most $2n-1$.
\item Setting $x_n = \b{z} \b{y}_n$, we have the recursion relation
\begin{multline}
\label{rec-rel}
Z_{\rm UU}(n)
\Big|_{x_n = \b{z} \b{y}_n}
=
(1-t)(1-z^2)
\prod_{i=1}^{n-1}
(1-t z x_i y_n)(1-t z \b{x}_i y_n)(1-z x_i \b{y}_n)(1-z \b{x}_i \b{y}_n)
\\
\times
(1-t y_i \b{y}_n)(1-t \b{y}_i \b{y}_n)(1-z^2 y_i y_n)(1-z^2 \b{y}_i y_n)
Z_{\rm UU}(n-1).
\end{multline}

\item The smallest partition function $Z_{\rm UU}(1)$ is given explicitly by
\begin{align*}
Z_{\rm UU}(1)
=
(1-t)(1-z^2).
\end{align*}

\end{enumerate}

\begin{rmk}{\rm 
Although we have not solved the above conditions for general $z$, there are some special values of $z$ for which we are able to obtain an explicit solution. First, and trivially, when $z = \pm1$ we find that 
$Z_{\rm UU}(n)$ vanishes for all $n \geq 1$. This is clear from the fact that $Z_{\rm UU}(1) =0$ at these values of $z$, or by direct application of Lemma \ref{lem:vanishing}.

When $z=t^{-1/2}$, we are able to solve the above conditions explicitly:
\begin{align*}
Z_{\rm UU}(n)
=
&
\frac{\prod_{i,j=1}^{n} (1-t^{\pm \frac{1}{2}} x_i y_j) (1-t^{\pm \frac{1}{2}} x_i \b{y}_j)
(1-t^{\pm \frac{1}{2}} \b{x}_i y_j) (1-t^{\pm \frac{1}{2}} \b{x}_i \b{y}_j)}
{\prod_{1 \leq i<j \leq n} (x_i-x_j)^2 (y_i-y_j)^2 (1-\b{x}_i \b{x}_j)^2 (1-\b{y}_i \b{y}_j)^2}
\\
&\times
\det_{1\leq i,j \leq n}
\left[
\frac{1-t}{(1-t^{\frac{1}{2}} x_i y_j)(1-t^{\frac{1}{2}} x_i \b{y}_j)
(1-t^{\frac{1}{2}} \b{x}_i y_j)(1-t^{\frac{1}{2}} \b{x}_i \b{y}_j)}
\right]
\\
&\times
\det_{1\leq i,j \leq n}
\left[
\frac{1-t^{-1}}{(1-t^{-\frac{1}{2}} x_i y_j)(1-t^{-\frac{1}{2}} x_i \b{y}_j)
(1-t^{-\frac{1}{2}} \b{x}_i y_j)(1-t^{-\frac{1}{2}} \b{x}_i \b{y}_j)}
\right].
\end{align*}
This evaluation of the partition function was obtained by Kuperberg in \cite{kup2}, although here we use slightly different Boltzmann weights and less general boundaries compared with those of \cite{kup2}.

Unfortunately, all of the values $z=\pm1$ and $z=t^{-1/2}$ lead to the divergence of the left hand side of \eqref{thm-uuasm}, meaning that this identity is no longer valid in these special cases.
}
\end{rmk}

\begin{rmk}{\rm
Another special case of \eqref{thm-uuasm} is of interest, namely $t=0$. In this case the $BC_n$ Hall--Littlewood polynomials degenerate to symplectic characters, and we obtain
\begin{multline*}
\sum_{\lambda}
z^{|\lambda|} 
sp_{\lambda}(x^{\pm1}_1,\dots,x^{\pm1}_n) 
sp_{\lambda}(y^{\pm1}_1,\dots,y^{\pm1}_n) 
= 
\frac{z^{-n(n-1)/2}}{\prod_{1 \leq i<j \leq n} (x_i-x_j) (y_i-y_j) (1-\b{x}_i \b{x}_j) (1-\b{y}_i \b{y}_j)}
\\
\times
\det_{1\leq i,j \leq n}
\left[
\frac{(1-z)(1+z)}{(1-z x_i y_j)(1-z x_i \b{y}_j)(1-z \b{x}_i y_j)(1-z \b{x}_i \b{y}_j)}
\right].
\end{multline*}

}
\end{rmk}

\section{Reflecting Littlewood identity}
\label{sec:refl-little}

In this section we present a second new identity, equation \eqref{thm-uosasm}, which is of Littlewood-type. It relates an infinite sum over $BC_{2n}$ Hall--Littlewood polynomials with a partition function of the six-vertex model on a finite domain. As in the previous section, we introduce an extra parameter $z$ into the sum on the left hand side of \eqref{thm-uosasm}, which is necessary to avoid divergence issues. 

The partition function on the right hand side of \eqref{thm-uosasm} is closely related to $Z_{\rm UO}$ in \cite{kup2}, which is a multi-parameter generating series of off-diagonally symmetric ASMs with a U-turn boundary. Indeed, for 
$z = t^{-1/2}$ we obtain the same Boltzmann weights as in \cite{kup2}, and this partition function can be evaluated as the product of two Pfaffians. For generic values of $z$, however, we have not obtained a simple closed form expression. Our results in this section are thus a direct parallel of those obtained in section \ref{sec:doub-refl}.

Throughout this section, we take the Hall--Littlewood boundary parameters to be $\gamma = \delta = 0$.

\subsection{Lattice formulation of reflecting Littlewood identity}

\begin{thm}
$BC_{2n}$ Hall--Littlewood polynomials satisfy the following Littlewood identity:
\begin{align}
\nonumber
&\quad\quad\quad\quad
\prod_{j=1}^{2n} (x_j-t\b{x}_j)
\times
\\[-0.9cm]
\label{thm-uosasm}
&\quad\quad\quad\quad
\sum_{
\lambda: m_i(\lambda) \in 2 \mathbb{N}
}
\
z^{n+|\lambda|/2} 
\prod_{i=0}^{\infty}
\prod_{j=1}^{m_i(\lambda)/2}
(1-t^{2j-1})
K_{\lambda}(x^{\pm1}_1,\dots,x^{\pm1}_{2n};t)
= 
\begin{tikzpicture}[scale=0.7,baseline=(mid.base),rotate=45]
\node (mid) at (1,-2.5) {};
\foreach\x in {1,...,8}{
\draw[fline,cross=1] (1.5-0.5*\x,1+0.5*\x) -- (1.5-\x,1);
\ebull{1.5-0.5*\x}{1+0.5*\x}{0.09};
\draw[fline] (-3-0.5*\x,5.5-0.5*\x) -- (-2.5-0.5*\x,5-0.5*\x);
}
\foreach\x in {1,...,7}{
\draw[fline] (-2.5-0.5*\x,5-0.5*\x) -- (1.5-\x,1);
}
\foreach\x in {0,...,3}{
\draw[smooth] (-7+\x,1.5+\x) arc (225:45:0.353553);
\node at (-7+\x,2+\x) {\fs$\bullet$};
}
\node at (-4.2,5.8) {$\ss \red \sqrt{z} x_1$};
\node at (-4.7,5.3) {$\ss \red \sqrt{z} \b{x}_1$};
\node at (-6,4.2) {$\ss \red \vdots $};
\node at (-7.2,2.8) {$\ss \red \sqrt{z} x_{2n}$};
\node at (-7.7,2.3) {$\ss \red \sqrt{z} \b{x}_{2n}$};
\node at (-6.6,0.5) {$\ss \red \sqrt{\b{z}} x_{2n}$};
\node at (-5.6,0.5) {$\ss \red \sqrt{\b{z}} \b{x}_{2n}$};
\node at (-3.2,0.5) {$\ss \red \ddiagdot$};
\node at (-0.6,0.5) {$\ss \red \sqrt{\b{z}} x_1$};
\node at (0.4,0.5) {$\ss \red \sqrt{\b{z}} \b{x}_1$};
\end{tikzpicture}
\end{align}
where the sum is taken over partitions $\lambda$ with even multiplicities $m_i(\lambda)$ for all $i \geq 0$, and $z$ is a free parameter.
\end{thm}

\begin{proof}
The proof is closely related to the proof of Theorem \ref{thm:refined-little}. We start with the following partition function:
\begin{align}
\mathcal{P}(x_1^{\pm1},\dots,x_{2n}^{\pm1};t;z)
=
\begin{tikzpicture}[scale=0.6,baseline=(mid.base)]
\node (mid) at (1,4) {};
\foreach\x in {1,...,5}{
\foreach\y in {0,...,7}{
\dark{\x}{\y};
}
}
\foreach\x in {-1,0,1,2}{
\draw[smooth] (1,3.5+2*\x) arc (90:270:0.5);
\node at (0.5,3+2*\x) {\fs$\bullet$};
}
\node at (3.5,-0.5) {$\ss \bra{\rm e;0}$};
\node at (1.5,8.5) {$\ss 0$};
\node at (2.5,8.5) {$\ss 0$};
\node at (3.5,8.5) {$\ss 0$};
\node at (4.5,8.5) {$\ss 0$};
\node at (5.5,8.5) {$\ss \cdots$};
\foreach\x in {1,...,8}{
\gbull{6}{-0.5+\x}{0.09};
}
\node at (0,7.5) {$\ss \red  \sqrt{z} x_1$};
\node at (0,6.5) {$\ss \red  \sqrt{z} \b{x}_1$};
\node at (0,4.2) {$\ss \red \vdots$};
\node at (-0.2,1.5) {$\ss \red  \sqrt{z} x_{2n}$};
\node at (-0.2,0.5) {$\ss \red  \sqrt{z} \b{x}_{2n}$};
\end{tikzpicture}
\end{align}
where the state at the base of the lattice is given by \eqref{e-vec} with $\alpha = 0$ (since we do not have a refining parameter $u$). By inspection, we see that
\begin{align}
\mathcal{P}(x_1^{\pm1},\dots,x_{2n}^{\pm1};t;z)
=
\prod_{j=1}^{2n} (x_j-t\b{x}_j)
\times
\sum_{
\lambda: m_i(\lambda) \in 2 \mathbb{N}
}
\
z^{n+|\lambda|/2} 
\prod_{i=0}^{\infty}
\prod_{j=1}^{m_i(\lambda)/2}
(1-t^{2j-1})
K_{\lambda}(x^{\pm1}_1,\dots,x^{\pm1}_{2n};t).
\end{align}

To produce the right hand side of \eqref{thm-uosasm}, after a single application of \eqref{lem-even5}, \eqref{lem-even6} we find that
\begin{align}
\mathcal{P}(x_1^{\pm1},\dots,x_{2n}^{\pm1};t;z)
=
\begin{tikzpicture}[scale=0.6,baseline=(mid.base)]
\node (mid) at (1,4) {};
\foreach\x in {1,...,5}{
\foreach\y in {1,...,7}{
\dark{\x}{\y};
}
}
\foreach\x in {1,...,5}{
\light{\x}{0};
}
\foreach\x in {-1,0,1,2}{
\draw[smooth] (-1,3.5+2*\x) arc (90:270:0.5);
\node at (-1.5,3+2*\x) {\fs$\bullet$};
}
\foreach\x in {1,...,7}{
\draw[fline] (-1,\x+0.5) -- (1,\x+0.5);
}
\draw[fline,cross=0.5] (-1,0.5) -- (1,0.5);
\node at (3.5,-0.5) {$\ss \bra{\rm e;0}$};
\node at (1.5,8.5) {$\ss 0$};
\node at (2.5,8.5) {$\ss 0$};
\node at (3.5,8.5) {$\ss 0$};
\node at (4.5,8.5) {$\ss 0$};
\node at (5.5,8.5) {$\ss \cdots$};
\node at (0.3,7.8) {$\ss \red  \sqrt{z} x_1$};
\node at (0.3,6.8) {$\ss \red  \sqrt{z} \b{x}_1$};
\node at (0,4.2) {$\ss \red \vdots$};
\node at (0.3,1.8) {$\ss \red  \sqrt{z} x_{2n}$};
\node at (0.3,0.8) {$\ss \red \sqrt{\b{z}} x_{2n}$};
\foreach\x in {2,...,8}{
\gbull{6}{-0.5+\x}{0.09};
}
\ebull{6}{0.5}{0.09};
\end{tikzpicture}
=
\begin{tikzpicture}[scale=0.6,baseline=(mid.base)]
\node (mid) at (1,4) {};
\foreach\x in {1,...,5}{
\foreach\y in {1,...,7}{
\dark{\x}{\y};
}
}
\foreach\x in {1,...,5}{
\light{\x}{0};
}
\foreach\x in {-1,0,1,2}{
\draw[smooth] (-1,3.5+2*\x) arc (90:270:0.5);
\node at (-1.5,3+2*\x) {\fs$\bullet$};
}
\foreach\x in {1,...,7}{
\draw[fline] (-1,\x+0.5) -- (1,\x+0.5);
}
\draw[fline,cross=0.5] (-1,0.5) -- (1,0.5);
\node at (3.5,-0.5) {$\ss \bra{\rm e;0}$};
\node at (1.5,8.5) {$\ss 0$};
\node at (2.5,8.5) {$\ss 0$};
\node at (3.5,8.5) {$\ss 0$};
\node at (4.5,8.5) {$\ss 0$};
\node at (5.5,8.5) {$\ss \cdots$};
\draw[fline] (7,7.5) -- (7,8.5);
\draw[fline] (6,0.5) -- (7,0.5) -- (7,7.5);
\ebull{7}{8.5}{0.09}
\foreach\x in {1,...,7}{
\draw[fline] (8,\x+0.5) -- (6,\x+0.5);
\gbull{8}{\x+0.5}{0.09}
}
\node at (0.3,7.8) {$\ss \red  \sqrt{z} x_1$};
\node at (0.3,6.8) {$\ss \red  \sqrt{z} \b{x}_1$};
\node at (0,4.2) {$\ss \red \vdots$};
\node at (0.3,1.8) {$\ss \red  \sqrt{z} x_{2n}$};
\node at (0.3,0.8) {$\ss \red \sqrt{\b{z}} x_{2n}$};
\end{tikzpicture}
\end{align}
where the final equality follows from the fact that we must have a particle at the right edge of each darkly shaded row, otherwise everything vanishes (assuming that $|z| < 1$). Hence the column of $R$ vertices that we have introduced at the right edge is in fact completely frozen. By repeated application of the intertwining equation \eqref{rll*}, we can transfer the single lightly shaded row to the top of the lattice:
\begin{align}
\mathcal{P}(x_1^{\pm1},\dots,x_{2n}^{\pm1};t;z)
=
\begin{tikzpicture}[scale=0.6,baseline=(mid.base)]
\node (mid) at (1,4) {};
\foreach\x in {1,...,7}{
\draw[fline] (-1,\x+0.5) -- (1,\x+0.5);
}
\draw[fline,cross=1] (-1,0.5) -- (0,0.5); 
\draw[fline] (0,0.5) -- (0,8.5) -- (1,8.5);
\foreach\x in {-1,0,1,2}{
\draw[smooth] (-1,3.5+2*\x) arc (90:270:0.5);
\node at (-1.5,3+2*\x) {\fs$\bullet$};
}
\foreach\x in {1,...,5}{
\foreach\y in {1,...,7}{
\dark{\x}{\y};
}
}
\foreach\x in {1,...,5}{
\light{\x}{8};
}
\node at (1.5,9.5) {$\ss 0$};
\node at (2.5,9.5) {$\ss 0$};
\node at (3.5,9.5) {$\ss 0$};
\node at (4.5,9.5) {$\ss 0$};
\node at (5.5,9.5) {$\ss \cdots$};
\node at (3.5,0.5) {$\ss \bra{\rm e;0}$};
\ebull{6}{8.5}{0.09}
\foreach\x in {1,...,7}{
\gbull{6}{0.5+\x}{0.09}
}
\node at (-0.6,7.8) {$\ss \red  \sqrt{z} x_1$};
\node at (-0.6,6.8) {$\ss \red  \sqrt{z} \b{x}_1$};
\node at (-1,4.2) {$\ss \red \vdots$};
\node at (-0.7,1.8) {$\ss \red  \sqrt{z} x_{2n}$};
\node at (0.3,0) {$\ss \red  \sqrt{\b{z}} x_{2n}$};
\end{tikzpicture}
\end{align}
It is now only a matter of repeating these steps over the remaining darkly shaded rows. The final result is
\begin{align}
\begin{tikzpicture}[scale=0.6,baseline=(mid.base)]
\node (mid) at (1,3) {};
\foreach\x in {1,...,8}{
\draw[fline,cross=1] (1,0.5+\x) -- (1.5-\x,1);
\draw[fline] (-3-0.5*\x,5.5-0.5*\x) -- (-2.5-0.5*\x,5-0.5*\x);
}
\foreach\x in {1,...,7}{
\draw[fline] (-2.5-0.5*\x,5-0.5*\x) -- (1.5-\x,1);
}
\foreach\x in {0,...,3}{
\draw[smooth] (-7+\x,1.5+\x) arc (225:45:0.353553);
\node at (-7+\x,2+\x) {\fs$\bullet$};
}
\foreach\x in {1,...,5}{
\foreach\y in {1,...,8}{
\light{\x}{\y};
}
}
\node at (1.5,9.5) {$\ss 0$};
\node at (2.5,9.5) {$\ss 0$};
\node at (3.5,9.5) {$\ss 0$};
\node at (4.5,9.5) {$\ss 0$};
\node at (5.5,9.5) {$\ss \cdots$};
\node at (3.5,0.5) {$\ss \bra{\rm e;0}$};
\foreach\x in {1,...,8}{
\ebull{6}{0.5+\x}{0.09};
}
\node at (-4.2,5.8) {$\ss \red \sqrt{z} x_1$};
\node at (-5,5.3) {$\ss \red \sqrt{z} \b{x}_1$};
\node at (-6,4.2) {$\ss \red \ddiagdot$};
\node at (-7.4,2.8) {$\ss \red \sqrt{z} x_{2n}$};
\node at (-8.2,2.3) {$\ss \red \sqrt{z} \b{x}_{2n}$};
\node at (-7,0.5) {$\ss \red \sqrt{\b{z}} x_{2n}$};
\node at (-5.6,0.5) {$\ss \red \sqrt{\b{z}} \b{x}_{2n}$};
\node at (-3.2,0.5) {$\ss \red \cdots$};
\node at (-0.8,0.5) {$\ss \red \sqrt{\b{z}} x_1$};
\node at (0.4,0.5) {$\ss \red \sqrt{\b{z}} \b{x}_1$};
\end{tikzpicture}
\end{align}
and since no particles are incident either at the top or the right edge of the lattice, all tiles in the bulk of the lattice freeze to the same configuration with Boltzmann weight 1: 

\begin{align}
\begin{tikzpicture}[scale=0.6,baseline=(mid.base)]
\node (mid) at (1,3) {};
\foreach\x in {1,...,8}{
\draw[fline,cross=1] (1,0.5+\x) -- (1.5-\x,1);
\draw[fline] (-3-0.5*\x,5.5-0.5*\x) -- (-2.5-0.5*\x,5-0.5*\x);
}
\foreach\x in {1,...,7}{
\draw[fline] (-2.5-0.5*\x,5-0.5*\x) -- (1.5-\x,1);
}
\foreach\x in {0,...,3}{
\draw[smooth] (-7+\x,1.5+\x) arc (225:45:0.353553);
\node at (-7+\x,2+\x) {\fs$\bullet$};
}
\node at (1.5,5) {$\times$};
\foreach\x in {2,...,6}{
\foreach\y in {1,...,8}{
\light{\x}{\y};
}
}
\node at (2.5,9.5) {$\ss 0$};
\node at (3.5,9.5) {$\ss 0$};
\node at (4.5,9.5) {$\ss 0$};
\node at (5.5,9.5) {$\ss 0$};
\node at (6.5,9.5) {$\ss \cdots$};
\node at (2.5,0.5) {$\ss 0$};
\node at (3.5,0.5) {$\ss 0$};
\node at (4.5,0.5) {$\ss 0$};
\node at (5.5,0.5) {$\ss 0$};
\node at (6.5,0.5) {$\ss \cdots$};
\foreach\x in {1,...,8}{
\ebull{1}{0.5+\x}{0.09};
\ebull{2}{0.5+\x}{0.09};
\ebull{7}{0.5+\x}{0.09};
}
\node at (-4.2,5.8) {$\ss \red \sqrt{z} x_1$};
\node at (-5,5.3) {$\ss \red \sqrt{z} \b{x}_1$};
\node at (-6,4.2) {$\ss \red \ddiagdot$};
\node at (-7.4,2.8) {$\ss \red \sqrt{z} x_{2n}$};
\node at (-8.2,2.3) {$\ss \red \sqrt{z} \b{x}_{2n}$};
\node at (-7,0.5) {$\ss \red \sqrt{\b{z}} x_{2n}$};
\node at (-5.6,0.5) {$\ss \red \sqrt{\b{z}} \b{x}_{2n}$};
\node at (-3.2,0.5) {$\ss \red \cdots$};
\node at (-0.8,0.5) {$\ss \red \sqrt{\b{z}} x_1$};
\node at (0.4,0.5) {$\ss \red \sqrt{\b{z}} \b{x}_1$};
\end{tikzpicture}
\end{align}
The residual partition function is precisely the right hand side of \eqref{thm-uosasm}.

\end{proof}

\subsection{Analysis of $Z_{\rm UO}$}

As was the case with equation \eqref{thm-uuasm} in the previous section, we have not obtained an explicit evaluation of the partition function on the right hand side of \eqref{thm-uosasm}, for arbitrary $z$. In spite of this, the partition function appearing in \eqref{thm-uosasm} does satisfy a set of properties which determine it uniquely. Let $Z_{\rm UO}(2n)$ denote the right hand side of \eqref{thm-uosasm}, multiplied by 
$$z^{-n}\times \prod_{1 \leq i<j \leq 2n}(1-z x_i x_j)(1-z x_i \b{x}_j)(1-z \b{x}_i x_j)(1-z \b{x}_i \b{x}_j) \prod_{k=1}^{2n}(x_k - t \b{x}_k)^{-1}.$$ It satisfies the following properties:
\begin{enumerate}[label=\bf\arabic*.]
\item $Z_{\rm UO}(2n)$ is symmetric in $\{x_1,\dots,x_{2n}\}$ and in their inversions.
\item $Z_{\rm UO}(2n)$ is a Laurent polynomial in $x_{2n}$ with top/bottom degree of at most $4n-3$.
\item Setting $x_{2n} = \b{z} \b{x}_{2n-1}$, we have the recursion relation
\begin{multline*}
Z_{\rm UO}(2n)
\Big|_{x_{2n} = \b{z} \b{x}_{2n-1}}
=
(1-t)(1-tz)(1+z)
\prod_{i=1}^{2n-2}
(1-t x_i \b{x}_{2n-1}) (1-t \b{x}_i \b{x}_{2n-1})
(1-z^2 x_i x_{2n-1}) (1-z^2 \b{x}_i x_{2n-1})
\\
\times
(1-t z x_i x_{2n-1})  (1-t z \b{x}_i x_{2n-1})
(1-z x_i \b{x}_{2n-1}) (1-z \b{x}_i \b{x}_{2n-1})
Z_{\rm UO}(2n-2).
\end{multline*}

\item The smallest partition function $Z_{\rm UO}(2)$ is given explicitly by
\begin{align*}
Z_{\rm UO}(2)
=
(1-t) (1-tz) (1+z).
\end{align*}

\end{enumerate}

\begin{rmk}{\rm 
When $z=t^{-1/2}$, we are able to solve the above conditions explicitly:
\begin{align*}
Z_{\rm UO}(2n)
=
(-\sqrt{t})^n 
&\times
\prod_{1 \leq i<j \leq 2n}
\frac{(1-t^{\pm \frac{1}{2}} x_i x_j) (1- t^{\pm \frac{1}{2}} x_i \b{x}_j)
(1-t^{\pm \frac{1}{2}} \b{x}_i x_j) (1-t^{\pm \frac{1}{2}} \b{x}_i \b{x}_j)}
{(x_i-x_j)^2(1-\b{x}_i \b{x}_j)^2}
\\
&\times
\pf_{1\leq i<j \leq 2n} \left[
\frac{(1-t)(x_i-x_j)(1-\b{x}_i \b{x}_j)}
{(1-t^\frac{1}{2} x_i x_j) (1-t^\frac{1}{2} x_i \b{x}_j)
(1-t^\frac{1}{2} \b{x}_i x_j)(1-t^\frac{1}{2} \b{x}_i \b{x}_j)}
\right]
\\
&\times
\pf_{1\leq i<j \leq 2n} \left[
\frac{(1-t^{-1})(x_i-x_j)(1-\b{x}_i \b{x}_j)}
{(1-t^{-\frac{1}{2}} x_i x_j) (1-t^{-\frac{1}{2}} x_i \b{x}_j)
(1-t^{-\frac{1}{2}} \b{x}_i x_j) (1-t^{-\frac{1}{2}} \b{x}_i \b{x}_j)}
\right].
\end{align*}
The evaluation of this partition function as a product of Pfaffians is again due to Kuperberg \cite{kup2}. Similarly to the previous section, while the $z=t^{-1/2}$ specialization can be explicitly calculated, it causes the left hand side of \eqref{thm-uosasm} to diverge, meaning that the identity breaks down at this value of $z$.
}
\end{rmk}

\begin{rmk}{\rm 
The $t=0$ case of \eqref{thm-uosasm} leads to the explicit identity
\begin{multline*}
\sum_{
\lambda: m_i(\lambda) \in 2 \mathbb{N}
}
\
z^{|\lambda|/2}
sp_{\lambda}(x^{\pm1}_1,\dots,x^{\pm1}_{2n})
=
\\
\frac{z^{-n(n-1)}}
{\prod_{1 \leq i<j \leq 2n} (x_i - x_j) (1-\b{x}_i \b{x}_j)} 
\pf_{1\leq i< j \leq 2n}
\left[
\frac{(1+z) (x_i-x_j)(1-\b{x}_i \b{x}_j)}
{(1-z {x}_i {x}_j)(1-z {x}_i \b{x}_j)(1-z \b{x}_i {x}_j)(1- z \b{x}_i \b{x}_j)}
\right].
\end{multline*}
}
\end{rmk}

\appendix

\section{$F$ basis in models with six-vertex $R$ matrix}
\label{app:fbasis}

In this appendix we recall the use of factorizing $F$ matrices in quantum integrable models \cite{ms} based on the 
$\mathcal{R}$ matrix
\begin{align}
\label{renorm-R}
\mathcal{R}_{ab}(x/y)
=
\frac{1-x/y}{1-t x/y}
R_{ab}(x/y)
=
\begin{pmatrix}
1 & 0 & 0 & 0 \\
0 & \frac{t(1-x/y)}{1-t x/y} & \frac{(1-t)x/y}{1-t x/y} & 0 \\
0 & \frac{1-t}{1-t x/y} & \frac{1-x/y}{1-t x/y} & 0 \\
0 & 0 & 0 & 1 
\end{pmatrix}_{ab}.
\end{align}
The $\mathcal{R}$ matrix is a re-normalization of \eqref{Rmat} such that the unitarity relation $\mathcal{R}_{ab}(x/y) \mathcal{R}_{ba}(y/x) = 1$ holds. For each $1 \leq i \leq n$ let $W_i = \mathbb{C}^2$, and consider linear operators acting in $W_1 \otimes \cdots \otimes W_n$. For $\sigma$ a permutation of $(1,\dots,n)$, decompose it into simple transpositions and define $\mathcal{R}_{\sigma}^{1\dots n} \in {\rm End}(W_1 \otimes \cdots \otimes W_n)$ as the product of 
$\mathcal{R}$ matrices $\mathcal{R}_{ij}(x_i/x_j) \in {\rm End}(W_i \otimes W_j)$ corresponding with this decomposition. For example, for $n=6$, $\sigma=(4,5,3,2,1,6)$, one has
\begin{align*}
R^{1\dots 6}_{\sigma}
=
\begin{tikzpicture}[scale=0.5,baseline=(mid.base)]
\node (mid) at (1,4) {};
\draw[fline] plot [smooth] coordinates {(2,6.5) (2.5,6.5) (5.5,5.5) (7.5,3.5) (8.5,3.5)};
\draw[fline] plot [smooth] coordinates {(2,5.5) (2.5,5.5) (5.5,4.5) (7.5,2.5) (8.5,2.5)};
\draw[fline] plot [smooth] coordinates {(2,4.5) (2.5,4.5) (5.5,6.5) (7.5,4.5) (8.5,4.5)};
\draw[fline] plot [smooth] coordinates {(2,3.5) (2.5,3.5) (5.5,2.5) (7.5,5.5) (8.5,5.5)};
\draw[fline] plot [smooth] coordinates {(2,2.5) (2.5,2.5) (5.5,3.5) (7.5,6.5) (8.5,6.5)};
\draw[fline] (2,1.5) -- (8.5,1.5);
\node at (1.3,6.5) {$\ss \red \sigma(1)\ $}; 
\node at (1.3,5.5) {$\ss \red \sigma(2)\ $}; 
\node at (1.3,4.5) {$\ss \red \sigma(3)\ $};
\node at (1.3,3.5) {$\ss \red \sigma(4)\ $}; 
\node at (1.3,2.5) {$\ss \red \sigma(5)\ $};
\node at (1.3,1.5) {$\ss \red \sigma(6)\ $};
\node at (9,6.5) {$\ss \red 1$};
\node at (9,5.5) {$\ss \red 2$};
\node at (9,4.5) {$\ss \red 3$};
\node at (9,3.5) {$\ss \red 4$};
\node at (9,2.5) {$\ss \red 5$};
\node at (9,1.5) {$\ss \red 6$};
\end{tikzpicture}
=
R_{21}(x_2/x_1) R_{53}(x_5/x_3) R_{43}(x_4/x_3) R_{51}(x_5/x_1) R_{52}(x_5/x_2)
&
\\[-1cm]
\times\ \  
R_{41}(x_4/x_1) R_{42}(x_4/x_2) R_{31}(x_3/x_1) R_{32}(x_3/x_2)
&.
\\
\end{align*}
Note that, in view of the Yang--Baxter and unitarity equations, all possible decompositions of 
$\sigma$ into simple transpositions lead to the same expression for 
$\mathcal{R}_{\sigma}^{1\dots n}$. The $F$ matrix $F_{1\dots n} \in {\rm End}(W_1 \otimes \cdots \otimes W_n)$ satisfies the factorizing equation
\begin{align}
\label{factor-eqn}
F_{\sigma(1) \dots \sigma(n)} \mathcal{R}_{\sigma}^{1\dots n} = F_{1\dots n}
\end{align}
for all permutations $\sigma$, and following \cite{abfr} it can be expressed explicitly as\footnote{The original solution of equation \eqref{factor-eqn} was given in \cite{ms}, in a form quite different to \eqref{F}. Subsequently, a general solution of \eqref{factor-eqn}, for models with higher-rank $\mathcal{R}$ matrices, was obtained in \cite{abfr}. Here we quote the solution of \cite{abfr}, since it is better suited to our aims.}
\begin{align}
\label{F}
F_{1\dots n}
=
\sum_{i=0}^{n}
\left(
\prod_{1 \leq j \leq i}
E^{(22)}_{j}
\prod_{i < k \leq n}
E^{(11)}_{k}
\right)
+
\sum_{i=1}^{n}
\sum_{\substack{\rho \in S_n \\ \rho(i) > \rho(i+1)}}
\left(
\prod_{1 \leq j \leq i}
E^{(22)}_{\rho(j)}
\prod_{i < k \leq n}
E^{(11)}_{\rho(k)}
\right)
\mathcal{R}_{\rho}^{1\dots n}
\end{align}
where the final sum is over permutations $\rho$ with a single inversion 
$\rho(i) > \rho(i+1)$, and $E^{(kk)}_j$ is a $2 \times 2$ elementary matrix with entry 
$(k,k)$ equal to 1 and all other entries 0, acting in $W_j$. One of the fundamental properties of the $F$ matrix, which follows immediately from \eqref{F}, is that its components are given by
\begin{align*}
\Big[ F_{1\dots n} \Big]_{i_1\dots i_n}^{j_1 \dots j_n}
=
\Big[ \mathcal{R}_{\rho}^{1\dots n} \Big]_{i_1\dots i_n}^{j_1 \dots j_n}
\end{align*}
where $\rho$ is the unique permutation such that
\begin{align*}
i_{\rho(1)} \geq \cdots \geq i_{\rho(n)},
\qquad
\rho(k) < \rho(k+1)\ 
{\rm when}\
i_{\rho(k)} = i_{\rho(k+1)}.
\end{align*}
The $F$ matrix \eqref{F} can be shown to be upper triangular, with non-zero entries on its diagonal, and is thus invertible. Its inverse can be constructed via the $F^{*}$ matrix:
\begin{align*}
F^{*}_{1\dots n}
=
\sum_{i=0}^{n}
\left(
\prod_{1 \leq j \leq i}
E^{(11)}_{j}
\prod_{i < k \leq n}
E^{(22)}_{k}
\right)
+
\sum_{i=1}^{n}
\sum_{\substack{\rho \in S_n \\ \rho(i) > \rho(i+1)}}
\mathcal{R}_{1\dots n}^{\rho}
\left(
\prod_{1 \leq j \leq i}
E^{(11)}_{\rho(j)}
\prod_{i < k \leq n}
E^{(22)}_{\rho(k)}
\right)
\end{align*}
where the final sum is again over permutations $\rho$ with a single inversion $\rho(i) > \rho(i+1)$. It is possible to show that
\begin{align*}
F_{1\dots n}
F^{*}_{1\dots n}
=
\prod_{1 \leq k<l \leq n}
\Delta_{kl}(x_k,x_l),
\end{align*}
where $\Delta_{kl}(x_k,x_l) \in {\rm End}(W_k \otimes W_l)$ is a diagonal matrix with components
\begin{align*}
[\Delta_{kl}(x_k,x_l)]_{i_k i_l}^{j_k j_l}
=
\delta_{i_k,j_k} \delta_{i_l,j_l}
b_{i_k,i_l}(x_k,x_l),
\qquad
b_{i_k,i_l}(x_k,x_l)
=
\left\{
\begin{array}{ll}
1, & i_k = i_l
\\
(x_l - x_k)/(x_l - t x_k), & i_k > i_l
\\
(x_k - x_l)/(x_k - t x_l), & i_k < i_l.
\end{array}
\right.
\end{align*}
The inverse of the $F$ matrix is then given by
\begin{align}
\label{F-inv}
F^{-1}_{1\dots n}
=
F^{*}_{1\dots n}
\prod_{1 \leq k<l \leq n}
\Delta^{-1}_{kl}(x_k,x_l),
\end{align}
and in particular its components can be written in the form
\begin{align*}
\Big[
F^{-1}_{1\dots n}
\Big]_{i_1\dots i_n}^{j_1 \dots j_n}
=
\Big[
\mathcal{R}_{1\dots n}^{\rho}
\prod_{1 \leq k<l \leq n}
\Delta^{-1}_{kl}(x_k,x_l)
\Big]_{i_1\dots i_n}^{j_1 \dots j_n}
=
\Big[
\mathcal{R}_{1\dots n}^{\rho}
\Big]_{i_1\dots i_n}^{j_1 \dots j_n}
\prod_{1\leq k<l \leq n}
b^{-1}_{j_k,j_l}(x_k,x_l),
\end{align*}
where $\rho$ is the unique permutation such that
\begin{align*}
j_{\rho(1)} \leq \cdots \leq j_{\rho(n)},
\qquad
\rho(k) < \rho(k+1)\ 
{\rm when}\
j_{\rho(k)} = j_{\rho(k+1)}.
\end{align*}
The main feature of factorizing $F$ matrices is that they can be used to effect a change of basis, in such a way that certain linear operators acting in $W_1 \otimes \cdots \otimes W_n$ become completely symmetric over the spaces $W_i$. To illustrate this, consider an operator $\mathcal{O}_{1\dots n}(x_1,\dots,x_n) \in {\rm End}(W_1\otimes \cdots \otimes W_n)$ which satisfies
\begin{align*}
\mathcal{R}_{i+1,i}(x_{i+1}/x_i)
\mathcal{O}_{1\dots i,i+1\dots n}(x_1,\dots,x_i,x_{i+1},\dots,x_n)
=
\mathcal{O}_{1\dots i+1,i \dots n}(x_1,\dots,x_{i+1},x_i,\dots,x_n)
\mathcal{R}_{i+1,i}(x_{i+1}/x_i)
\end{align*}
for all $1 \leq i \leq n-1$. Then more generally, 
\begin{align*}
\mathcal{R}_{\sigma}^{1\dots n}
\mathcal{O}_{1\dots n}(x_1,\dots,x_n)
=
\mathcal{O}_{\sigma(1) \dots \sigma(n)}\left(x_{\sigma(1)},\dots,x_{\sigma(n)}\right)
\mathcal{R}_{\sigma}^{1\dots n},
\end{align*}
and using the factorizing equation \eqref{factor-eqn} we find that
\begin{align*}
F_{1\dots n} \mathcal{O}_{1\dots n}(x_1,\dots,x_n) F_{1\dots n}^{-1}
=
F_{\sigma(1) \dots \sigma(n)} 
\mathcal{O}_{\sigma(1)\dots \sigma(n)} \left(x_{\sigma(1)},\dots,x_{\sigma(n)} \right) 
F_{\sigma(1) \dots \sigma(n)}^{-1}.
\end{align*}
Hence the conjugated operator $\widetilde{O}_{1\dots n}(x_1,\dots,x_n) := 
F_{1\dots n} \mathcal{O}_{1\dots n}(x_1,\dots,x_n) F_{1\dots n}^{-1}$ (usually called a ``twisted'' operator in the literature, after {\it Drinfeld twists}) satisfies
\begin{align}
\label{symmetry-O}
\widetilde{O}_{1\dots n}(x_1,\dots,x_n)
=
\widetilde{O}_{\sigma(1) \dots \sigma(n)}\left(x_{\sigma(1)},\dots,x_{\sigma(n)}\right).
\end{align}
The symmetry property \eqref{symmetry-O} is the main reason to study the change of basis induced by the $F$ matrix, since many operators can be written explicitly in the new basis.

\section{Lagrange interpolation}
\label{app:lagrange}

In this paper we evaluate many partition functions by listing a set of properties that they obey, then showing that a relevant determinant or Pfaffian satisfies the same list of properties. It remains only to show that such a list of properties uniquely determines the object in question, which we do with the following lemma.

\begin{lem}
Consider a family of polynomials $f_m(x_1,\dots,x_m)$, $m \geq 1$, which are symmetric in $(x_1,\dots,x_m)$ and degree $d_m$ in any one of these variables. Suppose that 
\begin{align*}
f_m(x_1,\dots,x_{m-1},z_i) 
&= 
C^{(i)}_m
f_{m-1}(x_1,\dots,x_{m-1}),
\quad
1 \leq i \leq d_m,
\end{align*}
for a suitable set of points $\left\{ z_1,\dots,z_{d_m} \right\}$ and coefficients $C^{(i)}_m$, and furthermore that $f_m(0,\dots,0) = C_m^{(0)}$. Then if another family of polynomials $g_m(x_1,\dots,x_m)$, $m \geq 1$ satisfies all of these conditions and $f_1(x) = g_1(x)$, then necessarily
\begin{align*}
f_m(x_1,\dots,x_m) = g_m(x_1,\dots,x_m), 
\quad
\text{for all}\ m \geq 1.
\end{align*}
\end{lem}

\begin{proof}
Since $f_1(x) = g_1(x)$, assume that 
$f_{n-1}(x_1,\dots,x_{n-1}) = g_{n-1}(x_1,\dots,x_{n-1})$ for some $n \geq 2$. It follows that
\begin{align*}
f_n(x_1,\dots,x_{n-1},z_i) = g_n(x_1,\dots,x_{n-1},z_i),
\quad
1 \leq i \leq d_n,
\end{align*}
meaning that $f_n$ and $g_n$ are equal at as many points as their degree. Then clearly
\begin{align*}
f_n(x_1,\dots,x_n)
=
\kappa_n
g_n(x_1,\dots,x_n)
\end{align*}
where $\kappa_n$ does not depend on $(x_1,\dots,x_n)$ (if it did depend on the variables, it would have to be symmetric in them, but this would then violate the condition that both $f_n$ and $g_n$ are degree $d_n$ polynomials in $x_n$). Finally, using the fact that $f_n(0,\dots,0) = C^{(0}_n = g_n(0,\dots,0)$, we find that $\kappa_n=1$. This proves that 
$f_{n}(x_1,\dots,x_{n}) = g_{n}(x_1,\dots,x_{n})$ for all $n \geq 1$, by induction on $n$.
\end{proof}



\bibliographystyle{abbrv}
\bibliography{references}

\end{document}